\documentclass[11pt]{article}
\usepackage{amsfonts,amsthm,amssymb}
\usepackage[fleqn]{amsmath}
\usepackage{hyperref}
\usepackage[pdftex]{graphicx}
\usepackage{color}
\usepackage[leftcaption]{sidecap}
\newcommand{\be}{\begin{eqnarray}}
\newcommand{\ee}{\end{eqnarray}}
\newcommand{\bez}{\begin{eqnarray*}}
\newcommand{\eez}{\end{eqnarray*}}

\newcommand{\pa}{\partial}
\newcommand{\cA}{\mathcal{A}}
\newcommand{\cB}{\mathcal{B}}

\newcommand{\bA}{\boldsymbol{A}}
\newcommand{\balpha}{\boldsymbol{\alpha}}
\newcommand{\bB}{\boldsymbol{B}}
\newcommand{\bbeta}{\boldsymbol{\beta}}
\newcommand{\bC}{\boldsymbol{C}}
\newcommand{\bF}{\boldsymbol{F}}
\newcommand{\bI}{\boldsymbol{I}}
\newcommand{\bJ}{\boldsymbol{J}}
\newcommand{\bLa}{\boldsymbol{\Lambda}}
\newcommand{\bP}{\boldsymbol{P}}
\newcommand{\bPhi}{\boldsymbol{\Phi}}
\newcommand{\bPsi}{\boldsymbol{\Psi}}
\newcommand{\bQ}{\boldsymbol{Q}}
\newcommand{\bTh}{\boldsymbol{\Theta}}
\newcommand{\bU}{\boldsymbol{U}}
\newcommand{\bV}{\boldsymbol{V}}
\newcommand{\bX}{\boldsymbol{X}}

\newcommand{\bbC}{\mathbb{C}}
\newcommand{\bbE}{\mathbb{E}}
\newcommand{\bbN}{\mathbb{N}}
\newcommand{\bbR}{\mathbb{R}}
\newcommand{\bbS}{\mathbb{S}}
\newcommand{\bbZ}{\mathbb{Z}}

\newcommand{\tbP}{\tilde{\boldsymbol{P}}}
\newcommand{\tbQ}{\tilde{\boldsymbol{Q}}}
\newcommand{\tbX}{\tilde{\boldsymbol{X}}}

\newcommand{\bd}{\bar{\mathrm{d}}}

\newcommand{\imag}{\mathrm{i}}
\renewcommand{\Re}{\mathrm{Re}}
\renewcommand{\Im}{\mathrm{Im}}

\numberwithin{equation}{section}

\theoremstyle{plain}
\newtheorem{Theorem}{Theorem}[section]

\newtheorem{Proposition}[Theorem]{Proposition}
\newtheorem{corollary}[Theorem]{Corollary}

\theoremstyle{definition}

\newtheorem{remark}[Theorem]{Remark}
\newtheorem{example}[Theorem]{Example}

\date{}

\begin{document}

\title{\huge \bf `Riemann Equations' \\ in Bidifferential Calculus}

\author{ 
\sc{O. Chvartatskyi}$^{a,b}$, 
\sc{F. M\"uller-Hoissen}$^a$ 
and \sc{N. Stoilov}$^{a,b}$ \\ 
 $^a$ \small Max Planck Institute for Dynamics and Self-Organization, 
         37077 G\"ottingen, Germany \\
 \small e-mail: alex.chvartatskyy@gmail.com, 
 folkert.mueller-hoissen@ds.mpg.de  \\
 $^b$ \small Mathematisches Institut,
 Georg-August-Universit\"at G\"ottingen, 37073 G\"ottingen, Germany \\ 
 \small e-mail: nstoilo@uni-math.gwdg.de  
}

\maketitle

\begin{abstract} 
We consider equations that formally resemble a matrix Riemann (or Hopf) equation in the framework of 
bidifferential calculus. With different choices of a first-order bidifferential calculus, 
we obtain a variety of equations, including a semi-discrete and a fully discrete version of the matrix 
Riemann equation. A corresponding universal solution-generating method then either yields a 
(continuous or discrete) Cole-Hopf transformation, or leaves us with the problem of solving Riemann 
equations (hence an application of the hodograph method). 
If the bidifferential calculus extends to second order, 
solutions of a system of `Riemann equations' are also solutions of an equation that arises, on 
the universal level of bidifferential calculus, as an integrability condition. 
Depending on the choice of bidifferential calculus, the latter can represent a 
number of prominent integrable equations, like self-dual Yang-Mills, as well as matrix versions of 
the two-dimensional Toda lattice, Hirota's bilinear difference equation, (2+1)-dimensional NLS,
KP and Davey-Stewartson equations. For all of them, a recent (non-isospectral) binary Darboux transformation 
result in bidifferential calculus applies, which can be specialized to generate solutions of the 
associated `Riemann equations'. For the latter, we clarify the relation between these specialized 
binary Darboux transformations and the aforementioned solution-generating method. 
From (arbitrary size) matrix versions  of the `Riemann equations' associated with an integrable 
equation, possessing a bidifferential calculus formulation, 
multi-soliton-type solutions of the latter can be generated. This includes `breaking'  
multi-soliton-type solutions of the self-dual Yang-Mills and the (2+1)-dimensional NLS equation, which 
are parametrized by solutions of Riemann equations. 
\end{abstract}

\noindent
\begin{center}
\small
\parbox{15cm}{
 }
\end{center}

\section{Introduction}
Given an associative algebra $\cA$ and two derivations $\mathrm{d} , \bd : \cA \rightarrow \Omega^1$ 
into an $\cA$-bimodule, the two equations 
\be
     \bd \phi - (\mathrm{d} \phi) \, \phi = 0     \label{Riemann_1}
\ee
and
\be
     \bd \phi + \phi \, \mathrm{d} \phi = 0 \, ,    \label{Riemann_2}
\ee
for $\phi \in \cA$, resemble a \emph{Riemann equation} (also known as Hopf, 
inviscid Burgers, dispersionless KdV, or nonlinear transport equation). For a simple choice 
of the \emph{first-order bidifferential calculus}, given by 
$\cA, \Omega^1,\mathrm{d},\bd$, these are indeed matrix Riemann 
equations (see in particular 
\cite{LRB83,Bogo92V,Fera93,Fera94,GMA94,Jose+Vasu01,Arri+Hick02,Sant+Zenc07,Zenc07a,Zenc08sdYM}
for appearances in the literature). 
For other choices, (\ref{Riemann_1}) and (\ref{Riemann_2}) turn out to be very different 
equations, however. This is so because we allow $\phi$ to be  
an operator (e.g., a differential or difference operator), a familiar step   
in the theory of integrable systems, where one considers a `zero curvature' (Zakharov-Shabat, 
or Lax) equation in general for operator expressions. 
Accordingly, $\cA$ will then be an algebra involving operators. 
Several examples will be presented in this work. The most basic ones are semi and a full 
discretizations of matrix Riemann equations. The following 
table shows that they are obtained from the bidifferential calculus formulation for the continuous 
Riemann equation essentially by replacing operators $\pa_x,\pa_t$ of taking the partial derivatives 
with respect to $x$, respectively $t$, by shift operators $\bbS_0$, respectively $\bbS_1$, acting on 
corresponding discrete variables. 
\begin{center}
\def\arraystretch{1.5}
\begin{tabular}[t]{l|l|l|l}
   $\mathrm{d}$ & $\bd$ & & $\bd \phi - (\mathrm{d} \phi) \, \phi = 0$ \\
   \hline
   $[\pa_x , \cdot ]$ &  $[\pa_t , \cdot ]$ & $\varphi = \phi$ & $\varphi_t - \varphi_x \, \varphi = 0$ \\
   \hline
   $[\bbS_0 , \cdot ]$ &  $[\pa_t , \cdot ]$ & $\varphi = \phi \, \bbS_0$ 
         & $\varphi_t - (\varphi_{,0} - \varphi) \, \varphi = 0$ \\
   \hline
   $ [ \bbS_1^{-1} \bbS_0 , \cdot]$ & $[ \bbS_1^{-1} , \cdot]$ & $\varphi = \phi \, \bbS_0$ 
         & $\varphi_{,1} - \varphi - (\varphi_{,1} - \varphi_{,0}) \, \varphi =0 $ 
\end{tabular}
\end{center}
In these examples, $\cA$ is the algebra of $m \times m$ matrices of functions, extended by shift operators 
in the last two cases.
$\Omega^1$ is simply given by $\cA$, $\varphi$ is a matrix of functions, 
and we set $\varphi_{,0} := \bbS_0 \varphi \bbS_0^{-1}$, $\varphi_{,1} := \bbS_1 \varphi \bbS_1^{-1}$. 
Since $\mathrm{d}$ and $\bd$ are given by commutators, they are obviously derivations. 
In all examples considered in this work, (\ref{Riemann_1}) and (\ref{Riemann_2}) are related by 
a transpose or adjoint operation. It is therefore sufficient to concentrate on (\ref{Riemann_1}).

Suppose there is an extension of the derivation $\mathrm{d}$ to a map 
$\cA \stackrel{\mathrm{d}}{\rightarrow} \Omega^1 \stackrel{\mathrm{d}}{\rightarrow} \Omega^2$, with another 
$\cA$-bimodule $\Omega^2$, and correspondingly for $\bd$, such that 
\be
    \mathrm{d}^2 = \bd^2 = \mathrm{d} \bd + \bd \mathrm{d} = 0  \, .   \label{bidiff_conds}
\ee
In this case we have a \emph{second-order bidifferential calculus},  
$(\Omega,\mathrm{d},\bd)$, with $\Omega = \bigoplus_{r=0}^2 \Omega^r$, $\Omega^0 := \cA$ 
\cite{DMH00a,DMH08bidiff}. 
Then, acting with $\mathrm{d}$ on (\ref{Riemann_1}) or (\ref{Riemann_2}) yields 
\be
     \mathrm{d} \bd \phi + \mathrm{d} \phi \, \mathrm{d} \phi = 0 \, .   \label{phi_eq}
\ee
By choosing appropriate bidifferential calculi, this equation 
leads to various integrable partial differential and/or difference equations (PDDEs) 
(cf. \cite{DMH08bidiff} and references therein). A particular example is another 
semi-discretization of the Riemann equation, the (Lotka-) Volterra lattice equation (see Appendix~\ref{app:LV}), 
which is \emph{not} obtained from (\ref{Riemann_1}).
Equations like (\ref{Riemann_1}), (\ref{Riemann_2}) and (\ref{phi_eq}) are of a universal nature 
and integrable PDDEs, derived from them, can be thought of as realizations.  
Equation (\ref{phi_eq}) originally arose from replacing $\mathrm{d}$ and $\bd$ by flat 
anticommuting `covariant derivatives', as explained, e.g., in \cite{DMH08bidiff}. The use of a calculus 
similar to the calculus of differential forms on a manifold reduces otherwise lengthy and often rather 
intransparent computations to a few lines, simply by exploiting the Leibniz rule for $\mathrm{d}$ and $\bd$, 
and (\ref{bidiff_conds}). Why are there \emph{two} maps, $\mathrm{d}$ and $\bd$, instead of a single  
analog of the familiar exterior derivative? In this way the integrable structure underlying 
the self-dual Yang-Mills equation is expressed most concisely \cite{DMH08bidiff} and drastically generalized. 
Moreover, it generalizes the situation in Fr\"olicher-Nijenhuis theory (cf. Remark 2.5 in \cite{DMH13SIGMA}). 

Some efficient solution generating methods can be easily derived for the 
universal equations. By choosing a bidifferential calculus in such a way that one of these 
equations becomes equivalent to some PDDE, the method applies to the latter, and in this way 
one typically obtains a method for that equation. 
In particular, this shows that solution generating methods 
for various equations have a surprisingly simple origin and a simple universal proof. 

For some realizations of (\ref{Riemann_1}) (and (\ref{Riemann_2})), there is a simple 
`linearization method' (see Section~\ref{sec:CH}), which in several cases is the origin of a Cole-Hopf-type 
transformation. Such realizations are in the class of `C-integrable equations' (see, e.g., 
\cite{Calo+Eckh87,Calo+Xiao92JMP,Calo92JMP}). Matrix Burgers equations are the prototype examples. 
The method is ineffective for the Riemann equation, in which case the method of 
characteristics, or hodograph method, applies instead (cf. \cite{Sant+Zenc07}). 

A Cole-Hopf-type transformation does not extend to (\ref{phi_eq}), for which, however, there is  
another universal method. Indeed, in \cite{DMH13SIGMA} (also see Section~\ref{sec:bDT}), 
a solution-generating result representing an abstract version of binary 
Darboux transformations \cite{Matv+Sall91,Matv00} has been derived for  
(\ref{phi_eq}) and the `(Miura-) dual' equation 
\be
    \mathrm{d} [ (\bd g) \, g^{-1} ] = 0 \, ,  \label{g_eq}
\ee
with (invertible) dependent variable $g \in \cA$. 
More precisely, this is a solution generating result for the `Miura equation'
\be
     (\bd g) \, g^{-1} = \mathrm{d} \phi \, ,  \label{Miura} 
\ee
which has both equations, (\ref{phi_eq}) and (\ref{g_eq}), as integrability conditions, 
provided that (\ref{bidiff_conds}) holds. This binary Darboux transformation method 
requires solutions of versions of (\ref{Riemann_1}) and (\ref{Riemann_2}) as inputs (see (\ref{P,Q_eqs})), 
which yields yet another motivation to explore these `Riemann equations'. In most cases, soliton 
families are obtained by choosing $\mathrm{d}$- and $\bd$-constant solutions of these equations, which 
are very special and somewhat trivial solutions. 
The non-autonomous chiral model equation that arises in integrable reductions of the vacuum Einstein
(-Maxwell) equations is an important exception in this respect, see \cite{DKMH11sigma,DMH13SIGMA} 
and also Section~\ref{subsubsec:naCM}. More generally, this concerns equations possessing a non-isospectral 
linear problem (see \cite{Maison78,BZM87} and also, e.g., \cite{Bogo90RMS,CGP97IP}). 

Furthermore, the present work partly originated from the simple observation that (\ref{Miura}) becomes 
(\ref{Riemann_1}) (respectively (\ref{Riemann_2})), if we set $g = \pm \phi$ (respectively 
$g^{-1} = \pm \phi$). The solution-generating result in \cite{DMH13SIGMA} then 
still works and can indeed be applied to generate large classes of exact solutions 
of various realizations of (\ref{Riemann_1}) (respectively (\ref{Riemann_2})). According to the implication
\be
    \bd \phi - (\mathrm{d} \phi) \, \phi = 0  \quad \Longrightarrow \quad
    \left \{ \begin{array}{l} \mathrm{d} \bd \phi + (\mathrm{d} \phi) \, \mathrm{d} \phi = 0 \, , \\
             \mathrm{d} [ (\bd \phi) \, \phi^{-1} ] = 0 \, ,
             \end{array} \right.      \label{Riem_to_phi_eq}
\ee
assuming that the first order bidifferential calculus extends to second order, the system 
of equations given by a realization of (\ref{Riemann_1}) provides us with a special class 
of solutions of the associated realizations of (\ref{phi_eq}) and (\ref{g_eq}). It is one of the main 
aims of this work to explore what this `Riemann system' is for several integrable equations and 
what the corresponding class of solutions contains 
(see Sections~\ref{sec:Riemann_associates}-\ref{sec:DS}).  
 For example, the `Riemann system' associated with the 
(matrix) KP equation consists of the first two members of a (matrix) Burgers hierarchy 
(see Section~\ref{sec:B&KP}), so in this case the implication 
(\ref{Riem_to_phi_eq}), with the upper equation on the right, expresses a well-known fact 
(cf. \cite{DMH09Sigma} and references cited there). 
The `Riemann system' associated with the self-dual 
Yang-Mills equation consists of two matrix Riemann equations (see Section~\ref{subsec:sdYM}). 
Here we recover an observation made in \cite{Zenc08sdYM}. In the case of the integrable 
non-autonomous chiral model underlying integrable reductions of the Einstein vacuum equations, 
the `Riemann system' determines a matrix version of the pole trajectory equation of Belinski and 
Zakharov \cite{Beli+Zakh78} (see Section~\ref{subsubsec:naCM}). We also explore hitherto unknown 
`Riemann systems' associated with several other integrable equations possessing a bidifferential calculus 
formulation. 

In those cases where a `Riemann system' admits a Cole-Hopf-type transformation, this is 
certainly the simplest way to generate exact solutions. Darboux transformations for PDDEs resulting 
from (\ref{Riemann_1}) (or (\ref{Riemann_2})) are in this respect not the first choice, 
but may be helpful, depending on the addressed problem. We are particularly interested in 
understanding how the two methods are related. 
Of course, more general classes of solutions of (\ref{phi_eq}) 
and (\ref{g_eq}) than those obtained from the associated `Riemann system' can be generated 
using the binary Darboux transformation method (Theorem~\ref{thm:main}) 
and we will report more on this in a separate work. 

Another aspect addressed in the present work concerns the way in which bidifferential 
calculi for many integrable equations are composed of those for `Riemann equations'. 

The paper is organized as follows. 
Section~\ref{sec:CH} presents the aforementioned simple solution-generating method for (\ref{Riemann_1}).  
In Section~\ref{sec:bDT} we recall from \cite{DMH13SIGMA} the binary Darboux transformation theorem 
in bidifferential calculus, in a slightly generalized form, and specialize it in a corollary in 
the way described above. A simple, but crucial observation is that the theorem already works 
on the level of a first order bidifferential calculus and then applies to `Riemann equations'. 

Section~\ref{sec:Riemann} treats matrix Riemann 
equations, their integrable (semi and full) discretizations, and corresponding hierarchies. 
In the semi-discrete case, this is the semi-discrete Burgers hierarchy first treated 
in \cite{LRB83}. The (fully) discrete Riemann hierarchy contains a matrix version of the 
discrete Burgers equation derived in \cite{HLW99}. We are not aware of previous explorations of its 
first member, an integrable discrete Riemann equation. Furthermore, the Darboux transformations 
derived for the matrix versions of the semi- and fully discrete Riemann equations are new to the 
best of our knowledge. 

Sections~\ref{sec:Riemann_associates}-\ref{sec:DS} present a collection of important 
examples of (matrix versions of) integrable equations arising via (\ref{Riem_to_phi_eq}) 
from a system of `Riemann equations'. This includes a generalization of Hirota's bilinear 
difference equation, which we have not seen in the literature yet. In Section~\ref{sec:NLS}, we 
consider the (2+1)-dimensional Nonlinear Schr\"odinger (NLS) equation \cite{Calo+Dega76NCB32,Zakh80,Stra92NLS}. 
In particular, we show that `breaking solitons' obtained in \cite{Bogo91III} are solutions of the associated 
`Riemann system'. Furthermore, we obtain matrix versions of these solutions and moreover multi-soliton  
solutions, which are parametrized by solutions of a `Riemann system' 
(see Proposition~\ref{prop:2+1NLS_multi-breaking_solitons}). 
Section~\ref{sec:B&KP} presents the relation between the first two equations of a (matrix) Burgers 
hierarchy and the (matrix) KP equation as a special case of (\ref{Riem_to_phi_eq}). Section~\ref{sec:DS} 
treats the Davey-Stewartson (DS) equation. In the scalar case, single dromion, soliton and solitoff 
solutions turn out to be solutions of the associated Riemann system. 
Section~\ref{sec:conclusion} contains some concluding remarks.

As above, also in the following \emph{Riemann equation}, respectively \emph{Riemann system}, without an adjective  
\emph{discrete} or \emph{semi-discrete}, will always refer to the familiar partial differential 
equation, respectively a system of such equations. In contrast, \emph{`Riemann equation'}, respectively 
\emph{`Riemann system'}, will refer to any equation, respectively a system of equations, that realizes  
(\ref{Riemann_1}) (which only in special cases becomes a Riemann equation or a Riemann system).

\section{The linearization method}
\label{sec:CH}
Writing
\be
     \phi = \Phi \, \phi_0 \, \Phi^{-1} \, ,   \label{phi_transform}
\ee
with an invertible $\Phi \in \cA$, (\ref{Riemann_1}) is equivalent to
\be
    \bd \phi_0 - (\mathrm{d} \phi_0) \, \phi_0 + [ \gamma \, , \,  \phi_0 ] = 0 \, , 
           \label{Riemann_1_mod}
\ee
where $\gamma \in \Omega^1$ is defined by 
\be
     \bd \Phi - (\mathrm{d} \Phi) \, \phi_0 - \Phi \, \gamma  = 0 \, .  \label{CH_Phi_eq}
\ee
Let us, however, consider the last equation as an equation for $\Phi$, for a given $\gamma$. 
Then, if $\phi_0$ is a solution of (\ref{Riemann_1_mod}), $\phi$ given by 
(\ref{phi_transform}) solves (\ref{Riemann_1}). 
For fixed $\phi_0$, (\ref{phi_transform}) 
and (\ref{CH_Phi_eq}), written as linear equations for $\Phi$, can thus be regarded as a 
Lax pair for (\ref{Riemann_1}).  

Here we only have to solve \emph{linear} equations in order to construct new solutions. 
Obviously, (\ref{phi_transform}) can only lead to a new solution if the solution $\Phi$ 
of the linear equation does \emph{not} commute with $\phi_0$. 
This excludes the example of the scalar Riemann equation (see 
Section~\ref{subsec:Riem}), but non-trivial solutions of \emph{matrix} 
Riemann equations can be obtained (also see \cite{Sant+Zenc07}).  
Since $\phi_0$ may involve an operator, (\ref{phi_transform}) can also express a Cole-Hopf-type 
transformation. This includes the well-known Cole-Hopf 
transformation for the Burgers equation, see Section~\ref{sec:B&KP}. 

\begin{remark}
(\ref{Riemann_1_mod}) looks more general than (\ref{Riemann_1}). But 
if $\gamma$ does not depend on $\phi_0$, as assumed above, the $\gamma$ term in 
(\ref{Riemann_1_mod}) can be absorbed via the redefinition $\bd \mapsto \bd - [\gamma, \cdot]$. 
Furthermore, replacing (\ref{CH_Phi_eq}) by $\bd \Phi - (\mathrm{d} \Phi) \, \phi_0 + [\gamma , \Phi ]  = 0$, 
then $\phi$ given by (\ref{phi_transform}) also satisfies (\ref{Riemann_1_mod}). 
\end{remark}

\begin{remark}
Instead of (\ref{Riemann_1}), we may consider the more general equation
\be
     \bd \phi - (\mathrm{d} \phi) \, \eta(\phi) = \rho(\phi)  \, ,    \label{gRiemann_1}
\ee
where the function $\eta$ and $\rho \in \Omega^1$ are required to satisfy
$\eta(\phi)=\Phi \, \eta(\phi_0) \, \Phi^{-1}$, $\rho(\phi)= \Phi \, \rho(\phi_0) \, \Phi^{-1}$. 
(\ref{gRiemann_1}) is then equivalent to
\bez
    \bd \phi_0 - (\mathrm{d} \phi_0) \, \eta(\phi_0) + [ \gamma \, , \,  \phi_0 ] = \rho(\phi_0) \, ,
\eez
where $\gamma \in \Omega^1$ is defined by 
\bez
     \bd \Phi - (\mathrm{d} \Phi) \, \eta(\phi_0) - \Phi \, \gamma  = 0 \, .  
\eez
Considering these as equations for a given $\gamma$, everything works well as long as $\phi$ 
has values in an algebra of matrices of functions, which is then a case treated in \cite{Sant+Zenc07}. 
If the algebra $\cA$ contains operators and $\phi$ is an operator expression, it will typically 
be impossible to reduce 
(\ref{gRiemann_1}) to a PDDE. There are exceptions when $\eta$ is homogeneous. 
However, in those cases we looked at, it turned out that they can also be treated 
starting from (\ref{Riemann_1}), see Remarks~\ref{rem:r-sdRiemann} and \ref{rem:r-dRiemann}.
In \cite{Sant+Zenc07} also generalizations of matrix Riemann equations to any number 
of independent variables are treated. In principle, our formalism can incorporate this via a straight
generalization of (\ref{gRiemann_1}) and extending $\mathrm{d}$ to several commuting derivations 
$\mathrm{d}_i: \cA \rightarrow \cA$, $i=1,\ldots,N$. However, we have not been able to find a PDDE 
arising in this way outside of the (continuous) framework of \cite{Sant+Zenc07}. 
\end{remark}

\section{Binary Darboux transformations in bidifferential calculus}
\label{sec:bDT}
\setcounter{equation}{0}
In the following, let $\cA$ be the algebra of all finite-dimensional matrices with entries 
in a unital associative algebra $\cB$, where the product of two matrices is defined to be zero if the 
sizes of the two matrices do not match. We assume that there is an $\cA$-bimodule
$\Omega^1$ and derivations $\mathrm{d}$ and $\bd$ on $\cA$ with values in $\Omega^1$,  
such that $\mathrm{d}$ and $\bd$ preserve the size of matrices. $\mathrm{Mat}(m,n,\cB)$ denotes 
the set of $m \times n$ matrices over $\cB$. For fixed $m,n \in \bbN$, 
$I=I_m$ and $\bI=I_n$ denote the $m \times m$, respectively $n \times n$, identity matrix, 
and we assume that they are constant with respect to $\mathrm{d}$ and $\bd$. 
Let us recall the main theorem in \cite{DMH13SIGMA}, in a slightly generalized form. It should be noticed, 
however, that here we do \emph{not} require the conditions (\ref{bidiff_conds}). All we need is that 
$\mathrm{d}$ and $\bd$ are derivations on $\cA$, i.e., they satisfy the Leibniz rule.

\begin{Theorem}
\label{thm:main}
Let $\phi_0,g_0 \in \mathrm{Mat}(m,m,\cB)$ solve (\ref{Miura}), and let 
$\bP,\bQ \in \mathrm{Mat}(n,n,\cB)$ be solutions of 
\be
   \bd \bP - (\mathrm{d} \bP) \, \bP = - [ \balpha  \, , \, \bP ] \, , \qquad
   \bd \bQ - \bQ \, \mathrm{d} \bQ = [ \bbeta \, , \, \bQ ] \, ,   \label{P,Q_eqs}
\ee
with some $\balpha, \bbeta \in \Omega^1$.
Let $\bU \in \mathrm{Mat}(m,n,\cB)$ and $\bV \in \mathrm{Mat}(n,m,\cB)$ be solutions
of the linear equations 
\be
    \bd \bU = (\mathrm{d} \bU) \, \bP + (\mathrm{d} \phi_0) \, \bU + \bU \, \balpha\, , \qquad 
    \bd \bV = \bQ \, \mathrm{d} \bV - \bV \, \mathrm{d} \phi_0 + \bbeta \, \bV \, .    \label{U,V_eqs}
\ee 
Furthermore, let $\bX \in \mathrm{Mat}(n,n,\cB)$ be an invertible solution of the 
(inhomogeneous) linear equations 
\be
   && \bX \, \bP - \bQ \, \bX = \bV \, \bU \, ,  \label{preSylvester}  \\
   && \bd \bX - (\mathrm{d} \bX) \, \bP + (\mathrm{d} \bQ) \, \bX + (\mathrm{d} \bV) \, \bU 
      = \bX \, \balpha + \bbeta \, \bX \, . \label{dX_eq}
\ee
Then 
\be
    \phi = \phi_0 + \bU \bX^{-1} \bV   \qquad \mbox{and} \qquad
    g = (I + \bU (\bQ \, \bX)^{-1} \bV) \, g_0                   \label{thm_new_Miura_sol}
\ee
yields a new solution of the Miura equation (\ref{Miura}).    \hfill $\square$
\end{Theorem}

\begin{remark}
For $\balpha = \bbeta =0$, the two equations in (\ref{P,Q_eqs}) are $n \times n$ 
matrix versions of (\ref{Riemann_1}) and (\ref{Riemann_2}), respectively.
Unless $\bP$ and $\bQ$ are $\mathrm{d}$- and $\bd$-constant,  
the linear equations (\ref{U,V_eqs}) constitute a \emph{non-isospectral problem} 
(cf. \cite{Maison78,BZM87} and also, e.g., \cite{Bogo90RMS,CGP97IP}). $\bP$ and $\bQ$ may be 
regarded as operator (in particular, matrix) versions of a spectral parameter. 
\end{remark}

\begin{remark}
\label{rem:transformation_to_alpha,beta}
Theorem~\ref{thm:main} appeared in \cite{DMH13SIGMA} with $\balpha = \bbeta = 0$. 
Let us start with this version. Correspondingly, in this remark a reference to an equation in 
Theorem~\ref{thm:main} shall mean the respective equation with $\balpha = \bbeta =0$.
The introduction of $\balpha$ and $\bbeta$ is motivated by a freedom 
of transformations.  Let us write
\be
     \bP = \bPsi_1 \tilde{\bP} \bPsi_1^{-1} \, , \qquad
     \bQ = \bPsi_2^{-1} \tilde{\bQ} \bPsi_2 \, ,    \label{P,Q_transform}
\ee
with invertible $n \times n$ matrices $\bPsi_1, \bPsi_2$. 
(\ref{P,Q_eqs}) then takes the form
\bez
 \bd \tilde{\bP} - (\mathrm{d} \tilde{\bP}) \, \tilde{\bP} 
   = - [ \balpha  \, , \, \tilde{\bP} ] \, , \qquad
 \bd \tilde{\bQ} - \tilde{\bQ} \, \mathrm{d} \tilde{\bQ} 
   = [ \bbeta \, , \, \tilde{\bQ} ] \, ,
\eez
where $\balpha, \bbeta$ are now \emph{defined} by 
\bez
    \bd \bPsi_1 - (\mathrm{d} \bPsi_1)  \, \tilde{\bP}
  = \bPsi_1 \, \balpha \, ,  \qquad
    \bd \bPsi_2 - \tilde{\bQ} \, \mathrm{d} \bPsi_2 
  = \bbeta \, \bPsi_2 \, .
\eez
In terms of $\tbX := \bPsi_2 \bX \bPsi_1$, $\tilde{\bU} := \bU \bPsi_1$ and $\tilde{\bV} := \bPsi_2 \bV$,
(\ref{preSylvester}) and (\ref{dX_eq}) are equivalent to
\bez
 \tbX \, \tilde{\bP} - \tilde{\bQ} \, \tbX = \tilde{\bV} \, \tilde{\bU} \, , \qquad
  \bd \tbX - (\mathrm{d} \tbX) \, \tilde{\bP} + (\mathrm{d} \tilde{\bQ}) \, \tbX 
      + (\mathrm{d} \tilde{\bV}) \, \tilde{\bU} 
  = \tbX \, \balpha + \bbeta \, \tbX \, ,
\eez
where we used the linear equations for $\bPsi_1$ and $\bPsi_2$. 
The linear equations (\ref{U,V_eqs}) are correspondingly transformed to
\bez
   \bd \tilde{\bU} = (\mathrm{d} \tilde{\bU} ) \, \tilde{\bP} + (\mathrm{d} \phi_0) \, \tilde{\bU} 
                     + \tilde{\bU} \, \balpha \, , \qquad
   \bd \tilde{\bV} = \tilde{\bQ} \, \mathrm{d} \tilde{\bV} - \tilde{\bV} \, \mathrm{d} \phi_0
                     + \bbeta \, \tilde{\bV} \, .
\eez
The expressions (\ref{thm_new_Miura_sol}) for the new solutions are invariant under 
$\bU \mapsto \tilde{\bU}$, $\bV \mapsto \tilde{\bV}$ and $\bX \mapsto \tbX$. 
Abstracting $\balpha, \bbeta$ from their above origin, leads to a slightly generalized version 
of Theorem~2.1 in \cite{DMH13SIGMA}, which is our Theorem~\ref{thm:main}. 
The freedom in the choice of $\balpha$ and $\bbeta$ turns out to be very helpful 
in order to derive a convenient expression for the solution of (\ref{preSylvester}) 
and (\ref{dX_eq}) in concrete examples. 
\end{remark}

\begin{remark}
\label{rem:CH=>Riemann_redundant}
A particularly important observation is the following. If the `Riemann equations' (\ref{P,Q_eqs}) 
are completely solvable via the method in Section~\ref{sec:CH}, with (\ref{P,Q_transform}) and  
$\mathrm{d}$- and $\bd$-constant $\tbP$ and $\tbQ$,  
then the computation in Remark~\ref{rem:transformation_to_alpha,beta} 
shows that Theorem~\ref{thm:main} is equivalent to its restriction, where $\bP$ and $\bQ$ 
are $\mathrm{d}$- and $\bd$-constant and commute with $\balpha$, respectively $\bbeta$. 
In this case, (\ref{P,Q_eqs}) is redundant and Theorem~\ref{thm:main} 
reduces to a method that generates solutions of (\ref{Miura}) from solutions of only \emph{linear} equations. 
We meet this situation if (\ref{Riemann_1}) is solvable by a Cole-Hopf 
transformation. But it does \emph{not} hold if (\ref{Riemann_1}) (hence (\ref{P,Q_eqs})) 
involves a Riemann equation. 
\end{remark}

The theorem includes 
a case, where solutions are generated from solutions of nonlinear `Riemann equations', and 
the linear equations (\ref{U,V_eqs}) are eliminated. 

\begin{corollary}
\label{cor:sol_via_Riem_sys}
Let $\phi_0,g_0 \in \mathrm{Mat}(m,m,\cB)$ solve (\ref{Miura}), and let 
$\bP,\bQ \in \mathrm{Mat}(n,n,\cB)$ be solutions of 
\bez
   \bd \bP - (\mathrm{d} \bP) \, \bP = - [ \balpha  \, , \, \bP ] \, , \qquad
   \bd \bQ - \bQ \, \mathrm{d} \bQ = [ \bbeta \, , \, \bQ ] \, ,
\eez
with $\balpha, \bbeta \in \Omega^1$. 
Let $\bX \in \mathrm{Mat}(n,n,\cB)$ be an invertible solution of the 
linear equations 
\bez
   \bX \, \bP - \bQ \, \bX = \bV_0 \, \bU_0 \, ,  \qquad
   \bd \bX - (\mathrm{d} \bX) \, \bP + (\mathrm{d} \bQ) \, \bX  
      = \bX \, \balpha + \bbeta \, \bX \, ,
\eez
where $\bU_0 \in \mathrm{Mat}(m,n,\cB)$ and $\bV_0 \in \mathrm{Mat}(n,m,\cB)$ are 
$\mathrm{d}$- and $\bd$-constant. 
Then 
\bez
    \phi = \phi_0 + \bU_0 \bX^{-1} \bV_0   \qquad \mbox{and} \qquad
    g = (I + \bU_0 (\bQ \, \bX)^{-1} \bV_0) \, g_0                  
\eez
yields a new solution of the Miura equation (\ref{Miura}), if 
\be
     (\mathrm{d} \phi_0) \, \bU_0 + \bU_0 \, \balpha = 0 \, , \qquad
     \bV_0 \, \mathrm{d} \phi_0 - \bbeta \, \bV_0 = 0 \, .   \label{phi0_alpha_beta_cond}
\ee
\end{corollary}
\begin{proof}
This is obtained by choosing $\bU$ and $\bV$ to be $\mathrm{d}$- and $\bd$-constant in 
Theorem~\ref{thm:main}. 
(\ref{U,V_eqs}) is then satisfied iff (\ref{phi0_alpha_beta_cond}) holds.
\end{proof}

One way to satisfy (\ref{phi0_alpha_beta_cond}) is to set 
$\bbeta = -\balpha = \bV_0 \, (\mathrm{d} \phi_0) \, \bU_0$,  
and $\bU_0 \bV_0 = I$ if $\mathrm{d} \phi_0 \neq 0$. 

\vskip.2cm
The above results are of a very general algebraic nature and may also find applications 
in mathematical problems remote from differential or difference equations, which we address in 
this work. 
In the following sections we choose the graded algebra $\Omega$ to be of the form
\be
    \Omega = \cA \otimes \bigwedge(\bbC^K) \, ,    \label{Omega_wedge}
\ee
where $\bigwedge(\bbC^K)$ is the exterior (Grassmann) algebra of the vector space $\bbC^K$. 
In this case it is sufficient to define $\mathrm{d}$ and $\bd$ on $\cA$. Then they extend to 
$\Omega$ by treating elements of $\bigwedge(\bbC^K)$ as constants. We denote by 
$\xi_1,\ldots,\xi_K$ a basis of $\bigwedge^1(\bbC^K)$. 
Furthermore, we will henceforth assume that $\phi$ in (\ref{Riemann_1}) (or (\ref{Riemann_2}))  
is an $m \times m$ matrix (with entries in a unital associative algebra $\cB$).

\subsection{Darboux transformations for `Riemann equations'}
Now we state conditions under which Theorem~\ref{thm:main} generates solutions 
of the special cases (\ref{Riemann_1}) and (\ref{Riemann_2}) of the Miura equation 
(\ref{Miura}). These are the `Riemann equations' on which we concentrate in this work.

\begin{corollary}
\label{cor:main}
Let $\phi_0$ be a solution of (\ref{Riemann_1}), respectively (\ref{Riemann_2}). 
Let $\bP,\bQ, \bU, \bV, \bX$ be solutions of (\ref{P,Q_eqs}) - (\ref{dX_eq}) and
\be
     \bQ \, \bV = \bV \, \phi_0 \, ,  \qquad \mbox{respectively} \qquad 
     \bU \, \bP = - \phi_0 \, \bU  \, .   \label{cor_constraints}
\ee
Let $\phi$ be given by the expression in (\ref{thm_new_Miura_sol}). Then 
$\phi$ is a solution of (\ref{Riemann_1}), respectively (\ref{Riemann_2}). 
\end{corollary}
\begin{proof}
Setting $g = \pm \phi$ in (\ref{Miura}), turns it into (\ref{Riemann_1}). 
The additional condition, the first of (\ref{cor_constraints}), originates from evaluating $g = \pm \phi$ 
using the expressions (\ref{thm_new_Miura_sol}) for $\phi$ and $g$. 
Correspondingly, setting $g^{-1} = \pm \phi$ in (\ref{Miura}), it becomes (\ref{Riemann_2}).
Using $g^{-1} = g_0^{-1} \, [I - \bU \, (\bX \bP)^{-1} \, \bV]$, 
we are led to the second of (\ref{cor_constraints}). 
\end{proof}

\begin{remark}
As a consequence of the assumptions in Corollary~\ref{cor:main}, we obtain
\bez
     (\bX \bP \bX^{-1}) \, \bV = \bV \phi \, ,
              \qquad \mbox{respectively} \qquad
     \bU \, (\bX^{-1} \bQ \bX)  = - \phi \, \bU  \, .
\eez
These are counterparts of (\ref{cor_constraints}). 
If $n=m$ and $\bQ=\phi_0=0$, so that the first of conditions (\ref{cor_constraints}) holds, and 
if $\bU$ is invertible, then $\phi = \bU \, \bP \, \bU^{-1}$. In this special case, 
the Corollary thus boils down to the method described in Section~\ref{sec:CH}. 
\end{remark}

\begin{remark}
\label{rem:no-go_for_Riem}
Let $\cB$ be an algebra of (real or complex) functions of independent variables. 
We assume the spectrum condition $\sigma(\bP) \cap \sigma(\bQ) = \emptyset$, and $n>m$. 
The first of (\ref{cor_constraints}) then implies that the $n \times nm$ matrix
$(\bV | \bQ \bV | \cdots |\bQ^{n-1} \bV) = (\bV | \bV \phi_0 | \cdots |\bV \phi_0^{n-1})$
has at most $m$, hence less than $n$ linearly independent columns. In this case the pair $(\bQ,\bV)$ 
is said to be \emph{not controllable} (see, e.g., \cite{Hearon77}). 
A corresponding statement holds in case of the second of (\ref{cor_constraints}). 
Theorem 3 in \cite{Hearon77} 
then says that (\ref{preSylvester}) has no \emph{invertible} solution, so that under the stated 
conditions the solution-generating method in Corollary~\ref{cor:main} does \emph{not} work.  
This concerns in particular the case of the Riemann equation (see   
Section~\ref{subsec:Riem}). This negative result 
should not come as a surprise since `soliton methods' are known not to work in case 
of hydrodynamic-type systems of which the Riemann equation is the prototype. 
\end{remark}

\begin{remark}
\label{rem:thm_gen}
In (\ref{Riemann_1_mod}) we met an apparently generalized version of (\ref{Riemann_1}). 
By a redefinition of the derivation $\bd$, we can cast it into the form (\ref{Riemann_1}). 
In Theorem~\ref{thm:main}, we can perform the reverse step.
The Miura equation then reads $( \bd g + [\gamma , g] ) \, g^{-1} = \mathrm{d} \phi$. Equations
(\ref{P,Q_eqs}), (\ref{preSylvester}), (\ref{dX_eq}) and (\ref{thm_new_Miura_sol}) remain unchanged, 
while (\ref{U,V_eqs}) is modified to
$\bd \bU = (\mathrm{d} \bU) \, \bP + (\mathrm{d} \phi_0) \, \bU + \bU \, \balpha - \gamma \, \bU$ and  
$\bd \bV = \bQ \, \mathrm{d} \bV - \bV \, \mathrm{d} \phi_0 + \bbeta \, \bV + \bV \, \gamma$.  
The first part of Corollary~\ref{cor:main} then holds correspondingly. For example, 
if the first of (\ref{cor_constraints}) is satisfied, and if $\phi_0$ solves (\ref{Riemann_1_mod}), 
then $\phi$, given by the formula in (\ref{thm_new_Miura_sol}), solves the same equation, i.e.,
$\bd \phi - (\mathrm{d} \phi) \, \phi + [ \gamma \, , \,  \phi ] = 0$.
\end{remark}

\section{Matrix Riemann equations and integrable discretizations}
\label{sec:Riemann}
In this section, we consider the case $\Omega^1 = \cA$.  
$\mathrm{d}$ and $\bd$ then have to be derivations of $\cA$.

\subsection{Riemann equation} 
\label{subsec:Riem}
Let $\cA$ be the algebra of matrices of (real or complex) smooth functions of 
independent variables $x,t$. For $f \in \cA$, let
\bez
     \mathrm{d} f = f_x  \, , \qquad   \bd f = f_t  \, ,
\eez
where a subscript indicates a partial derivative with respect to the corresponding 
independent variable. (\ref{Riemann_1}) is now the matrix Riemann equation
\be
     \phi_t = \phi_x \, \phi \, .   \label{Riem_eq}
\ee
As a consequence of (\ref{Riem_eq}), the eigenvalues of $\phi$ satisfy the corresponding scalar 
version of this equation \cite{Bogo92V}, hence a scalar Riemann equation. 
(\ref{phi_transform}) only generates new solutions in the matrix case ($m>1$). 
Let $\phi_0$ be a solution of (\ref{Riem_eq}) that commutes with its partial 
derivatives. By use of the method of characteristics, solutions of (\ref{CH_Phi_eq}), 
with $\gamma=0$, are then given by (also see \cite{Sant+Zenc07})
\bez
  \Phi = A_0 
    + \sum_{i=1}^k A_i \, f_i( t \, \phi_0 + x \, I) \, ,  
\eez
with any analytic functions $f_i$ and constant $m \times m$ matrices $A_0,A_i$. 
(\ref{phi_transform}) then yields a new solution of (\ref{Riem_eq}). 
According to Remark~\ref{rem:no-go_for_Riem}, our Corollary~\ref{cor:main} is not helpful 
in this particular example.

\subsection{Semi-discrete Riemann equation} 
\label{subsec:sdRiem}
Let $\cA_0$ be the algebra of matrices of functions on $\bbR \times \bbZ$, smooth in the 
first variable $t$, and $\cA = \cA_0[\bbS,\bbS^{-1}]$. For $f \in \cA$, we set
\be
     \mathrm{d} f = [\bbS, f ]  \, , \qquad   \bd f = f_t  \, ,  \label{sdRiem_bdc}
\ee
where $\bbS$ is the shift operator in the discrete variable $k$. Then, in terms of 
\be
         \varphi = \phi \, \bbS \, ,   \label{sdRiem_varphi}
\ee       
and using the notation
\bez
     \varphi^+ := \bbS \varphi \bbS^{-1} \, , \qquad \varphi^- := \bbS^{-1} \varphi \bbS \, ,
\eez
(\ref{Riemann_1}) is the semi-discrete matrix Riemann equation
\be
     \varphi_t = (\varphi^+ - \varphi) \, \varphi \, ,  \label{sd_Riemann_eq_1}
\ee
where $\varphi$ can now be restricted to be an $m \times m$ matrix of functions (not 
involving the shift operator explicitly). 
Such a matrix equation already appeared in \cite{LRB83}. 
The scalar version has been called `lattice Burgers equation' in \cite{Matv+Sall83,Matv+Sall91,WGXGQY12}. 
It also appears as a symmetry of a discrete Burgers equation in \cite{HLW99}, and   
in a model for socio-economical systems in \cite{Ben+Redn05}. 
A lattice spacing $h$ can be introduced via a rescaling $t \mapsto t/h$.  
If $\varphi$, $\varphi_t$
and $(\varphi^+ - \varphi)/h$ have limits as $h \to 0$, keeping $x := k \, h$ fixed, 
then $\varphi$ solves the Riemann equation $\varphi_t = \varphi_x \, \varphi$. 

\begin{remark}
Let us fix $\varphi(k_0,t)$ at some lattice point $k_0$. Writing (\ref{sd_Riemann_eq_1}) as 
$\varphi^+ = \varphi + \varphi_t \, \varphi^{-1}$, 
as long as the inverse of $\varphi$ exists, extends this to a solution for $k > k_0$. 
To the left of $k_0$ on the lattice,
we obtain from (\ref{sd_Riemann_eq_1}) iteratively at each lattice point (with $k < k_0$) 
a matrix Riccati equation:
\bez
    \varphi_t(k_0-n,t) = - \varphi(k_0-n,t)^2 + \varphi(k_0-n+1,t) \, \varphi(k_0-n,t)  
            \qquad n=1,2,\ldots \, .
\eez
This presents another view of the integrability of (\ref{sd_Riemann_eq_1}) and, moreover, illustrates 
the `asymmetry' arising from the presence of the forward difference in (\ref{sd_Riemann_eq_1}).
\end{remark}

The alternative `Riemann equation' (\ref{Riemann_2}) takes the form $\varphi_t = - \varphi \, (\varphi - \varphi^-)$,
which is obtained from (\ref{sd_Riemann_eq_1}) for the transpose of $\varphi$, if we replace $\bbS$ 
by its inverse. There is an alternative integrable semi-discretization 
of the Riemann equation, the (Lotka-) Volterra lattice equation, see Appendix~\ref{app:LV}.

\subsubsection{Cole-Hopf transformation}
Choosing $\phi_0 = \bbS^{-1}$, (\ref{phi_transform}) becomes
\be
      \varphi = \Phi \, (\Phi^-)^{-1} \, ,   \label{sdRiem1_CH}
\ee
and (\ref{CH_Phi_eq}) with $\gamma = A(t)$, where 
$A$ commutes with $\bbS$, takes the form
\be
      \Phi_t = \Phi^+ - \Phi + \Phi \, A  \, ,    \label{sd_transport_eq}
\ee
which for $A=0$ is a semi-discrete version of the transport equation. 
Since (\ref{sd_Riemann_eq_1}) is autonomous, if $\varphi$ is a solution, then also $\varphi^+$ is a solution. 
We can therefore redefine $\varphi$ and replace (\ref{sdRiem1_CH}) by $\varphi = \Phi^+ \, \Phi^{-1}$.
Furthermore, we can 
eliminate $A$ by a redefinition of $\Phi$ that preserves (\ref{sdRiem1_CH}). Equations 
(\ref{sdRiem1_CH}) and (\ref{sd_transport_eq}) constitute a discrete Cole-Hopf transformation 
for the semi-discrete Riemann equation (\ref{sd_Riemann_eq_1}), also see 
\cite{LRB83,Matv+Sall83,Matv+Sall91}.

\begin{example}
\label{ex:sdRiem1_multi-kink_via_CH}
A set of solutions of (\ref{sd_transport_eq}), with $A=0$, is given by 
\bez
    \Phi = I + \sum_{i=1}^N \bA_i \, e^{\bTh_i} \, \bB_i  \, , \qquad
    \bTh_i := \bLa_i \, k + (e^{\bLa_i} - \bI) \, t  \, , 
\eez
where $\bA_i$ and $\bB_i$ are constant $m \times n$, respectively $n \times m$ matrices, and the $\bLa_i$ 
are $n \times n$ matrices.  
In the scalar case ($m=1$), we can set $N=1$ without restriction of generality.
Choosing $\bLa = \mathrm{diag}(\lambda_1,\ldots,\lambda_n)$, we obtain 
\be
    \Phi = 1 + \sum_{i=1}^n e^{\theta_i} \, , \qquad
    \theta_i := \lambda_i \, k + (e^{\lambda_i}-1) \, t + \gamma_i \, ,  \label{sdRiem_1_scalar_CH_sols}
\ee
with constants $\gamma_i$. If the constants are real, then (\ref{sdRiem1_CH}) yields an 
$n$-\emph{kink} solution (see Fig.~\ref{fig:sdRiemann1_2kink_via_CH}) of the scalar version of 
(\ref{sd_Riemann_eq_1}), cf. \cite{Matv+Sall83,Matv+Sall91}. 
In the continuum limit, such solutions become constant. Thus, regarding (\ref{sd_Riemann_eq_1}) 
as a discretization of the Riemann equation, the kink solutions are simply artifacts of the 
discretization. 
If $\bLa$ has non-diagonal Jordan normal form, further solutions are obtained.  
Examples of matrix shock wave solutions 
already appeared in \cite{LRB83}, derived via B\"acklund transformations.
\begin{SCfigure}[1.4][hbtp] 
\hspace{.5cm}
\includegraphics[scale=.4]{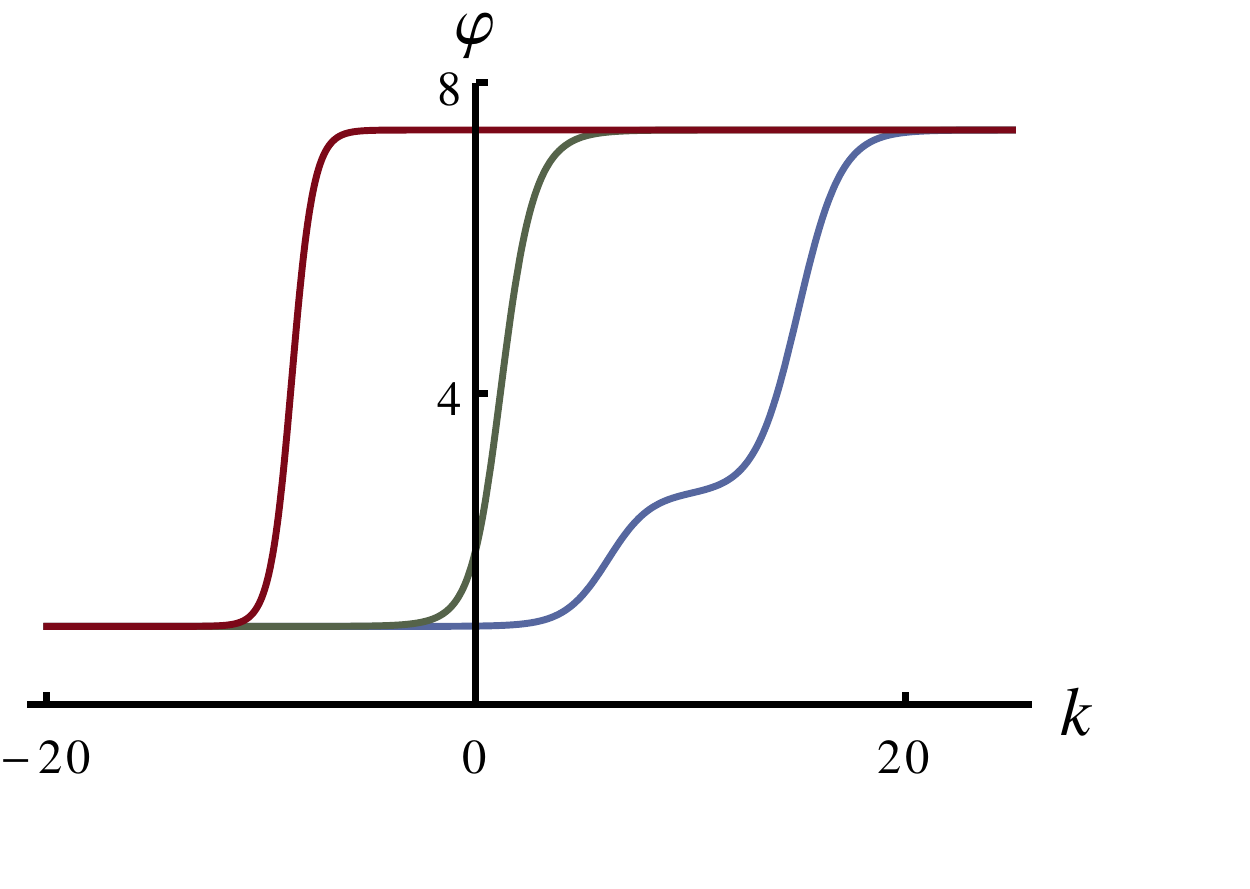} 
\caption{Plots (with interpolation) of a regular solution of the scalar semi-discrete 
Riemann equation (\ref{sd_Riemann_eq_1}), as given in Example~\ref{ex:sdRiem1_multi-kink_via_CH}, 
with $n=2$, $\lambda_1=1$, $\lambda_2=2$, $\gamma_1=\gamma_2=0$. 
Evolution in time $t$ is from right to left.
For negative time, this is a $2$-kink solution (right curve), which turns into a single 
kink at $t=0$ (middle curve) and then becomes steeper and steeper (shock wave, left curve). 
}
\label{fig:sdRiemann1_2kink_via_CH} 
\end{SCfigure}
\end{example}

The scalar semi-discrete Riemann equation (\ref{sd_Riemann_eq_1}) is perhaps the 
simplest soliton equation (calling a kink a soliton). Although it is obtained via the 
simplest discretization from the Riemann equation, which is \emph{not} a soliton equation (though 
integrable by the method of characteristics, or hodograph method), it is of a rather 
different nature. 
The behavior of the multiple kink solutions is actually very similar to that of 
corresponding solutions of the scalar Burgers equation, see Section~\ref{sec:B&KP}. 
This is explained by the fact that (\ref{sd_Riemann_eq_1}) is 
a member of a semi-discrete Burgers hierarchy, see Remark~\ref{rem:sd-Hopf&Burgers} below.

\begin{example}
We note that, if $\Phi$ has a nowhere vanishing continuum limit, then $\varphi$ tends to $I$. 
The kink solutions in Example~\ref{ex:sdRiem1_multi-kink_via_CH} are corresponding examples. 
Setting $A =I$ in (\ref{sd_transport_eq}), and introducing the lattice spacing $h$, it reads 
$\Phi_t = \frac{1}{h} \Phi^+$, which is singular as $h \to 0$. In the scalar case ($m=1$), a particular 
solution is given by $\Phi = (-1)^{k+1} (k+1)! \, h^{k+2} t^{-(k+2)}$ for $h>0$. Although it has no limit 
as $h \to 0$, we find $\varphi = \Phi/\Phi^- = - (k+1) h/t \rightarrow -x/t$, which is a special 
solution of the Riemann equation.
\end{example}

\begin{remark}
\label{rem:sdRiem_Lax_pair}
(\ref{sdRiem1_CH}), written as $\Phi = \varphi \, \Phi^-$, together with (\ref{sd_transport_eq}) 
is a linear system that has (\ref{sd_Riemann_eq_1}) as its compatibility condition, hence the two equations 
constitute a Lax pair for (\ref{sd_Riemann_eq_1}). 
Choosing $A = (1-\lambda) \, I$, with a parameter $\lambda$, (\ref{sd_transport_eq}) reads
\bez
    \Phi_t - \Phi^+ + \lambda \, \Phi = 0 \, .
\eez
If $\Phi$ is a solution, then also $\Phi^+ - \mu \, \Phi$, with any constant $\mu$. 
So we may replace (\ref{sdRiem1_CH}) by $\Phi^+ - \mu \, \Phi = \varphi \, (\Phi - \mu \, \Phi^-)$  
(now with a different $\varphi$), which is 
\bez
   \Phi^+ = (\mu \, I + \varphi) \, \Phi - \mu \, \varphi \, \Phi^- \, .
\eez
Setting $\mu = 2 \, \lambda$, $\psi_1 := \Phi$ and $\psi_2 := \sqrt{2 \lambda} \, \Phi^-$, 
we recover the Lax pair considered in \cite{WGXGQY12} in the scalar case. 
\end{remark}

\subsubsection{Darboux transformations}
We exploit Corollary~\ref{cor:main} in the case where $\balpha, \bbeta, \bP,\bQ$ are constant. 
The restriction to $\mathrm{d}$- and $\bd$-constant $\bP$ and $\bQ$ is suggested by 
Remark~\ref{rem:CH=>Riemann_redundant}. 
Setting
\bez
     \bP = \bA \, \bbS^{-1} \, , \qquad   
     \bQ = \bB \, \bbS^{-1} \, , \qquad
     \bX = \bbS \, \tbX \, ,
\eez
eliminates explicit appearances of the shift operator in the equations in Theorem~\ref{thm:main}. 
Then (\ref{P,Q_eqs}) is satisfied if $[\balpha, \bA] = [\bbeta, \bB] = 0$.
The remaining equations in Theorem~\ref{thm:main} now take the form
\bez
  \bU_t = (\bU^+ - \bU) \, \bA + (\varphi_0^+ - \varphi_0) \, \bU + \bU \, \balpha \, , \qquad
  \bV_t = \bB \, (\bV - \bV^-) - \bV \, (\varphi_0^+ - \varphi_0) + \bbeta \, \bV  \, , 
\eez
and 
\bez
    \tbX^+ \bA - \bB \, \tbX - \bV \, \bU = 0\, , \qquad
    \tbX_t - (\tbX^+ - \tbX) \, \bA + (\bV - \bV^-) \, \bU 
            - \tbX \, \balpha - \bbeta \, \tbX = 0 \, ,
\eez
which are compatible equations.  
By use of the first equation for $\tbX$, we can replace the second by
\bez
   \tbX_t = \bV^- \, \bU + (\bbeta + \bB) \, \tbX + \tbX \, ( \balpha - \bA) \, .
\eez
In the case under consideration, the first condition of (\ref{cor_constraints}) in Corollary~\ref{cor:main} 
has to be considered, which here takes the form
\bez
     \bB \, \bV^- = \bV \, \varphi_0 \, .
\eez
Using (\ref{sdRiem_varphi}), the solution formula for $\phi$ in (\ref{thm_new_Miura_sol}) reads
\be
    \varphi = \varphi_0 + \bU \tbX^{-1} \bV^- \, .    \label{sol_sd_Riem}
\ee
In all these equations, we can now restrict $\varphi_0,\bU,\bV,\tbX$ to $\cA_0$, and $\bA$, $\bB$ 
to be matrices over $\bbC$. 

\begin{Proposition}
\label{prop:sd_Riem}
Let $\varphi_0$ solve (\ref{sd_Riemann_eq_1})
and $\bU$, $\bV$ be solutions of
\be
    \bU_t = \bU^+ \bA + (\varphi_0^+ -\varphi_0) \, \bU \, , \qquad 
    \bV_t = - \bV\varphi_0^+  \, ,  \qquad
    \bB \bV^- = \bV \varphi_0 \, ,    \label{U,V_eqs_sd_Riem}
\ee
with constant $n \times n$ matrices $\bA$, $\bB$. Then a new solution of
(\ref{sd_Riemann_eq_1}) is given by (\ref{sol_sd_Riem}) with 
\be
    \tbX = \bC^- + \int_{0}^t \bV^- \bU \, dt \, ,   \label{sol_sd_Riem_X}
\ee
where $\bC$ does not depend on $t$ and satisfies the constraint
\be 
    \bC \bA - \bB \bC^- = \bV \bU \Big|_{t=0} \, .  \label{preSylvester_sd_Riem}
\ee
\end{Proposition}
\begin{proof} We choose $\balpha = \bA$ and $\bbeta = -\bB$. 
Then $\tbX_t = \bV^- \, \bU$, which integrates to the stated expression. 
It remains to solve $\tbX^+ \bA - \bB \, \tbX = \bV \, \bU$. It is 
sufficient to do this at $t=0$, where $\tbX = \bC^-$, and this leads to 
(\ref{preSylvester_sd_Riem}).
\end{proof}

\begin{remark}
In the scalar case ($m=1$), under the conditions specified in Proposition~\ref{prop:sd_Riem}, 
we have 
\bez
    \varphi = \varphi_0 + \mathrm{tr}(\bV^- \, \bU \tbX^{-1})
            = \varphi_0 + \mathrm{tr}(\tbX_t \tbX^{-1})
            = \varphi_0 + (\ln \det \tbX)_t \, .  
\eez
Alternatively, we can write
\bez
    \varphi &=& \varphi_0 + \mathrm{tr}( \bV^- \, \bU \tbX^{-1} ) 
            = \varphi_0 \, ( 1 + \mathrm{tr}( \bB^{-1} \bV \, \bU \, \tbX^{-1} ) )
            = \varphi_0 \, ( 1 + \bU \, \tbX^{-1} \, \bB^{-1} \bV )   \\
            &=& \varphi_0 \, \det( \bI + \bB^{-1} \bV \, \bU \, \tbX^{-1} )
            = \varphi_0 \, \det\Big( \bI + \bB^{-1} \, (\tbX^+ \bA - \bB \, \tbX) 
              \, \tbX^{-1} \Big)   \\
            &=& \varphi_0 \, \det( \tbX^+ \, \bA \, \tbX^{-1} \, \bB^{-1} ) 
             = \varphi_0 \, \det(\bA) \, \det(\bB)^{-1} \,  
                  \det(\tbX)^+ \, \det(\tbX)^{-1} \, ,  
\eez
using Sylvester's determinant theorem in the fourth step. This makes contact with the 
Cole-Hopf transformation (\ref{sdRiem1_CH}).
\end{remark}

\begin{example}
\label{ex:sd-Riem_multi-kinks}
We consider the scalar case and set $\varphi_0 = \mu \in \bbC \setminus \{0\}$, 
$\bA = \bB = \mu \, \bI$. Writing $\bU = (u_1,\ldots,u_n)$ and 
$\bV = (v_1,\ldots,v_n)^\intercal$ (where $^\intercal$ denotes the transpose), 
the equations in (\ref{U,V_eqs_sd_Riem}) are solved by
\bez   
  u_j = (e^{\lambda_j}-1) \, e^{ \lambda_j \, k + \mu \, e^{\lambda_j} \, t + \gamma_j } \, , \qquad 
  v_j = \mu \, e^{-\mu \, t} \, , \qquad
  \lambda_j, \gamma_j \in \bbC \, , \qquad j=1,\ldots,n \, .
\eez
(\ref{sol_sd_Riem_X}) and (\ref{preSylvester_sd_Riem}) are then solved by
\bez  
 \tbX 
   = \left( c_{ij} + e^{\theta_j} \right)_{i,j=1}^n \, ,
\eez   
where $c_{ij}$ are arbitrary constants, which we choose as $\delta_{ij}$, and 
$\theta_i := \lambda_i \, k + \mu \, ( e^{\lambda_i} -1) \, t + \gamma_i$. 
Using Sylvester's determinant theorem, we obtain $\det \tbX = 1 + \sum_{j=1}^n e^{\theta_j}$,
which is (\ref{sdRiem_1_scalar_CH_sols}) if we set $\mu=1$. The latter can be achieved by an obvious 
scaling symmetry of (\ref{sd_Riemann_eq_1}).
\end{example}

\subsection{Discrete Riemann equation}
\label{subsec:discr_Riem}
Let $\cA_0$ be the algebra of matrices of functions on $\bbZ^2$. Let   
$\bbS_0, \bbS_1$ be the corresponding commuting shift operators and 
$\cA = \cA_0[\bbS_0^{\pm 1}, \bbS_1^{\pm 1}]$.
We will use the notation 
\bez
    f_{,0} := \bbS_0 \, f \, \bbS_0^{-1} \, , \qquad 
    f_{,1} := \bbS_1 \, f \, \bbS_1^{-1} \, ,
\eez
and also $f_{,-0} := \bbS_0^{-1} \, f \, \bbS_0$, $f_{,-1} := \bbS_1^{-1} \, f \, \bbS_1$.
Let
\bez
  \mathrm{d} f = \frac{1}{h_0} \, [ \bbS_1^{-1} \bbS_0 , f]  \, , \qquad
  \bd f = \frac{1}{h_1} \, [ \bbS_1^{-1} , f ]  \, ,  
\eez
with constants $h_0,h_1 \neq 0$. In terms of 
\bez
     \varphi = \phi \, \bbS_0  \, ,   
\eez
Equation (\ref{Riemann_1}) becomes 
\be
   \frac{1}{h_1} \, ( \varphi_{,1} - \varphi ) &=& \frac{1}{h_0} \, ( \varphi_{,1} - \varphi_{,0} ) \, \varphi 
                       \nonumber  \\
             &=& \frac{1}{h_0} \, ( \varphi_{,1} - \varphi ) \, \varphi
                 - \frac{1}{h_0} \, ( \varphi_{,0} - \varphi ) \, \varphi    \, .
                 \label{discrete_Riemann_eq_h}
\ee
If $h_0 =-1$, this formally tends to the semi-discrete Riemann equation (\ref{sd_Riemann_eq_1}) as $h_1 \to 0$.
In the following we set $h_0 = h_1 = 1$, hence we consider the 
\emph{discrete matrix Riemann equation}
\be
    \varphi_{,1} - \varphi  = ( \varphi_{,1} - \varphi_{,0} ) \, \varphi  \, ,
      \label{discrete_Riemann_eq_v1}
\ee
which can be rewritten as $\varphi_{,1} =  (I - \varphi_{,0}) \, \varphi \, (I - \varphi)^{-1}$. 

(\ref{Riemann_2}) 
has the form $\varphi_{,0,1} - \varphi_{,0} = - \varphi_{,0,1} \, (\varphi_{,0} - \varphi_{,1} )$,
which becomes (\ref{discrete_Riemann_eq_v1}) for the transpose of $\varphi$, if we replace the two 
shift operators by their inverses.

\subsubsection{Cole-Hopf transformation}
Choosing $\phi_0 = \bbS_0^{-1}$ in (\ref{phi_transform}) and $\gamma = -A \, \bbS_1^{-1}$ 
in (\ref{CH_Phi_eq}), with $A_{,0} = A$, we obtain for (\ref{discrete_Riemann_eq_v1}) 
the discrete Cole-Hopf transformation
\be
     \varphi = \Phi \, \Phi_{,-0}^{-1} \, , \qquad
     \Phi_{,0} - \Phi = (\Phi A)_{,1} \, .    \label{dRiem_1_CH}
\ee

\begin{example}
\label{ex:dRiem1_multi-kink_via_CH}
A set of solutions of the linear equation in (\ref{dRiem_1_CH}), with the choice $A=-I$, is given by 
\bez
    \Phi = \sum_{i=1}^N \bA_i \, \bLa_i^{k_0} \, (I- \bLa_i)^{k_1} \, \bB_i  \, , 
\eez
where $\bA_i$ and $\bB_i$ are constant $m \times n$, respectively $n \times m$ matrices, and the $\bLa_i$ 
are constant $n \times n$ matrices. In the scalar case ($m=1$), we can set $N=1$ without restriction 
of generality. Choosing $\bLa = \mathrm{diag}(\lambda_1,\ldots,\lambda_n)$, we obtain 
\bez
    \Phi = \sum_{i=1}^{n+1} \gamma_i \, \lambda_i^{k_0} \, (1- \lambda_i)^{k_1} \, ,
\eez
with constants $\gamma_i$. If the constants are real and $\gamma_i >0$, $0 < \lambda_i <1$, 
then (\ref{dRiem_1_CH}) yields an $n$-\emph{kink} solution of the scalar version of (\ref{discrete_Riemann_eq_v1}), 
also see Fig.~\ref{fig:d_Riemann_eq_2kinks_via_CH}. We can write the last expression in the form
\be
    \Phi = \mu^{k_0} \, (1-\mu)^{k_1} \, \Big( 1 
     + \sum_{i=1}^n \gamma_i \, \lambda_i^{k_0} \, \Big(\frac{1- \mu \lambda_i}{1-\mu} \Big)^{k_1} \Big) \, ,
         \label{scalar_dRiem_1_Phi_of_n-kink}
\ee
after the redefinitions $\lambda_1 \mapsto \mu$, $\lambda_{i+1}/\lambda_1 \mapsto \lambda_i$, 
$\gamma_{i+1}/\gamma_1 \mapsto \gamma_i$, $i=1,\ldots,n$,
and a rescaling of $\Phi$ that preserves $\varphi$. 
\begin{SCfigure}[1.4][h] 
\hspace{1.cm}
\includegraphics[scale=.5]{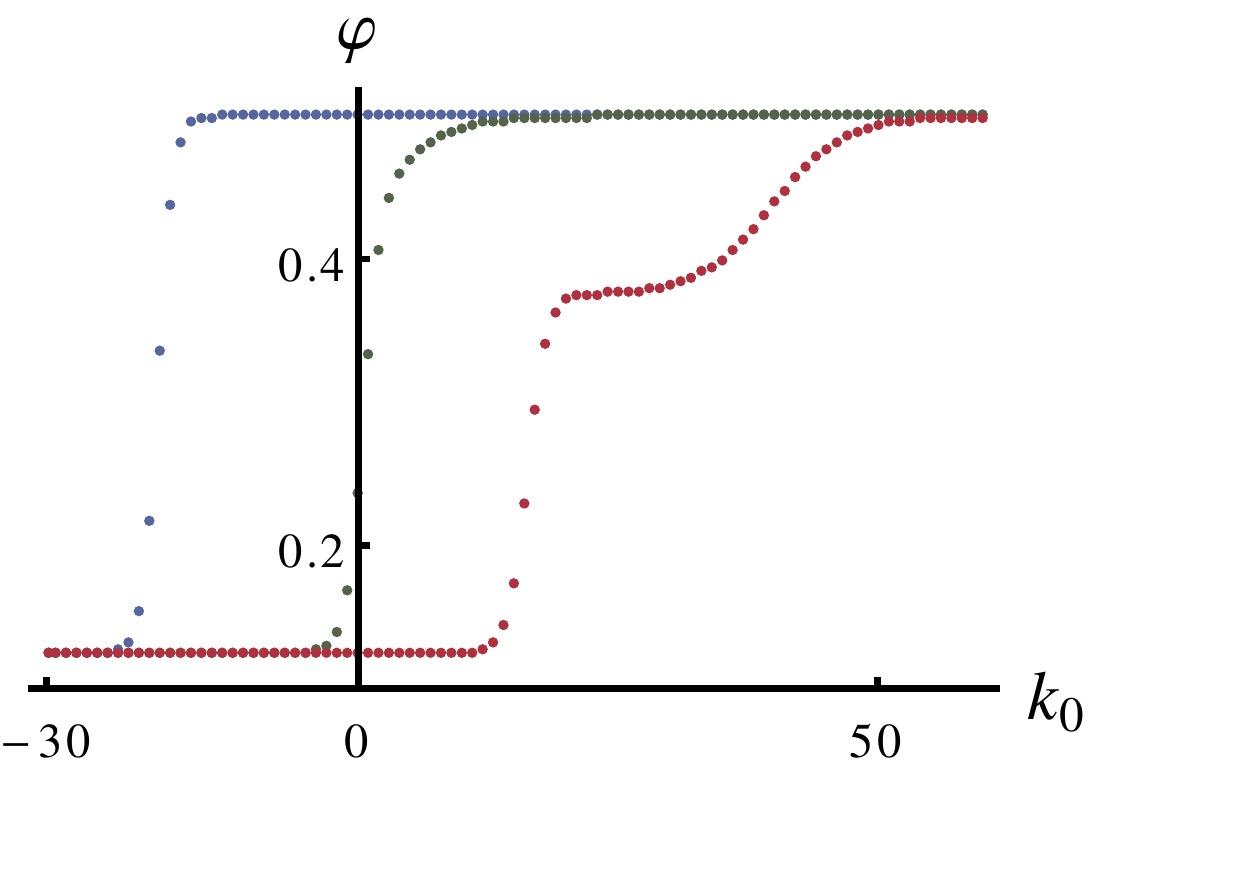} 
\caption{A regular $2$-kink solution of the scalar discrete Riemann equation 
(\ref{discrete_Riemann_eq_v1}) obtained via the discrete Cole-Hopf transformation (\ref{dRiem_1_CH}) with 
(\ref{scalar_dRiem_1_Phi_of_n-kink}), where we chose $\mu=1/2$, $\lambda_1=1/4$, $\lambda_2=3/4$,  
$\gamma_1 = \gamma_2 = 1$. The plots, from left to right,  
correspond to consecutive values ($-50, 0, 50$) of $k_1$. 
}
\label{fig:d_Riemann_eq_2kinks_via_CH} 
\end{SCfigure}
\end{example}

\begin{remark} 
Choosing $A = - \lambda \, I$, with a constant $\lambda$, the second equation in 
(\ref{dRiem_1_CH}) reads
\bez
    \Phi_{,0} - \Phi + \lambda \, \Phi_{,1} = 0 \, .
\eez
Writing the first equation in (\ref{dRiem_1_CH}) as a linear equation 
for $\Phi$, we have a \emph{Lax pair} for the discrete Riemann equation, with spectral parameter $\lambda$.  
If $\Phi$ is a solution of the above equation, then also $\Phi_{,0} - \mu \, \Phi$, with any constant $\mu$. 
We may thus replace the first equation in (\ref{dRiem_1_CH}) by 
\bez
     \Phi_{,0} = (\mu \, I + \varphi ) \, \Phi - \mu \, \varphi \, \Phi_{,-0}  
\eez
(with a different $\varphi$). This is the counterpart of the 
corresponding Lax pair for the semi-discrete Riemann equation, see Remark~\ref{rem:sdRiem_Lax_pair}.
\end{remark}

\subsubsection{Darboux transformations}
In Corollary~\ref{cor:main} we are led to set 
\bez
     \balpha = \tilde{\balpha} \bbS_1^{-1} \, , \quad 
     \bbeta = \tilde{\bbeta} \, \bbS_1^{-1} \, , \quad
     \bP = \bA \, \bbS_0^{-1} \, , \quad
     \bQ = \bB \, \bbS_0^{-1} \, , \quad
     \bX = \tbX\bbS_0 \, ,
\eez
with constant matrices $\tilde{\balpha}, \tilde{\bbeta}, \bA, \bB$. 
Then (\ref{P,Q_eqs}) is satisfied if $[\tilde{\balpha}, \bA] = [\tilde{\bbeta}, \bB] = 0$. 
The remaining equations we have to consider now take the form
\bez
  && \bU - \bU_{,1} = (\bU_{,0} - \bU_{,1}) \, \bA 
      + (\varphi_{0,0} - \varphi_{0,1}) \, \bU + \bU_{,1} \, \tilde{\balpha} \, ,  \\
  && \bV_{,1} - \bV = \bB \, (\bV_{,-0,1} - \bV ) 
      - \bV_{,1} \, (\varphi_{0,1} - \varphi_{0,0} ) - \tilde{\bbeta} \, \bV  \, ,
\eez
and
\bez
   && \tbX_{,0} \bA - \bB \, \tbX - \bV_{,0} \, \bU_{,0} = 0 \, , \\
   && \tbX_{,1} - \tbX 
       - ( \tbX_{,1} - \tbX_{,0} ) \bA
       + (\bV_{,1} - \bV_{,0} ) \, \bU_{,0} 
       + \tbX_{,1} \, \tilde{\balpha} + \tilde{\bbeta} \, \tbX = 0 \, .
\eez
By use of the first equation for $\tbX$, we can replace the second by
\bez
    \tbX_{,1} - \tbX 
       + \tbX_{,1} \, ( \tilde{\balpha} - \bA )
       + ( \tilde{\bbeta} + \bB ) \, \tbX        
       + \bV_{,1} \, \bU_{,0}  = 0   \, .
\eez
We have to consider the first condition of (\ref{cor_constraints}), which 
here takes the form
\bez
     \bB \, \bV = \bV_{,0} \, \varphi_{0,0} \, .
\eez
The solution formula in (\ref{thm_new_Miura_sol}) reads
\be
    \varphi = \varphi_0 + \bU (\tbX^{-1} \bV)_{,-0}  \, .  \label{sol_discr_Riem_v1}
\ee
We can now restrict $\varphi_0,\bU,\bV,\tbX$ to $\cA_0$.

\begin{Proposition}
\label{prop:discr_Riem_v1} 
Let $\varphi_0$ solve (\ref{discrete_Riemann_eq_v1}), and let $\bU$, $\bV$ 
be solutions of the linear equations
\bez
  && \bU_{,1} - \bU + \bU_{,0} \, \bA = ( \varphi_{0,1} - \varphi_0 ) \, \bU \, ,  \nonumber   \\
  && \bB \, \bV_{,-0} = \bV \, \varphi_0 \, , \qquad
     \bV_{,1} - \bV = \bV_{,1} \, \varphi_{0,0}  \, ,  
\eez
with constant matrices $\bA$, $\bB$. Let
\bez
    \tbX(k_0,k_1) = \bF(k_0) - \left\{ \begin{array}{l} \sum_{j=0}^{k_1-1} \bV(k_0,j+1) \, \bU(k_0+1,j) \\
           \sum_{j=k_1-1}^0 \bV(k_0,j+1) \, \bU(k_0+1,j) \end{array} \right. 
           \quad \mbox{if} \quad  \begin{array}{l} k_1 \geq 1 \\ k_1 < 1 \end{array}  
       \, ,         
\eez
where $k_0$ and $k_1$ are the discrete variables on which the shift operators 
$\bbS_0$, respectively $\bbS_1$ act, and $\bF$ is an arbitrary $n \times n$ matrix function satisfying 
\bez
    \bF(k_0+1) \, \bA - \bB \bF(k_0) - \bV(k_0+1,0) \, \bU(k_0+1,0)=0 \, .
\eez
Then a new solution of (\ref{discrete_Riemann_eq_v1}) is given by (\ref{sol_discr_Riem_v1}).
\end{Proposition}
\begin{proof}
We choose $\tilde{\balpha} = \bA$ and $\tilde{\bbeta} = -\bB$. Then the equations for $\tbX$
read
\bez
  \tbX_{,1} - \tbX + \bV_{,1} \, \bU_{,0} = 0  \, ,  \qquad
  \tbX_{,0} \, \bA - \bB \, \tbX = \bV_{,0} \, \bU_{,0} 
     \, .   
\eez
The first equation is completely solved by the expression for $\tbX$ in 
the proposition. The second equation then results in the stated constraint. 
\end{proof}

\begin{remark}
In the scalar case, we have
\be
   \varphi 
  &=& \varphi_0 \, (1 + \bU \tbX_{,-0}^{-1} \bB^{-1} \bV )
  = \varphi_0 \, \det( \bI + \bB^{-1} \bV \, \bU \tbX_{,-0}^{-1}) \nonumber \\
  &=& \varphi_0 \, \det\Big( \bI + \bB^{-1} \, ( \tbX \, \bA - \bB \, \tbX_{,-0} ) 
       \, \tbX_{,-0}^{-1} \Big) \nonumber \\
  &=& \varphi_0 \, \det(\bA) \, \det(\bB)^{-1} \, \det(\tbX) \, \det(\tbX)_{,-0}^{-1} 
      \, ,     \label{scalar_dRiem_1_varphi_DT->CH}
\ee
which makes contact with the Cole-Hopf transformation.
\end{remark}

\begin{example}
\label{ex:discr-Riem_multi-kinks_v1} 
Let $m=1$ (scalar case) and $\varphi_0 = \mu \in \bbC \setminus \{0,1\}$, $\bA = \bB = \mu \, \bI$. 
Writing $\bU = (u_1,\ldots,u_n)$ and $\bV = (v_1,\ldots,v_n)^\intercal$, the linear equations 
for $\bU$ and $\bV$ are solved by
\bez
  u_j = a_j \lambda_j^{k_0-1}(1-\lambda_j\mu)^{k_1} \, , \qquad
  v_j = (1-\mu)^{-k_1} \, , \qquad
  a_j, \lambda_j \in \bbC \, , \qquad j=1,\ldots,n \, .
\eez
$\tbX$ is then given by 
\bez
    \tbX 
 = \Big( c_{ij} +  \gamma_j \, \lambda_j^{k_0} \,  
       \Big( \frac{1- \mu \lambda_j}{1-\mu} \Big)^{k_1} \Big)_{i,j=1}^n \, , 
   \qquad  \gamma_j :=  \frac{a_j }{\mu \, (\lambda_j-1) }  \, ,
\eez
where 
$c_{ij}$ are arbitrary constants, which we set to $\delta_{ij}$. We obtain 
\bez
    \det \tbX 
  = 1 + \sum_{i=1}^n \gamma_i \, \lambda_j^{k_0} \, \Big( \frac{1- \mu \, \lambda_j}{1-\mu} \Big)^{k_1} \, ,
\eez
and $\varphi$ given by (\ref{scalar_dRiem_1_varphi_DT->CH}) coincides with $\varphi$ obtained from  
(\ref{scalar_dRiem_1_Phi_of_n-kink}) via the Cole-Hopf transformation. 
\end{example}

\subsection{Hierarchies}
\label{subsec:Riemann_hier}

\subsubsection{Riemann hierarchy} 
Let $\cA$ be the algebra of matrices of real (or complex) smooth functions of 
independent variables $t_k$, $k=0,1,2,\ldots$, and $\lambda$ an arbitrary parameter 
(or an indeterminate). Let
\bez
     \mathrm{d} f = \sum_{k=0}^\infty \lambda^k \, f_{t_k} \, , \qquad   
     \bd f = \sum_{k=0}^\infty \lambda^k \, f_{t_{k+1}} \, ,
\eez
where a subscript means a partial derivative with respect to the corresponding variable. 
(\ref{Riemann_1}) is then equivalent to the matrix Riemann hierarchy
\bez
     \phi_{t_k} = \phi_x \, \phi^k \, , 
\eez
where $k=0,1,2,\ldots$, and we set $x := t_0$. By taking linear combinations, we obtain 
equations of the form $\phi_t = \phi_x \, p(\phi)$, 
where $p(\phi)$ is a polynomial in $\phi$.

\subsubsection{Semi-discrete Riemann hierarchy} 
Let $\cA_0$ be the algebra of matrices of functions of a discrete variable $k$, and smoothly 
dependent on variables $t_j$, $j=1,2,\ldots$. Using the \emph{Miwa shift} operator 
(see, e.g., \cite{DMH06func}, and the references therein)
\bez
    \bbE_\lambda := \exp\Big( \sum_{j \geq 1} \frac{1}{j} \lambda^j \, \pa_{t_j} \Big) \, ,
\eez
we set
\bez
     \mathrm{d} f = [\bbS \, \bbE_\lambda , f ] \, , \qquad
     \bd f = \lambda^{-1} \, [\bbE_\lambda , f ] \, ,
\eez
on $\cA = \cA_0[\bbS, \bbS^{-1}, \bbE_\lambda, \bbE_\lambda^{-1}]$. 
Writing $f_{[\lambda]} := \bbE_\lambda f \bbE_\lambda^{-1}$,  
in terms of $\varphi = \phi \, \bbS$, (\ref{Riemann_1}) takes the form
\be
     \lambda^{-1} \, (\varphi_{[\lambda]} - \varphi) 
     - (\varphi_{[\lambda]}^+ - \varphi) \, \varphi_{[\lambda]} = 0 \, . \label{sd-Riem_hier}
\ee
Expanding in powers of $\lambda$, to zero order we recover 
(\ref{sd_Riemann_eq_1}) with $t=t_1$. By use of it, the next hierarchy 
member can be written in the form
\be
    \varphi_{t_2} = (\varphi^{++} - \varphi) \, \varphi^+ \, \varphi \, .
    \label{sd-Riem_2nd_hier_eq}
\ee
Such a hierarchy apparently first appeared in \cite{LRB83} (where it has been called 
`discrete matrix Burgers hierarchy'). 
In the scalar case, Darboux transformations for this hierarchy have been studied 
in \cite{WGXGQY12}. 
(\ref{CH_Phi_eq}) with $\gamma=0$ takes the form 
\bez
     \lambda^{-1} (\Phi_{[\lambda]} - \Phi) - (\Phi^+_{[\lambda]} 
            - \Phi) \, \varphi_{0,[\lambda]}^+ = 0 \, .
\eez
According to Section~\ref{sec:CH}, together with (\ref{sdRiem1_CH}) this determines a discrete 
Cole-Hopf transformation for the whole hierarchy (\ref{sd-Riem_hier}). 
To zero order in $\lambda$, we have
\bez
      \Phi_{t_1} = (\Phi^+ - \Phi) \, \varphi_0^+  \, ,
\eez
which becomes (\ref{sd_transport_eq}) if $\varphi_0=I$. The next equation, which arises 
at first order in $\lambda$, is 
\bez
      \Phi_{t_2} = - \Phi_{t_1t_1} + 2 \, (\Phi^+ - \Phi) \, \varphi_{0,t_1}^+ 
                   + 2 \, \Phi^+_{t_1} \, \varphi_0^+   
                 = (\Phi^{++} - \Phi^+) \, \varphi_0^{++} \varphi_0^+ 
                   + (\Phi^+ - \Phi) \, [(\varphi_0^+)^2 + \varphi_{0,t_1}^+]  \, ,
\eez
by use of the first equation. For $\varphi_0=I$, this reduces to 
\bez
    \Phi_{t_2} = \Phi^{++} - \Phi \, ,
\eez   
which is the second member of the semi-discrete linear heat hierarchy.

\begin{remark}
\label{rem:sd-Hopf&Burgers}
The nonlinear hierarchy for $\varphi$ contains a 
semi-discretization of a matrix Burgers equation and can thus be regarded as a semi-discrete 
version of a matrix Burgers hierarchy \cite{LRB83}. 
Indeed, a combination of the first two equations of the semi-discrete Riemann hierarchy is
\bez
    \varphi_s := \frac{1}{h^2} (\pa_{t_2} - 2 \pa_{t_1}) \varphi
              = \frac{1}{h^2} \left( (\varphi^{++} - \varphi) \, \varphi^+ \varphi
                - 2 \, (\varphi^+ - \varphi) \, \varphi \right) \, ,
\eez
introducing a lattice spacing $h$ and a new variable $s$. 
In terms of the new dependent variable
\be
     \tilde{\varphi} := \frac{1}{h} (\varphi -I) \, ,   \label{sdRiem_hier_tvarphi}
\ee 
this takes the form
\bez
 \tilde{\varphi}_s = \frac{1}{h^2} ( \tilde{\varphi}^{++} -2 \tilde{\varphi}^+ + \tilde{\varphi})
     (I + h \, \tilde{\varphi}^+ ) ( I + h \, \tilde{\varphi}) 
     + \frac{2}{h} ( \tilde{\varphi}^+ - \tilde{\varphi}) \tilde{\varphi}^+ \, (I+ h \, \tilde{\varphi})
        \, ,
\eez
which formally tends to the Burgers equation 
$\tilde{\varphi}_t = \tilde{\varphi}_{xx} + 2 \, \tilde{\varphi}_x \, \tilde{\varphi}$ 
as $h \to 0$ (also see \cite{Comm+Muse01}).
The corresponding combination of the above first two equations of the semi-discrete linear heat hierarchy is
\bez
    \Phi_s = \frac{1}{h^2} (\Phi^{++} - 2 \, \Phi^+ + \Phi ) \, ,
\eez
which tends to the heat equation as $h \to 0$. Correspondingly, the transformation 
(\ref{sdRiem1_CH}) reads
\bez
     \tilde{\varphi} = \frac{1}{h} (\Phi - \Phi^-) \, (\Phi^-)^{-1} \, ,
\eez
so we also recover the continuous Cole-Hopf transformation as $h \to 0$. 
This limit takes discrete kink solutions to kink solutions of the Burgers equation. 
The observation that the scalar semi-discrete Riemann equation possesses solutions of the 
type we meet in case of the scalar Burgers equation is explained by the fact that they 
extend to solutions of the whole hierarchy. 
\end{remark}

\begin{remark}
\label{rem:r-sdRiemann}
The equations of the semi-discrete Riemann hierarchy can also be obtained in a different way.
Let us consider the following generalizations of (\ref{Riemann_1}),
\be
    \bd \phi - (\mathrm{d} \phi) \, \phi^r = 0   \qquad \quad r \in \bbN \, ,
        \label{r-Riemann}
\ee
and a corresponding generalization of the calculus determined by (\ref{sdRiem_bdc}),
\bez
     \mathrm{d} f = [\bbS^r, f ]  \, , \qquad   \bd f = f_t  \, .
\eez
Setting $\phi = \varphi \, \bbS^{-1}$, where $\varphi$ is a matrix of functions, (\ref{r-Riemann}) 
turns out to be a PDDE for $\varphi$, namely
\bez
     \varphi_t = (\varphi^{(r)} - \varphi) \, (\varphi^{(r-1)} \, \bbS^{-1})^r \, \bbS^r  \, ,
\eez
where $\varphi^{(r)} := \bbS^r \varphi \, \bbS^{-r}$. For $r=1,2$, we recover (\ref{sd_Riemann_eq_1}) 
and (\ref{sd-Riem_2nd_hier_eq}), respectively.
\end{remark}

\subsubsection{Discrete Riemann hierarchy}
Let $\bbS_0,\bbS_1,\bbS_2, \ldots$ be commuting shift operators and
\bez
 \mathrm{d} f = \sum_{i=1}^\infty \lambda^i \, \alpha_i \, [\bbS_i^{-1} \prod_{j=0}^{i-1} \bbS_j , f]
    \, , \qquad
 \bd f = \sum_{i=1}^\infty \lambda^i \, \beta_i \, [\bbS_i^{-1} \prod_{j=1}^{i-1} \bbS_j , f] 
   \, ,
\eez
with constants $\alpha_i, \beta_i$. 
At first and second order in $\lambda$, we obtain from (\ref{Riemann_1}), with 
$ \varphi = \phi \, \bbS_0$, the following equations,
\bez
  \beta_1 \, ( \varphi_{,1} - \varphi ) - \alpha_1 \, ( \varphi_{,1} - \varphi_{,0} ) \, \varphi = 0 \, , 
     \qquad
 \beta_2 \, (\varphi_{,2} - \varphi_{,1}) - \alpha_2 \, (\varphi_{,2} - \varphi_{,0,1}) \, \varphi_{,1} = 0 \, .
\eez
The first coincides with (\ref{discrete_Riemann_eq_v1}) if $\alpha_1=\beta_1 = 1$.
Solving these equations for $\varphi_{,1}$, respectively $\varphi_{,2}$, and assuming the necessary 
invertibility conditions, we have
\be
 &&
 \hspace*{-.8cm} 
 \varphi_{,1} = (\beta_1 - \alpha_1 \, \varphi_{,0}) \, \varphi \, (\beta_1 - \alpha_1 \, \varphi)^{-1} 
      \, , \nonumber  \\
 &&
 \hspace*{-.8cm} 
 \varphi_{,2} = \left( \beta_1 \beta_2 - (\alpha_1 \beta_2  + \alpha_2 \beta_1) \, \varphi
      + \alpha_1 \alpha_2 \, \varphi_{,0} \, \varphi \right)_{,0} \, \varphi \, 
        (\beta_1 \beta_2 - (\alpha_1 \beta_2 + \alpha_2 \beta_1) \, \varphi
     +\alpha_1 \alpha_2 \, \varphi_{,0} \, \varphi)^{-1} .  \qquad  \label{discrHier}
\ee
If $\alpha_i = \beta_i =1$, then we have $\varphi_{,2} = \varphi_{,1,1}$, and a corresponding 
relation holds for all equations of the hierarchy. The $(n+1)$-th flow is then simply the $n$-th 
shift of the first hierarchy equation with respect to its `evolution variable'. 
In this respect the hierarchy is `trivial'.

Setting $\phi_0 = \bbS_0^{-1}$ and 
$\gamma = \sum_{i=1}^\infty \lambda^i ( 1 + \alpha_i - \beta_i ) \, I \, \bbS_i^{-1} \prod_{j=1}^{i-1} \bbS_j$,
according to Section~\ref{sec:CH} a Cole-Hopf transformation is given by
\bez
    \varphi := \phi \, \bbS_0 = \Phi \, \Phi_{,-0}^{-1} \, , \qquad
    \Phi_{,i} = ( \beta_i \, \Phi - \alpha_i \, \Phi_{,0} )_{,1,\ldots,i-1} \qquad i=1,2,\ldots \, ,
\eez
where the linear equations are compatible. Using the first in the second equation, the latter becomes
\be
    \Phi_{,2} = \beta_1 \beta_2 \, \Phi - (\alpha_1 \beta_2 + \alpha_2 \beta_1) \, \Phi_{,0} 
                + \alpha_1 \alpha_2 \, \Phi_{,0,0} \, .  \label{dRiem_2nd_heat_eq}
\ee

Let us choose $\alpha_1 = \alpha_2 = \sqrt{h_2}$, 
$\beta_1 = \sqrt{h_2} + h_0$, $\beta_2 = \sqrt{h_2} - h_0$. In terms of the  
new dependent variable $\tilde{\varphi}$ given by $\varphi = I + h_0 \, \tilde{\varphi}$ 
(cf. (\ref{sdRiem_hier_tvarphi})), the second hierarchy equation takes the form
\be
    -\Delta_2{\tilde{\varphi}} 
  = \Big( h_0 \, (\Delta_0 \Theta) \, \tilde{\varphi} + \Delta_0 \Theta + [\Theta,{\tilde{\varphi}}] \Big)
       (I - h_2 \Theta)^{-1}   \, ,     \qquad
    \Theta := \tilde{\varphi}_{,0} \, \tilde{\varphi} + \Delta_0 \tilde{\varphi} \, ,  \label{dBurgers}
\ee
where
$\Delta_2 \tilde{\varphi} := (\bbS_2 \tilde{\varphi} \bbS_2^{-1} - \tilde{\varphi})/h_2$,
$\Delta_0 \tilde{\varphi} := (\bbS_0 \tilde{\varphi} \bbS_0^{-1} - \tilde{\varphi})/h_0$, 
and correspondingly for $\Delta_0$ acting on $\Theta$. 
In the scalar case, and after replacing
$h_2$ by $-h_2$, the latter equation coincides 
with the \emph{discrete Burgers equation} in \cite{HLW99,HLW00} (up to differences in notation). 
(\ref{dBurgers}) is thus a matrix version of the latter. Here we interprete $h_0$ and $h_2$ 
as lattice spacings. Replacing $h_2$ by $-h_2$ means that $\alpha_i$ 
and $\beta_i$, $i=1,2$, are complex, but the coefficients in the associated discrete heat 
hierarchy equation (\ref{dRiem_2nd_heat_eq}) are real if $h_0,h_2$ are real.

\begin{remark}
\label{rem:r-dRiemann}
Let us consider $\bd \phi - (\mathrm{d} \phi) \, \phi^r = 0$ (which is (\ref{r-Riemann})), with 
\bez
     \mathrm{d} f = [\bbS_r^{-1} \, \bbS_0^r, f ]  \, , \qquad   \bd f = [\bbS_r^{-1} , f ]  \, .
\eez
Setting $\phi = \varphi \, \bbS_0^{-1}$, where $\varphi$ is a matrix of functions,
(\ref{r-Riemann}) becomes a partial difference equation for $\varphi$, namely
\bez
    \varphi_{,r} - \varphi
   = ( \varphi_{,r} - \bbS_0^r \varphi \, \bbS_0^{-r} ) \, (\varphi \, \bbS_{0}^{-1})^r \, \bbS_{0}^r \, .
\eez
For $r=1$, this is (\ref{discrete_Riemann_eq_v1}), i.e., the first of (\ref{discrHier}) if 
$\alpha_1 = \beta_1 = 1$. For $r=2$, this is the second of (\ref{discrHier}) if 
$\alpha_1 = \beta_1 = \beta_2 = 1$ and $\alpha_2 = -1$. 
\end{remark}

\section{Some integrable equations associated with Riemann equations or their integrable discretizations}
\label{sec:Riemann_associates}

\subsection{Self-dual Yang-Mills equation} 
\label{subsec:sdYM}
Let us choose (\ref{Omega_wedge}) with $K=2$ and combine two bidifferential calculi 
of the kind considered in Section~\ref{subsec:Riem} to
\be
    \mathrm{d} f = - f_z \, \xi_1 + f_{\bar{y}} \, \xi_2 \, , \qquad
    \bd f = f_y \, \xi_1 + f_{\bar{z}} \, \xi_2 \, ,   \label{sdYM_bidiff}
\ee
where $\cB$ is the space of smooth complex functions of four (real or complex)  
variables $y, \bar{y}, z, \bar{z}$.  
(\ref{bidiff_conds}) holds and (\ref{phi_eq}) takes the form  
\be
     \phi_{z \bar{z}} + \phi_{y \bar{y}} + [ \phi_z , \phi_{\bar{y}} ] = 0  
              \label{sdYM_Leznov_eq}
\ee
(also see \cite{DMH08bidiff}), which is a well-known potential form of the 
\emph{self-dual Yang-Mills} (sdYM) equation (cf. \cite{Maso+Wood96}). 
Another well-known potential form of the sdYM equation is obtained from (\ref{g_eq}):
\be
      (g_{\bar{z}} \, g^{-1})_z + (g_y \, g^{-1})_{\bar{y}} = 0 \, .  \label{sdYM_Yang}
\ee
In the following subsections, we consider the Riemann system associated with these versions of 
the sdYM equation. Then we derive from 
Corollary~\ref{cor:sol_via_Riem_sys} a method to construct breaking
multi-soliton solutions. Finally, we consider a non-autonomous chiral model 
as an example of a reduction of the sdYM equation, making contact with the work in 
\cite{DKMH11sigma,DMH13SIGMA}.

\subsubsection{The sdYM Riemann system}
Using (\ref{sdYM_bidiff}), (\ref{Riemann_1}) is equivalent to the system
\be
     \phi_y = - \phi_z \, \phi \, , \qquad
     \phi_{\bar{z}} = \phi_{\bar{y}} \, \phi     \label{sdYM_Riem_sys}
\ee
of Riemann equations. Any solution of this system also solves (\ref{sdYM_Leznov_eq}) and, 
if it is invertible, also (\ref{sdYM_Yang}). 
Solutions are implicitly given by
\bez
     \phi = f( y \, \phi - z \, I, \bar{z} \, \phi + \bar{y} \, I ) \, ,
\eez
where $f$ is any analytic function. 
More generally, $f$ may depend in addition on constant $m \times m$ matrices, 
but we have to ensure that $\phi$ commutes with them. 

Let us turn to the `linearization method' of Section~\ref{sec:CH}.
Setting $\gamma=0$ and assuming that $\phi_0$ solves (\ref{sdYM_Riem_sys}) and 
commutes with its partial derivatives, then (\ref{CH_Phi_eq}), elaborated with (\ref{sdYM_bidiff}), 
is solved by 
\bez
  \Phi = A_0 
    + \sum_{i=1}^k A_i \, f_i( y \, \phi_0 - z \, I, \bar{z} \, \phi_0 + \bar{y} \, I ) 
\eez
(cf. \cite{Zenc08sdYM}), with any analytic functions $f_i$ and constant $m \times m$ matrices $A_0,A_i$.
$\phi = \Phi \, \phi_0 \, \Phi^{-1}$ is then a new solution of (\ref{sdYM_Riem_sys}),
and thus of (\ref{sdYM_Leznov_eq}) and also (\ref{sdYM_Yang}) with $g = \phi$.

\begin{remark}
Equation (7) in 
\cite{Zenc08sdYM} is a special case of the linear system (\ref{U,V_eqs}), where the 
matrix $\bP$, respectively $\bQ$, is given by $\lambda \, \bI$, with a spectral parameter 
$\lambda$ (allowed to be a function). 
 Further results in \cite{Zenc08sdYM} are closely 
related to Theorem~\ref{thm:main}, specialized to the sdYM case. 
\end{remark}

\subsubsection{Breaking multi-soliton-type solutions of the sdYM equation}
 From Corollary~\ref{cor:sol_via_Riem_sys}, setting $\balpha = \bbeta = \mathrm{d} \phi_0 = 0$, 
we obtain the following result.
 
\begin{Proposition}
\label{prop:sdYM_multi-break}
Let $\phi_0$ and $g_0$ be constant $m \times m$ matrices, $g_0$ invertible. Let $\bP$ and $\bQ$ 
be solutions of the $n \times n$ matrix Riemann equations
\bez
    \bP_y = - \bP_z \, \bP \, , \quad
    \bP_{\bar{z}} = \bP_{\bar{y}} \, \bP \, , \quad
    \bQ_y = - \bQ \, \bQ_z \, , \quad
    \bQ_{\bar{z}} = \bQ \, \bQ_{\bar{y}} \, ,
\eez
where $\bQ$ is invertible and such that $\sigma(\bP) \cap \sigma(\bQ) = \emptyset$. 
Let $\bX$ be the unique solution of the Sylvester equation
\bez
    \bX \, \bP - \bQ \, \bX = \bV_0 \, \bU_0   \, ,
\eez
where $\bU_0$ is a constant $m \times n$ and $\bV_0$ a constant $m \times n$ matrix. Then 
\bez
    \phi = \phi_0 + \bU_0 \, \bX^{-1} \, \bV_0 \, ,  \quad \mbox{respectively} \quad
    g = ( I + \bU_0 \, (\bQ \, \bX)^{-1} \bV_0 ) \, g_0 \, ,                 
\eez
(except at singular points) solves the respective potential form of the sdYM equation. 
\end{Proposition}

Here we used the fact that 
the differential equation for $\bX$ in Corollary~\ref{cor:sol_via_Riem_sys} is a 
consequence of the Sylvester equation if $\sigma(\bP) \cap \sigma(\bQ) = \emptyset$ holds 
(see the proof of Theorem~2.1 in \cite{DMH13SIGMA}). The above equations for 
$\bP$ form an $n \times n$ version of the sdYM Riemann system considered and solved above. 
The equations for $\bQ$ are obtained by transposition. 
Proposition~\ref{prop:sdYM_multi-break} expresses a nonlinear superposition rule for 
`breaking wave' solutions of the sdYM equation. Also see \cite{Zenc08sdYM}.

\subsubsection{A non-autonomous chiral model in three dimensions}
\label{subsubsec:naCM}
It is well-known that many integrable equations are reductions of the sdYM equation. 
As an example, the reduction condition $\pa_{\bar{z}} = \epsilon \, \pa_z$, 
where $\epsilon = \pm 1$, reduces (\ref{sdYM_bidiff}) to
\bez
    \mathrm{d} f = - f_z \, \xi_1 + f_{\bar{y}} \, \xi_2 \, , \qquad
    \bd f = f_y \, \xi_1 + \epsilon \, f_z \, \xi_2 \, .
\eez
Performing a change of variables $(y,\bar{y}) \mapsto (\rho,\theta)$ via 
$y = \frac{1}{2} \rho \, e^{\theta}$, $\bar{y} = \frac{1}{2} \rho \, e^{-\theta}$, 
we obtain 
\bez
    \mathrm{d} f = - f_z \, \xi_1 + e^\theta \, (f_\rho - \rho^{-1} f_\theta) \, \xi_2 
           \, , \qquad
    \bd f = e^{-\theta} \, (f_\rho + \rho^{-1} f_\theta) \, \xi_1 + \epsilon \, f_z \, \xi_2 \, .
\eez
This is the bidifferential calculus exploited in \cite{DKMH11sigma,DMH13SIGMA}. 
In terms of 
\be
        \varphi := e^{\theta} \, \phi \, ,   \label{naCM_varphi}
\ee
(\ref{phi_eq}) reads 
\be
    \epsilon \, \varphi_{zz} + \Phi_\rho + \rho^{-1} \, \Phi_\theta 
    = [ \Psi , \varphi_z]  \, , \qquad \mbox{where} \quad 
    \Psi := \varphi_\rho + \rho^{-1} (\varphi - \varphi_\theta) \, ,   \label{2+1naCM_phi_eq}
\ee
and (\ref{g_eq}) takes the form 
\be
    (\rho \, g_z \, g^{-1})_z + \epsilon \, (\rho \, g_\rho \, g^{-1})_\rho   
    - [(g_\rho + \rho^{-1} g_\theta ) \, g^{-1}]_\theta + (g_\theta \, g^{-1})_\rho 
     = 0 \, .    \label{2+1naCM_g_eq}
\ee
This is a three-dimensional generalization 
of the \emph{non-autonomous chiral model} that underlies integrable reductions of the vacuum 
Einstein (-Maxwell) equations and to which it reduces if $g$ does not depend on $\theta$ 
(see \cite{DKMH11sigma,DMH13SIGMA} and references cited there).

Now we apply the reduction condition and the change of variables to the sdYM Riemann 
system (\ref{sdYM_Riem_sys}). 
Imposing $\pa_{\bar{z}} = \epsilon \, \pa_z$, (\ref{sdYM_Riem_sys}) becomes
\bez
     \phi_y = - \phi_z \, \phi \, , \qquad
     \epsilon \, \phi_z = \phi_{\bar{y}} \, \phi \, ,
\eez
which is implicitly solved by
\bez
    \phi = f(\epsilon \, \bar{y} \, I + z \, \phi - y \, \phi^2) \, .
\eez
In terms of the independent variables $\rho,\theta,z$, the above system contains explicit 
factors $e^\theta$ (and is thus non-autonomous not only in $\rho$, but also in $\theta$). 
They are eliminated, however, by passing over to $\varphi$ given by (\ref{naCM_varphi}).
Now the Riemann system reads
\bez
  \varphi_\rho + \rho^{-1} (\varphi_\theta - \varphi) + \varphi_z \, \varphi = 0 \, , \qquad
  \epsilon \, \varphi_z - [ \varphi_\rho - \rho^{-1} (\varphi_\theta  - \varphi) ] \, \varphi = 0 \, ,
\eez
and it is implicitly solved by 
\bez
   \varphi = e^{\theta} \, f\Big(  e^{-\theta} \, \frac{\rho}{2} \, (  \epsilon \, I 
               + 2 \, \rho^{-1} z \, \varphi - \varphi^2 ) \Big) \, .
\eez
According to (\ref{Riem_to_phi_eq}), this solves (\ref{2+1naCM_phi_eq}) and $g = e^{\theta} \, \varphi$ 
solves (\ref{2+1naCM_g_eq}).
If we require $\varphi$ (assumed to be invertible) to be independent of $\theta$, this fixes the 
function $f$ to $f(x) = A^{-1} x$, with a constant matrix $A$  
(subject to conditions that ensure that $[A,\varphi]=0$). 
In this case we have  
\bez
    \varphi^2 - 2 \, \rho^{-1} (z \, I - A) \, \varphi - \epsilon \, I = 0 \, ,
\eez
which is a matrix version \cite{DKMH11sigma,DMH13SIGMA} of the `pole trajectories' in the 
Belinski-Zakharov approach to solutions of the integrable reductions of the 
Einstein vacuum equations \cite{Beli+Zakh78}. 
The non-autonomous chiral model is an example, where non-constant solutions 
of the `Riemann equations' (\ref{P,Q_eqs}), which are here in fact Riemann equations, 
in the binary Darboux transformation theorem are crucial in order to recover relevant solutions 
of integrable reductions of Einstein's equations, see \cite{DKMH11sigma,DMH13SIGMA}.

\subsection{A matrix version of the two-dimensional Toda lattice}
Now we compose a bidifferential calculus from two calculi of the kind considered in Section~\ref{subsec:sdRiem}, 
associated with the semi-discrete Riemann equation. We set 
\be
    \mathrm{d} f = [\bbS,f] \, \xi_1 + [\pa_x,f] \, \xi_2 \, , \qquad
    \bd f = [\pa_t,f] \, \xi_1 - [\bbS^{-1} ,f] \, \xi_2 \, ,   \label{Toda_bidiff}
\ee
where $\pa_x := \pa/\pa x$ and $\pa_t := \pa/\pa t$. Setting 
\bez
      \phi = \varphi \, \bbS^{-1} \, ,
\eez
with a matrix of functions $\varphi$, (\ref{Riemann_1}) takes the form
\be
    \varphi_t = (\varphi^+ - \varphi) \, \varphi \, , \qquad
    \varphi - \varphi^- = \varphi_x \, \varphi^- \, .   \label{Toda_Riem_sys}
\ee
The first equation is (\ref{sd_Riemann_eq_1}), the second is obtained from (\ref{sd_Riemann_eq_1})  
via $\varphi \mapsto \varphi^{-1}$, $t \mapsto -x$, and a reflection of the discrete variable.  
(\ref{phi_eq}) takes the form
\be
   \varphi_{tx} = (\varphi^+ - \varphi)(I + \varphi_x) 
                   - (I + \varphi_x)(\varphi - \varphi^-)  
         \label{Toda:2d_Toda_eq}          
\ee
(also see \cite{DMH08bidiff}). This may be regarded as a matrix version  
of the \emph{two-dimensional Toda lattice equation}. Indeed, in the commutative case, setting $V := \varphi_x$ and 
differentiating once with respect to $t$, (\ref{Toda:2d_Toda_eq}) becomes (cf. \cite{DMH08bidiff})
the well-known equation $(\log(1+V))_{tx} = V^+ - 2V + V^-$ \cite{Mikh79,HIK88}. 
The Miura-dual equation (\ref{g_eq}), with $g=\phi$, reads
\bez  
      (\varphi_t \, \varphi^{-1})_x - \varphi^+ \, \varphi^{-1} + \varphi \, (\varphi^{-1})^-  = 0 \, ,
\eez
which appeared in \cite{Li+Nimm08}. 
(\ref{Toda_Riem_sys}) is the `Riemann system' associated with (\ref{Toda:2d_Toda_eq}) and the 
last equation.

\subsubsection{Cole-Hopf transformation for the `Riemann system'}
Choosing $\phi_0 = \bbS^{-1}$ in (\ref{phi_transform}), and 
$\gamma = A \, \xi_1 + B \, \bbS^{-1} \, \xi_2$ in (\ref{CH_Phi_eq}), 
where $A$ and $B$ are constant, 
we obtain the following Cole-Hopf transformation for (\ref{Toda_Riem_sys}),
\be
    \varphi = \Phi \, (\Phi^-)^{-1} \, , \qquad
    \Phi_t = \Phi^+ - \Phi + \Phi \, A \, , \qquad 
    \Phi_x = \Phi - \Phi^- - \Phi \, B \, .    \label{CH_for_Toda-Riem}
\ee
The second equation is (\ref{sd_transport_eq}). The last two equations are compatible. 
Without restriction of generality we can set $A = B = 0$, since $A$ and $B$ can be eliminated by 
a transformation of $\Phi$ that preserves the expression for $\varphi$. 

\begin{example}
A class of solutions of (\ref{Toda_Riem_sys}) is determined by 
\bez
    \Phi = I + \sum_{i=1}^n \bA_i \, e^{\bTh_i} \, \bB_i \, , \qquad
    \bTh_i = \bLa_i \, k + (I-e^{-\bLa_i}) \, x + (e^{\bLa_i}-I) \, t \, ,
\eez
where $\bA_i,\bB_i$ are constant $m \times n$, respectively $n \times m$ matrices, and $\bLa_i$ 
are constant $n \times n$ matrices. In the scalar case ($m=1$), and if 
$\bLa = \mathrm{diag}(\lambda_1,\ldots,\lambda_n)$, this takes the form
\bez
    \Phi = 1 + \sum_{i=1}^n e^{\lambda_i \, k + (1- e^{-\lambda_i}) \, x 
                  + (e^{\lambda_i} -1) \, t + \gamma_j} \, ,
\eez
with constants $\gamma_i$. If all constants are real, (\ref{CH_for_Toda-Riem}) determines 
an $n$-kink solution of the two-dimensional Toda lattice equation. 
\end{example}

\subsubsection{Darboux Transformations for the `Riemann system'} 
Let 
\bez 
 \balpha =  \bA \, \xi_1 + \mathbb{S}^{-1} \, \xi_2 \, , \quad  
 \bbeta =  -\bB \, \xi_1 - \mathbb{S}^{-1} \, \xi_2 \, , \quad
 \bP = \bA \, \mathbb{S}^{-1} \, , \quad 
 \bQ = \bB \, \mathbb{S}^{-1} \, , \quad 
 \bX = \mathbb{S} \, \tbX \, , 
\eez
with invertible constant matrices $\bA, \bB$. Then Corollary~\ref{cor:main} yields the 
following system of equations,
\bez
 && \bU_t = \bU^+ \, \bA + (\varphi_0^+ - \varphi_0) \, \bU \, , \quad 
 \bU_x = - ( I + \varphi_{0,x} ) \, \bU^- \, \bA^{-1} \, , \\
 && \bV_t = - \bV \, \varphi_0^+ \, , \quad
    \bV_x = \bB^{-1} \bV^+ \, (I + \varphi_{0,x}^+) \, , \quad
    \bV^- = \bB^{-1} \bV \, \varphi_0 \, ,
\eez
and 
\bez
   \tbX^+ = ( \bV \, \bU + \bB \tbX ) \, \bA^{-1} \, ,  \quad
   \tbX_t = \bV^- \, \bU \, , \quad 
   \tbX_x = \bV^-_x \, \bU^- \, \bA^{-1} \, .
\eez
The latter system is compatible, which allows us to write 
\bez
  \tbX = \boldsymbol{C^-} + \int_{(0,0)}^{(t,x)} \Big( \bV^- \, \bU \, dt + \bV_x^- \, \bU^- \, \bA^{-1} \, dx 
          \Big) \, ,
\eez
where the integral does not depend on the path of integration. Here $\bC$ does not depend on $x$ and $t$ and 
has to satisfy the constraint $\bC \bA - \bB \bC^- = \bV \, \bU |_{t=0,x=0}$.
If $\varphi_0$ solves the `Riemann system' (\ref{Toda_Riem_sys}), and if $\bU, \bV$ satisfy the above 
linear equations, a new solution of (\ref{Toda_Riem_sys}) is given by 
\bez
   \varphi = \varphi_0 + \boldsymbol{U} \, \tbX^{-1} \, \bV^- \, .
\eez 

Darboux transformations for the two-dimensional Toda lattice equation appeared in \cite{Nimm+Will97,Li+Nimm08}.
A very compact version is obtained from Theorem~\ref{thm:main}, using the bidifferential calculus 
determined by (\ref{Toda_bidiff}).

\subsection{A matrix version of Hirota's bilinear difference equation} 
\label{subsec:Hirota}
Here we compose calculi of the kind considered in Section~\ref{subsec:discr_Riem}. 
In the following, $c_i$ are constants and $\bbS_0, \bbS_i : \cA_0 \rightarrow \cA_0$ commuting 
shift operators,  
where $\cA_0$ is the algebra of matrices of functions of corresponding discrete variables. 
Let $\cA = \cA_0[\bbS_0^{\pm 1}, \bbS_1^{\pm 1}, \ldots, \bbS_K^{\pm 1}]$. We 
use the notation introduced in Section~\ref{subsec:discr_Riem}:  
$f_{,0} := \bbS_0 \, f \, \bbS_0^{-1}$ and $f_{,i} := \bbS_i \, f \, \bbS_i^{-1}$. 

\subsubsection{First version}
\label{subsec:Hirota_1}
Let
\bez
   \mathrm{d} f = \sum_{i=1}^K [ \bbS_i^{-1} \bbS_0 , f] \, \xi_i \, , \qquad
   \bd f = \sum_{i=1}^K c_i^{-1} \, [ \bbS_i^{-1} , f] \, \xi_i \, .
\eez
Then, setting 
\bez
      \phi = \varphi \, \bbS_0^{-1} \, ,   
\eez
with a matrix of functions $\varphi$, (\ref{Riemann_1}) reads
\be
    \Delta_i \varphi = ( \varphi_{,i} - \varphi_{,0} ) \, \varphi \, , \qquad \quad i=1,\ldots,K \, ,
            \label{Hirota_1_Riem_sys}
\ee
where 
\bez
     \Delta_i \varphi := c_i^{-1} \, ( \varphi_{,i} - \varphi ) \, .
\eez     
(\ref{Hirota_1_Riem_sys}) is a system of discrete matrix Riemann equations (cf. (\ref{discrete_Riemann_eq_v1})). 
(\ref{phi_eq}) takes the form
\be
   (\Delta_i \varphi)_{,0} - (\Delta_i \varphi)_{,j} 
     + ( \varphi_{,i} - \varphi_{,0} )_{,j} \, ( \varphi_{,j} - \varphi_{,0} )
 = (\Delta_j \varphi)_{,0} - (\Delta_j \varphi)_{,i} 
   + ( \varphi_{,j} - \varphi_{,0} )_{,i} \, ( \varphi_{,i} - \varphi_{,0} )
    \, ,    \label{phi_Hirota_1}
\ee
where $i,j = 1,\ldots,K$, $i \neq j$. The Miura-dual (\ref{g_eq}) with $g = \phi$ takes the form
\bez
     (\Delta_i \varphi)_{,0} \, \varphi_{,0}^{-1} - (\Delta_i \varphi)_{,j} \, \varphi_{,j}^{-1} 
   = (\Delta_j \varphi)_{,0} \, \varphi_{,0}^{-1} - (\Delta_j \varphi)_{,i} \, \varphi_{,i}^{-1}  
     \qquad i,j=1,\ldots,K \, , \quad i \neq j \, ,
\eez
which is
\be
   c_i^{-1} \, [ (\varphi_{,i} \, \varphi^{-1})_{,0} - (\varphi_{,i} \, \varphi^{-1})_{,j} ]
   = c_j^{-1} \, [ (\varphi_{,j} \, \varphi^{-1})_{,0} - (\varphi_{,j} \, \varphi^{-1})_{,i} ] \, . 
             \label{Hirota_1_dual}
\ee
(\ref{Hirota_1_Riem_sys}) is the `Riemann system' associated with (\ref{phi_Hirota_1}) and (\ref{Hirota_1_dual}).

\paragraph{Cole-Hopf transformation.}
Choosing $\phi_0 = \bbS_0^{-1}$ in (\ref{phi_transform}), and 
$\gamma= - \sum_i A_i \, \, \bbS_i^{-1} \, \xi_i$, with constant matrices $A_i$, 
in (\ref{CH_Phi_eq}), we obtain the following Cole-Hopf transformation for (\ref{Hirota_1_Riem_sys}),
\bez
    \varphi = \Phi \, \Phi_{,-0}^{-1} \, , \qquad
    \Phi_{,0} - c_i^{-1} \, \Phi = \Phi_{,i} \, (A_i - (c_i^{-1}-1) \, I) \qquad i=1,\ldots,K \, . 
\eez
If $c_i=1$, these equations are of the form (\ref{dRiem_1_CH}).  
Choosing $A_i = ( c_i^{-1} - 2 ) \, I$, a set of solutions of the linear equations is given by 
\bez
    \Phi = \sum_{i=1}^N \bA_i \, \bLa_i^{k_0} \, \prod_{j=1}^K (c_i^{-1} I - \bLa_i)^{k_j} \, \bB_i  \, , 
\eez
where $\bA_i$ and $\bB_i$ are constant $m \times n$, respectively $n \times m$ matrices, and the $\bLa_i$ 
are constant $n \times n$ matrices.

\paragraph{Darboux transformations for the `Riemann system'.} Let
\bez
 \bP = \bA \mathbb{S}_0^{-1} \, , \quad 
 \bQ = \bB\mathbb{S}_0^{-1} \, , \quad 
 \balpha = \sum_{i = 1}^K \bA \, \mathbb{S}_i^{-1} \, \xi_i \, , \quad 
 \bbeta = -\sum_{i = 1}^K \bB \, \mathbb{S}_i^{-1} \, \xi_i \, ,\quad
 \tbX = \bX \, \mathbb{S}_0 \, , 
\eez
with constant $n \times n$ matrices $\bA,\bB$. 
Then (\ref{P,Q_eqs})-(\ref{dX_eq}) result in 
\bez
 \Delta_i \bU + \bU_{,0} \, \bA  = c_i \, \Delta_i \varphi_0 \, \bU \, , \quad 
 \Delta_i \bV = \bV_{,i} \, \varphi_{0,0} \, , \quad 
 \bB \, \bV_{,-0} = \bV \, \varphi_0 \, ,
\eez
and
\bez
 \Delta_i \tbX + \bV_{,i} \, \bU_{,0} = 0 \, , \quad 
 \tbX \, \bA - \bB \, \tbX_{,-0} = \bV \, \bU \, .  
\eez
The latter equations are compatible, hence
\bez
 \tbX(k_0,k_1, \dots, k_K) &=& \bF(k_0) - \sum_{i = 1}^K c_i \sum_{j_i = 0}^{k_i-1} 
      \bV(k_0,k_1, \ldots, k_{i-1}, j_i+1, 0, \ldots, 0) \, \times \\
     && \bU(k_0+1,k_1, \ldots,k_{i-1},j_i,0,\ldots, 0) \, , 
\eez
where $\bF(k_0)$ satisfies 
\bez
 \bF(k_0) \, \bA - \bB \, \bF(k_0-1) - \bV(k_0,0,\ldots,0) \, \bU(k_0,0,\ldots,0) = 0 \, . 
\eez
If $\varphi_0$ solves (\ref{Hirota_1_Riem_sys}), then also
\bez
    \varphi = \varphi_0 + \bU (\tbX^{-1} \bV)_{,-0}  \, .
\eez

\subsubsection{Second version}
Here we exchange $\mathrm{d}$ and $\bd$ in the first version:
\bez
  \mathrm{d} f = \sum_{i=1}^K c_i^{-1} \, [ \bbS_i^{-1} , f] \, \xi_i \, , \qquad
  \bd f = \sum_{i=1}^K [ \bbS_i^{-1} \bbS_0 , f] \, \xi_i \, .
\eez
In terms of $\varphi = \phi \, \bbS_0^{-1}$, taken to be a matrix of functions, (\ref{Riemann_1}) becomes
\bez
    \varphi_{,i} - \varphi_{,0} = (\Delta_i \varphi) \, \varphi_{,0} 
     \qquad i=1,\ldots,K \, .   
\eez
This system is simply obtained from (\ref{Hirota_1_Riem_sys}) with $\varphi$ replaced 
by $\varphi^{-1}$, which is a special case of (\ref{d,bd_exchange_Riem_sym}). 
The integrability condition (\ref{phi_eq}) takes the form
\be
     (I+\Delta_i \varphi)_{,j} \, (I+\Delta_j \varphi)_{,0} 
   = (I+\Delta_j \varphi)_{,i} \, (I+\Delta_i \varphi)_{,0} 
     \qquad i,j=1,\ldots,K, \quad i \neq j \, .      \label{phi_Hirota_2}
\ee
The Miura-dual, with $g = \phi = \varphi \, \bbS_0$, is
\be
     \Delta_j( \varphi_{,i} \, \varphi_{,0}^{-1} ) = \Delta_i( \varphi_{,j} \, \varphi_{,0}^{-1} )
     \qquad i,j=1,\ldots,K \, , \quad  i \neq j \, .       \label{Hirota_2_dual}
\ee
Via $\varphi \mapsto \varphi^{-1}$, (\ref{Hirota_2_dual}) becomes (\ref{Hirota_1_dual}), also see 
Section~\ref{sec:conclusion}.
There is no such relation between (\ref{phi_Hirota_1}) and (\ref{phi_Hirota_2}).

For $K=2$, (\ref{Hirota_2_dual}) (or (\ref{Hirota_1_dual})) can be regarded as a matrix 
version of Hirota's bilinear difference equation \cite{Hiro81,Miwa82}, as explained in the 
following remark which we owe to Aristophanes Dimakis.
This matrix version of Hirota's bilinear difference equation is different from the 
`noncommutative Hirota-Miwa equations' in \cite{Nimm06}.

\begin{remark}
In the scalar case ($m=1$), setting
\bez
       \varphi = \frac{\tau_{,-0}}{\tau}  \, ,  
\eez
(\ref{Hirota_2_dual}) becomes
\bez
    c_i \, \Big( \frac{\tau_{,i,j,-0} \, \tau_{,j,0}}{\tau_{,i,j} \, \tau_j}  
       - \frac{\tau_{,i,-0} \, \tau_{,0}}{\tau_{,i} \, \tau} \Big) 
  = c_j \, \Big( \frac{\tau_{,i,j,-0} \, \tau_{,i,0}}{\tau_{,i,j} \, \tau_i}  
       - \frac{\tau_{,j,-0} \, \tau_{,0}}{\tau_{,j} \, \tau} \Big) \, ,  
\eez
which is
\bez
   \Big( \frac{ c_i \, \tau_{,i} \, \tau_{,j,0} - c_j \, \tau_{,j} \, \tau_{,i,0} }{ \tau_{,i,j} \tau_{,0}} 
      \Big)_{,-0} 
 = \frac{ c_i \, \tau_{,i} \, \tau_{,j,0} - c_j \, \tau_{,j} \, \tau_{,i,0} }{ \tau_{,i,j} \, \tau_{,0}}  \, .
\eez
Hence, we obtain 
\bez
     c_i \, \tau_{,i} \,  \tau_{,j,0} - c_j \, \tau_{,j} \, \tau_{,i,0} 
   = c_{ij} \, \tau_{,i,j} \, \tau_{,0} \, , 
\eez
with new arbitrary constants $c_{ij} = - c_{ji}$. For $K=2$, this is Hirota's bilinear difference 
(or Hirota-Miwa) equation \cite{Hiro81,Miwa82}.
\end{remark}

\section{ (2+1)-dimensional matrix Nonlinear Schr\"odinger system}
\label{sec:NLS}
The Nonlinear Schr\"odinger equation in 2+1 dimensions can be treated as a reduction of the sdYM equation (cf. 
Section~\ref{subsec:sdYM}), but here we will take a more direct approach. 
Let $\cB$ be the space of smooth complex functions of independent variables $x,y,t$. 
Let $J \neq I$ be an invertible constant $m \times m$ matrix. 
We consider two calculi on $\mathrm{Mat}(m,m,\cB)$. \\
\textbf{(1)} 
Let
\bez
    \mathrm{d} f = f_y \, , \qquad
    \bd f = -\imag \, f_t \, .
\eez
Then (\ref{Riemann_1}) 
is the matrix Riemann equation
\be
    \imag \, \phi_t + \phi_y \, \phi = 0 \, .    \label{2+1NLS_Riem_1}
\ee
\textbf{(2)} Let
\bez
    \mathrm{d} f = \frac{1}{2} [J,f] \, , \qquad
    \bd f = f_x \, . 
\eez
In this case, (\ref{Riemann_1}) 
is the ordinary differential equation
\be
   \phi_x - \frac{1}{2} [J,\phi] \, \phi = 0 \, .     \label{2+1NLS_Riem_2}
\ee

Now we combine the two calculi to
\be
    \mathrm{d} f = f_y \, \xi_1 + \frac{1}{2} [J,f] \, \xi_2 \, , \qquad
    \bd f = -\imag \, f_t \, \xi_1 + f_x \, \xi_2     \label{2+1NLS_bidiff}
\ee
(also see \cite{DMH01FK}). Then (\ref{Riemann_1})
consists of the pair of equations given above. The integrability condition (\ref{phi_eq}) takes the form 
\bez
  - \frac{\imag}{2} [ J, \phi_t ] = \phi_{xy} + \frac{1}{2} [ \phi_y , [J,\phi]] \, .
\eez
Writing
\bez
    \phi = J \, \varphi \, ,    
\eez
it becomes 
\be
  - \frac{\imag}{2} [ J, \varphi_t ] = \varphi_{xy} 
  + \frac{1}{2} \, \Big( \varphi_y \, J \, [J,\varphi] 
        - [J,\varphi] \, J \, \varphi_y \Big) \, .        \label{pre_matrix2+1NLS}
\ee   
 From now on we set  
\be
    J := J_{(m_1,m_2)} := \left(\begin{array}{cc} I_{m_1} & 0 \\ 0 & -I_{m_2} \end{array}\right) \, , 
       \qquad
    \varphi =: \left(\begin{array}{cc} u & q \\ r & v \end{array} \right) \, ,
        \label{J_varphi_decomp}
\ee
with $m=m_1+m_2$. Then (\ref{pre_matrix2+1NLS}) splits into
\be
   -\imag \, q_t = q_{xy} + q \, v_y + u_y \, q \, , \quad 
   \imag \, r_t = r_{xy} + r \, u_y + v_y \, r  \, , \quad
    u_x = - q \, r \, , \quad v_x = - r \, q \, .   \label{matrix2+1AKNS}
\ee
Since the last two equations arise via an integration with respect to $y$, we should have added 
arbitrary matrices that do not depend on $y$. Since they do not influence the first two equations, they 
can be dropped, respectively set to zero. 
Imposing $\pa_x = \pa_y$, (\ref{matrix2+1AKNS}) reduces 
to the matrix \emph{NLS system} \cite{Ford+Kuli83,APT04book}, also see the references 
in \cite{DMH10NLS}.

Now we consider the \emph{reduction conditions}
\be
      \varphi^\dagger = \left\{ \begin{array}{l@{\quad}l}
                          \varphi + C  & \mbox{defocusing case}  \\
                        J \, \varphi \, J + C & \mbox{focusing case} 
                        \end{array} \right.     \label{NLS_hc_red_conditions}
\ee
where $C$ is an anti-Hermitian matrix commuting with $J$. 
This generalizes the `naive reductions' where $C=0$. 
The latter would exclude the solutions presented in Section~\ref{subsec:Bogo} below.  
Writing $C = \mbox{block-diag}(c_1,c_2)$ 
and using (\ref{J_varphi_decomp}), these conditions result in
\be
     r = \pm \, q^\dagger \, , \qquad    
     u^\dagger = u + c_1 \, , \qquad   v^\dagger = v + c_2  \, ,   \label{NLS_hc_red_conditions_q,r,u,v}   
\ee
where $c_i^\dagger = - c_i$, $i=1,2$.  Compatibility with (\ref{matrix2+1AKNS}) implies
\be
    c_{1,x} = c_{2,x} = 0 \, , \qquad   c_{1,y} \, q + q \, c_{2,y} = 0 \, .
       \label{2+1NLS_reduction_conditions_for_C}
\ee
Then (\ref{matrix2+1AKNS}) reduces to the matrix version 
\be
   -\imag \, q_t = q_{xy} + q \, v_y + u_y \, q \, , \quad 
    u_x = \mp q \, q^\dagger \, , \quad v_x = \mp q^\dagger \, q \, ,   \label{matrix2+1NLS}
\ee
of the \emph{(2+1)-dimensional NLS equation} \cite{Calo+Dega76NCB32,Zakh80} 
\be
    \imag \, q_t + q_{xy} \mp 2 \, q \, \Big( \int^x |q|^2 \, dx \Big)_y = 0  \, ,
       \label{2+1NLS}
\ee
which we obtain from (\ref{matrix2+1NLS}) in the scalar case ($m_1=m_2=1$). 
The upper (lower) choice of sign corresponds to the defocusing (focusing) NLS case. 
The matrix version probably first appeared in \cite{Stra92NLS}.

\subsection{`Riemann system' associated with the (2+1)-dimensional matrix NLS system}
\label{subsec:2+1NLS_reduction_of_Riem}
In terms of $\varphi$, the associated `Riemann system', given by (\ref{2+1NLS_Riem_1}) and 
(\ref{2+1NLS_Riem_2}), reads
\be
      \imag \, \varphi_t + \varphi_y \, J \, \varphi = 0 \, , \qquad 
      \varphi_x - \frac{1}{2} ( J \, \varphi \, J - \varphi ) \, \varphi = 0 \, ,
                  \label{2+1NLS_Riem_sys}
\ee
which decomposes into 
\be
 &&  \imag \, u_t + u_y \, u - q_y \, r  = 0 \, , \quad    
   \imag \, v_t - v_y \, v  + r_y \, q = 0 \, , \quad
   u_x + q \, r = 0 \, , \quad v_x + r \, q = 0 \, , \nonumber \\
 &&  \imag \, q_t + u_y \, q - q_y \, v  = 0 \, , \quad 
   \imag \, r_t + r_y \, u - v_y \, r = 0 \, , \quad
   q_x + q \, v = 0 \, , \quad r_x + r \, u = 0  \, .   \label{NLS_Riem_1_decomposed}
\ee
The first two equations can be dropped, since they are a consequence of the others    
if $q,r \neq 0$. The third and the fourth equation are just the last two equations 
in (\ref{matrix2+1AKNS}). The last four equations imply the first two equations in (\ref{matrix2+1AKNS}). 

The reduction conditions (\ref{NLS_hc_red_conditions}), respectively (\ref{NLS_hc_red_conditions_q,r,u,v}), 
imposed on the `Riemann system' (\ref{NLS_Riem_1_decomposed}), reduce it to
\bez
  && \imag \, q_t + u_y \, q - q_y \, v  = 0 \, , \qquad
   \imag \, u_t + u_y \, u \mp q_y \, q^\dagger  = 0 \, , \qquad 
   \imag \, v_t - v_y \, v  \pm q_y^\dagger  \, q = 0 \, , \\
  && q \, v = (u + c_1) \, q \, , \qquad
   q_x = - q \, v \, , \qquad
   u_x = \mp q \, q^\dagger \, , \qquad
   v_x = \mp q^\dagger q \, , \qquad
   c_{1,y} \, q + q \, c_{2,y} = 0  \, .
\eez
We already met the last equation in (\ref{2+1NLS_reduction_conditions_for_C}).
Recall that $u^\dagger = u + c_1$ and $v^\dagger = v + c_2$. 
If $m_1=m_2$ ($q,u,v$ are then square matrices of same size), and assuming invertibility of $q$,
the above system reduces to 
\be
   u = - q_x \, q^{-1} - c_1 \, , \quad 
   v = - q^{-1} \, q_x \, , \quad
   (q_x \, q^{-1})^\dagger = q_x \, q^{-1} + c_1 \, , \quad
   ( q^{-1} \, q_x)^\dagger =  q^{-1} \, q_x - c_2 \, , \label{NLS-Riemann_sys_u,v,q}
\ee
and
\be
  \imag \, q_t - q_{xy} + q_x \, q^{-1} \, q_y + q_y \, q^{-1} \, q_x 
   - c_{1,y} \, q = 0 \, , \quad
  (q_x \, q^{-1})_x \mp q \, q^\dagger = 0 \, , \quad
  (q^{-1} q_x)_x \mp q^\dagger \, q = 0 \, .
         \label{NLS-Riemann_sys_q}
\ee

\subsubsection{`Riemann system' associated with the scalar (2+1)-dimensional NLS equation}
\label{subsec:Bogo}
In the scalar case, i.e., $m_1=m_2=1$, the last two equations of (\ref{NLS-Riemann_sys_u,v,q}) 
require $c_2 = -c_1$. Writing 
$q = \sqrt{\rho} \, e^{\imag \, S}$, with real functions $\rho$ and $S$, the second equation 
of (\ref{NLS-Riemann_sys_q}) becomes $\log(\rho)_{xx} \mp 2 \, \rho = 0$ and $S_{xx} = 0$, 
hence $S = \lambda \, x + \beta$, with real functions $\lambda,\beta$ not depending on $x$.
 We obtain
\bez
    q = \left\{ \begin{array}{l} a  \, \mathrm{sec}(a \, x + b) \, e^{\imag \, (\lambda \, x + \beta) } \\
         a  \, \mathrm{sech}(a \, x + b) \, e^{\imag \, (\lambda \, x + \beta) } \end{array} \right. 
         \quad \begin{array}{l} \mbox{defocusing case (upper sign)} \\  
         \mbox{focusing case (lower sign)} \end{array}  
\eez
where $a,b$ are real functions that do not depend on $x$. 
We find $c_1 = -2 \imag \, \lambda$. 
Inserting the expression for $q$ in the first of (\ref{NLS-Riemann_sys_q}), leads to
\be
   \lambda_t + \lambda \, \lambda_y \pm a \, a_y = 0 \, , \qquad
   a_t + (\lambda \, a)_y = 0 \, ,  \label{2+1NLS_Bogo1}
\ee
and the linear system
\be
    b_t +  \lambda \, b_y + a \, \beta_y = 0 \, , \qquad
    \beta_t + \lambda \, \beta_y \pm a \, b_y = 0 \, .  \label{2+1NLS_Bogo2}
\ee

In the defocusing case (upper sign), the first pair of equations decouples into the 
two Riemann equations
\bez
    w^i_t + w^i \, w^i_y = 0 \qquad  i=1,2 \, , \qquad \mbox{where} \quad
    w^1 := \lambda + a \, , \quad  w^2 := \lambda - a \, .
\eez
With the reduction $a = \lambda$ (i.e., $w^2=0$), the function $\lambda$ satisfies 
the Riemann equation $\lambda_t + 2 \, \lambda \, \lambda_y = 0$ and we recover (with $b=\beta$) 
a case treated in \cite{Bogo91III} (see \S 5 therein), where `breaking soliton' 
solutions of (\ref{2+1NLS}) were searched for. We thus showed that this case is a subcase of what 
is covered by the `Riemann system' associated with the (2+1)-dimensional NLS equation. 

In the focusing case, (\ref{2+1NLS_Bogo1}) can be expressed as the complex Riemann equation
$w_t + w \, w_y = 0$, where $w := \lambda + \imag \, a$. In terms of $\zeta := b - \imag \, \beta$, 
(\ref{2+1NLS_Bogo2}) (with the lower sign) reads $\zeta_t + w \, \zeta_y = 0$, which is solved by 
$\zeta = f(w)$, with an arbitrary differentiable function $f$, as a consequence of the Riemann equation 
for $w$. A hodograph transformation (with $x(\lambda,a)$, $t(\lambda,a)$) turns (\ref{2+1NLS_Bogo1}) 
into the linear equations
\bez
     y_\lambda + a \, t_a - \lambda \, t_\lambda = 0 \, , \qquad
     y_a - \lambda \, t_a - a \, t_\lambda  = 0 \, ,
\eez
if $y_\lambda \, t_a - y_a \, t_\lambda \neq 0$ (also see, e.g., \cite{CCCF05}).
If $t_\lambda = 0$, then $y = y_0 + \lambda \, a_0/a $ and $t = t_0 + a_0/a$, with arbitrary 
constants $a_0,t_0,y_0 \in \bbR$, hence
\bez
    w = \frac{y-y_0}{t-t_0} + \imag \, \frac{a_0}{t-t_0} \, .
\eez
With this special explicit solution of the complex Riemann equation, we recover 
the singular solitary wave solution, of the scalar focusing (2+1)-dimensional NLS equation, that appeared in 
\cite{Jian+Bull94}. 
If we choose a `breaking wave' solution of the complex Riemann equation, the function $a$ inherits 
the singularity of $\lambda_y$ and thus the solution $q$ of the (2+1)-dimensional NLS equation 
itself is singular (and not just a partial derivative of it). Therefore, such solutions are 
\emph{not} breaking waves. 

If $\lambda$ is constant, the above system admits (regular) solitary wave solutions. From the above, 
we can conclude that the slightest generic perturbation will lead to a solution that breaks or blows up 
in finite (positive or negative) time $t$. 
This feature is absent in the (1+1)-dimensional NLS equation.

\subsubsection{The linearization method}
Let $\phi_0 = \imag \, \Lambda$ and $\gamma = \frac{\imag}{2} J \Lambda \, \xi_2$. 
We assume that $\Lambda$ does not depend on $x$ and $[\Lambda,J]=0$. 
Then (\ref{CH_Phi_eq}) with (\ref{2+1NLS_bidiff}) reads
\be
    \Phi_t + \Phi_y \, \Lambda = 0 \, , \qquad
    \Phi_x = \frac{\imag}{2} \, J \, \Phi \, \Lambda \, ,  \label{2+1NLS_Riem_lin_sys}
\ee
and (\ref{Riemann_1_mod}) is the matrix Riemann equation
\be
     \Lambda_t + \Lambda_y \, \Lambda = 0 \, .    \label{2+1NLS_Lambda_eq}
\ee
 From the general argument in Section~\ref{sec:CH} it follows that, for any solution $\Phi$ of 
 (\ref{2+1NLS_Riem_lin_sys}), 
\be
    \varphi = \imag \, J \Phi  \, \Lambda \, \Phi^{-1}    
     \label{2+1NLS_CH_transf}
\ee
is a solution of the `Riemann system' (\ref{2+1NLS_Riem_sys}). 
The general solution of the second of (\ref{2+1NLS_Riem_lin_sys}) is
\bez
    \Phi = A_1 \, e^{\frac{1}{2} \, \imag \, J \Lambda \, x}
           + A_2 \, e^{- \frac{1}{2} \, \imag \, J \Lambda \, x} \, ,
\eez
where the matrices $A_i$ do not depend on $x$ and satisfy $J A_1 = A_1 J$ and $J A_2 = - A_2 J$.
The first of (\ref{2+1NLS_Riem_lin_sys}) now results in 
\be
    A_{i,t} + A_{i,y} \, \Lambda = 0 \qquad \quad  i=1,2 \, .  \label{2+1NLS_A_eqs}
\ee
The reduction conditions (\ref{NLS_hc_red_conditions}) are translated as follows, 
\be
    \varphi^\dagger = \varphi + C \quad &\Longleftrightarrow& \quad
         \Lambda^\dagger \Phi^\dagger J \Phi + \Phi^\dagger J \Phi \Lambda 
           = \imag \, \Phi^\dagger C \Phi \, , 
              \label{2+1NLS_Phi_red_cond_defoc} \\           
    \varphi^\dagger = J \varphi J + C \quad &\Longleftrightarrow & \quad 
         \Lambda^\dagger \Phi^\dagger \Phi + \Phi^\dagger \Phi \Lambda = \imag \, \Phi^\dagger C J \Phi 
      \, .   \label{2+1NLS_Phi_red_cond_foc}
\ee
We note that these conditions are \emph{non}linear.   
Inserting the formula for $\Phi$, we obtain
\bez
    \Gamma_2 \mp  e^{-\imag J\Lambda^{\dagger} x} \, \Gamma_1 \, e^{\imag J \Lambda x} = 0 \, , \qquad 
    \Xi^{\dagger} \mp e^{-\imag J\Lambda^{\dagger} x} \, \Xi \, e^{-\imag J \Lambda x} = 0 \, ,
\eez
where the upper sign refers to the defocusing case (\ref{2+1NLS_Phi_red_cond_defoc}), the lower sign to the 
focusing case (\ref{2+1NLS_Phi_red_cond_foc}), and we introduced the $x$-independent expressions
\bez
   \Gamma_j := \Lambda^{\dagger} \, A_j^{\dagger} A_j + A_j^{\dagger} A_j \, \Lambda -\imag A_j^{\dagger}C JA_j 
   \qquad j=1,2 \, , 
      \qquad
   \Xi := \Lambda^{\dagger} A_1^{\dagger} A_2 + A_1^{\dagger} A_2 \Lambda - \imag A_1^{\dagger}C JA_2 \, . 
\eez  
$\Gamma_j$ is Hermitian and commutes with $J$. $\Xi$ anti-commutes with $J$. 
Since $\Gamma_j,\Xi$ are independent of $x$, differentiation of the above equations with respect to $x$ yields
\be
   \Lambda^\dagger \, \Gamma_1 = \Gamma_1 \, \Lambda  \, , \qquad
   \Lambda^\dagger  \, \Xi = \Xi \, \Lambda  \, .   \label{2+1NLS_red_data_1}
\ee
If these relations hold, the above conditions are reduced to 
\be
     \Gamma_2 = \pm \Gamma_1 \, , \qquad 
     \Xi^{\dagger} = \pm \Xi \, .        \label{2+1NLS_red_data_2}
\ee
The reduction condition (\ref{2+1NLS_Phi_red_cond_defoc}), respectively (\ref{2+1NLS_Phi_red_cond_foc}), 
is thus equivalent to (\ref{2+1NLS_red_data_1}) and (\ref{2+1NLS_red_data_2}), with the 
respective choice of sign. We were unable to solve these algebraic equations in general. The following 
examples present some special solutions.

\begin{example}
Setting $\Gamma_j =0$, $j=1,2$,  
by inspection of the above equations, a special solution of (\ref{2+1NLS_Riem_lin_sys}) and the reduction 
conditions is given by 
\bez
    \Phi = \left( \begin{array}{cc} A \, e^{\frac{1}{2} \imag \Lambda_1 x} 
                                       & e^{\frac{1}{2} \imag \Lambda_1^\dagger x} \\
                               e^{-\frac{1}{2} \imag \Lambda_1 x} 
                                       & \pm A^\dagger \, e^{-\frac{1}{2} \imag \Lambda_1^\dagger x} 
                  \end{array} \right) \, ,
\eez
where $\Lambda_1$ (originating from $\Lambda = \mbox{block-diag}(\Lambda_1,\Lambda_1^\dagger)$)
is assumed to be normal (i.e., $[\Lambda_1,\Lambda_1^\dagger]=0$) and to commute with 
$A$ and $A^\dagger$, and 
\bez
    \Lambda_{1,t} + \Lambda_{1,y} \, \Lambda_1 = 0 \, , \qquad
    A_t + A_y \, \Lambda_1 = 0 \, .
\eez
If $m_1=m_2=1$ (scalar case), setting $A = \pm e^{\imag \, \beta -b}$, $\Lambda_1 = \lambda + \imag \, a$,
we obtain
\bez
    \varphi = \left\{ 
              \begin{array}{l@{\qquad}l}
                \left( \begin{array}{cc} \imag \, \lambda + a \, \mathrm{coth}(a \, x + b) 
                      & - a \, \mathrm{cosech}(a \, x + b) \, e^{\imag \, (\lambda \, x + \beta)} \\ 
                                                &  -\imag \, \lambda + a \, \mathrm{coth}(a \, x + b)
                  \end{array} \right)  & \mbox{defocusing case (upper sign)}  \\
               \left( \begin{array}{cc} \imag \, \lambda + a \, \mathrm{tanh}(a \, x + b) 
                      & a \, \mathrm{sech}(a \, x + b) \, e^{\imag \, (\lambda \, x + \beta)} \\ 
                                                &  -\imag \, \lambda + a \, \mathrm{tanh}(a \, x + b)
                  \end{array} \right)  & \mbox{focusing case (lower sign)} 
              \end{array} \right.   
\eez
where $a$, $b$, $\beta$ and $\lambda$ have to satisfy the system (\ref{2+1NLS_Bogo1}), (\ref{2+1NLS_Bogo2}).
The omitted left lower component of $\varphi$ is given by $r = \pm q^\dagger$ in terms of the right 
upper component. 
In the focusing case, we recover the system derived in Section~\ref{subsec:Bogo} from 
the scalar `Riemann system'. 
The above $\Phi$ thus leads, via (\ref{2+1NLS_CH_transf}), to matrix versions of 
the corresponding scalar solutions. They solve the matrix version of the (2+1)-dimensional NLS 
equation (\ref{matrix2+1NLS}). 
\end{example}

\begin{example}
Setting $\Xi = 0$, by inspection of the above equations we are led to
the following special solution of (\ref{2+1NLS_Riem_lin_sys}) and the defocusing reduction condition,  
\bez
    \Phi = \left( \begin{array}{cc} e^{\frac{1}{2} \imag \Lambda_1 x} 
                                       & e^{\frac{1}{2} \imag \Lambda_2 x} \\
                              B \, e^{-\frac{1}{2} \imag \Lambda_1 x} 
                                       & A \, e^{-\frac{1}{2} \imag \Lambda_2 x} 
                  \end{array} \right) \, ,
\eez
with Hermitian $\Lambda = \mbox{block-diag}(\Lambda_1,\Lambda_2)$, where $[\Lambda_1 , \Lambda_2]=0$, 
and $A,B$ are unitary matrices with $[A^\dagger B , \Lambda_i]=0$, $i=1,2$. 
(\ref{2+1NLS_Lambda_eq}) and (\ref{2+1NLS_A_eqs}) then result in
\bez
    \Lambda_{i,t} + \Lambda_{i,y} \, \Lambda_i = 0 \, , \qquad
    A_t + A_y \, \Lambda_2 = 0  \, , \qquad
    B_t + B_y \, \Lambda_1 = 0 \, .
\eez
In the scalar case, setting $A = \imag \, e^{- \imag \, (b + \beta)}$, $B= -\imag \, e^{ \imag \, (b - \beta)}$, 
$\Lambda_1 = \lambda - a$, $\Lambda_2 = \lambda + a$, we obtain
\bez
    \varphi =    \left( \begin{array}{cc} \imag \, \lambda - a \, \mathrm{tan}(a \, x + b) 
                      & a \, \mathrm{sec}(a \, x + b) \, e^{\imag \, (\lambda \, x + \beta)} \\ 
                                                &  -\imag \, \lambda - a \, \mathrm{tan}(a \, x + b)
                  \end{array} \right) \, ,
\eez
 from which we recover the system derived in Section~\ref{subsec:Bogo} in the defocusing case.
Again, the above $\Phi$ leads, via (\ref{2+1NLS_CH_transf}), to matrix versions of 
the corresponding scalar solutions, solving the matrix version of the (2+1)-dimensional NLS 
equation (\ref{matrix2+1NLS}).  
\end{example}

\subsection{Multi-soliton solutions of the (2+1)-dimensional NLS equation, parametrized by 
solutions of Riemann equations} 
Breaking multi-solitons of the scalar (2+1)-dimensional NLS equation have been mentioned in \cite{Bogo91III}. 
A special class of such solutions in the focusing case was obtained a few years later in 
\cite{Jian+Bull94}. Apparently, 
the authors of \cite{Jian+Bull94} were not aware of Bogoyavlenskii's related work (in particular, 
\cite{Bogo91III}). They used the (AKNS) inverse scattering method and made an ansatz for solving 
the equations for the scattering data to generate multi-soliton-type solutions. That, more generally, 
the solutions can be parametrized by solutions of a Riemann equation, is not deducible from their work.
More general solutions, moreover for the matrix (2+1)-dimensional NLS system, are immediately obtained 
via Corollary~\ref{cor:sol_via_Riem_sys}. 
As in Proposition~\ref{prop:sdYM_multi-break}, we can drop the differential equation for $\bX$ 
if we impose the spectrum condition $\sigma(\bP) \cap \sigma(\bQ) = \emptyset$. 
In the following, $\bJ$ is an $n \times n$ counterpart of $J$.
The latter is given in (\ref{J_varphi_decomp}). See \cite{DMH10NLS} for 
the way in which the bidifferential calculus (\ref{2+1NLS_bidiff}) is extended to $n \times n$, 
$m \times n$ and $n \times m$ matrices. To obtain the following result, we made redefinitions 
$\bP \mapsto \bJ \bP$ and $\bQ \mapsto \bQ \bJ$ in Corollary~\ref{cor:sol_via_Riem_sys}.

\begin{Proposition}
\label{prop:2+1NLS_multi-breaking_solitons}
Let $\varphi_0$ be a constant $m \times m$ matrix with $[J,\varphi_0]=0$. Let $\bP$ and $\bQ$ be 
solutions of the $n \times n$ matrix equations
\bez
    \imag \, \bP_t + \bP_y \, \bJ \, \bP = 0 \, , \quad
    \bP_x = \frac{1}{2} ( \bJ \, \bP \, \bJ - \bP ) \, \bP \, , \quad
    \imag \, \bQ_t + \bQ \, \bJ \, \bQ_y = 0 \, , \quad
    \bQ_x = \frac{1}{2} \bQ \, ( \bQ - \bJ \, \bQ \, \bJ ) \, .
\eez
We further assume that $\sigma(\bJ \bP) \cap \sigma(\bQ \bJ) = \emptyset$. 
Let $\bX$ be the unique solution of the Sylvester equation
\bez
    \bX \, (\bJ \bP) - (\bQ \bJ) \, \bX = \bV_0 \, \bU_0   \, ,
\eez
where $\bU_0$ and $\bV_0$ are constant matrices of size $m \times n$, respectively $n \times m$, 
and satisfy $J \, \bU_0 = \bU_0 \bJ$ and $\bJ \bV_0 = \bV_0 J$. Then 
\be
    \varphi = \varphi_0 + J \, \bU_0 \, \bX^{-1} \, \bV_0       \label{2+1NLS_multi_break_sol}      
\ee
(except at singular points) solves (\ref{pre_matrix2+1NLS}).  \hfill $\square$
\end{Proposition}

It remains to implement the reduction conditions, so that we obtain solutions of the (focusing, 
respectively defocusing) matrix
(2+1)-dimensional NLS system (\ref{matrix2+1NLS}), where $u^\dagger = u + c_1$, 
$v^\dagger = v + c_2$, with anti-Hermitian matrices $c_i$ satisfying 
(\ref{2+1NLS_reduction_conditions_for_C}).

\begin{corollary}
Let $\varphi_0$ be a constant $m \times m$ matrix, with $[J,\varphi_0]=0$ and satisfying one of the 
reduction conditions (\ref{NLS_hc_red_conditions}). 
Let $\bP$ be a solution of the $n \times n$ matrix equations
\bez
    \imag \, \bP_t + \bP_y \, \bJ \, \bP = 0 \, , \qquad
    \bP_x = \frac{1}{2} ( \bJ \, \bP \, \bJ - \bP ) \, \bP \, , 
\eez
and such that $\sigma(\bJ \bP) \cap \sigma(-\bP^\dagger \bJ) = \emptyset$. 
Let $\bX$ be the unique solution of the Lyapunov equation
\bez
    \bX \, (\bJ \bP) + (\bJ \bP)^\dagger \, \bX = \bV_0 \, \bU_0  \, ,
\eez
where $\bU_0$ is a constant $m \times n$ matrix with $J \, \bU_0 = \bU_0 \bJ$, and  
\bez
     \bV_0 = \left\{ \begin{array}{l@{\quad}l}
                          \bU_0^\dagger \, J  & \mbox{defocusing case}  \\
                          \bU_0^\dagger  & \mbox{focusing case} \, . 
                        \end{array} \right.          
\eez
Then the blocks $q,u,v$, read off from $\varphi = \varphi_0 + J \, \bU_0 \, \bX^{-1} \, \bV_0$, 
solve the matrix (2+1)-dimensional NLS system. If $m_1=m_2=1$, then $q$ solves the (2+1)-dimensional NLS
equation (\ref{2+1NLS}).
\end{corollary}
\begin{proof}
Setting $\bQ = - \bP^\dagger$ in the preceding proposition reduces the differential equations 
for $\bQ$ to those for $\bP$. The relation between $\bV_0$ and $\bU_0$ then ensures that 
$\varphi$ given by (\ref{2+1NLS_multi_break_sol}) satisfies the same reduction condition as 
$\varphi_0$. Note that $\bX^\dagger = \bX$, since the solution of the Lyapunov equation is 
unique if the stated spectrum condition holds.
\end{proof}

By choosing
\bez
    \bP = \mbox{block-diag}(\varphi_1,\ldots,\varphi_N) \, , \qquad
    \bJ = \mbox{block-diag}(\underbrace{J,\ldots,J}_{N\mbox{ times}}) \, ,
\eez
where, for $i=1,\ldots,N$, $\varphi_i$ solves (\ref{2+1NLS_Riem_1}) and (\ref{2+1NLS_Riem_2}), 
we obtain a nonlinear superposition of $N$ singular solitons. Still more general solutions can be 
generated via Theorem~\ref{thm:main}. This includes nonlinear superpositions of regular and singular 
solitons. We should stress again that the singular soliton-type solutions cannot be called 
`breaking solitons'. Although we cannot exclude presently that the (2+1)-dimensional NLS equation 
possesses hitherto unkown solutions representing breaking waves, the soliton-like solutions 
obtained here are not quite of this type and thus do not justify to call this equation a 
`breaking soliton equation', as sometimes done in the literature.

\section{Matrix Burgers and KP equations}
\label{sec:B&KP}

\subsection{Burgers equation} 
Let $\Omega^1 = \cA := \cA_0[\pa_x]$, where $\cA_0$ is the algebra of 
matrices of smooth functions of variables $x,y$, and $\pa_x = \pa/\pa x$. 
Let 
\bez
    \mathrm{d} f = f_x  \, , \qquad
    \bd f = \frac{1}{2} \, [\alpha \, \pa_y - \pa_x^2 , f ] \, ,
\eez
with a constant $\alpha$. (\ref{Riemann_1}) then reads
\bez
    \alpha \, \phi_y - \phi_{xx} - 2 \, \phi_x \, (\pa_x + \phi) = 0 \, .
\eez
A solution $\phi$ has to be an operator. Writing 
\bez
     \phi = - \pa_x + \varphi \, ,
\eez
turns (\ref{Riemann_1}) into the matrix \emph{Burgers equation}
\be
    \alpha \, \varphi_y - \varphi_{xx} - 2 \, \varphi_x \, \varphi = 0 \, ,  \label{Burgers_eq1}
\ee
and $\varphi$ can now be taken to be a matrix of functions.

\subsubsection{Cole-Hopf transformation}
Choosing $\phi_0 = -\pa_x$, (\ref{phi_transform}) becomes the \emph{Cole-Hopf transformation}
\bez
        \varphi = \Phi_x \, \Phi^{-1} \,  ,
\eez
and (\ref{CH_Phi_eq}), with $\gamma=0$, is the linear \emph{heat equation}
\be
     \alpha \, \Phi_y - \Phi_{xx} = 0 \, .    \label{heat_1}
\ee

\begin{example}
A class of solutions of (\ref{Burgers_eq1}) is determined by the following solutions 
of (\ref{heat_1}), 
\bez
    \Phi = I + \sum_{i=1}^n \bA_i \, e^{\bTh_i} \, \bB_i \, , \qquad
    \bTh_i = \bLa_i \, x + \frac{1}{\alpha} \, \bLa_i^2 \, y \, ,
\eez
where $\bA_i,\bB_i$ are constant $m \times n$, respectively $n \times m$ matrices, and $\bLa_i$ 
are constant $n \times n$ matrices. In the scalar case ($m=1$), and if 
$\bLa = \mathrm{diag}(\lambda_1,\ldots,\lambda_n)$, this takes the form
\be
    \Phi = 1 + \sum_{i=1}^n e^{\lambda_i \, x + \lambda_i^2 \, y/\alpha + \gamma_i} \, ,
           \label{Burgers_heat_sol_n-kink}
\ee
with constants $\gamma_i$. It determines an $n$-kink solution of the 
scalar Burgers equation merging (at some value of $y$) into a single kink, which then 
develops into a shock wave, see Fig.~\ref{fig:Burgers_3kink}.
\end{example}
\begin{figure}[hbtp] 
\begin{center}
\includegraphics[scale=.45]{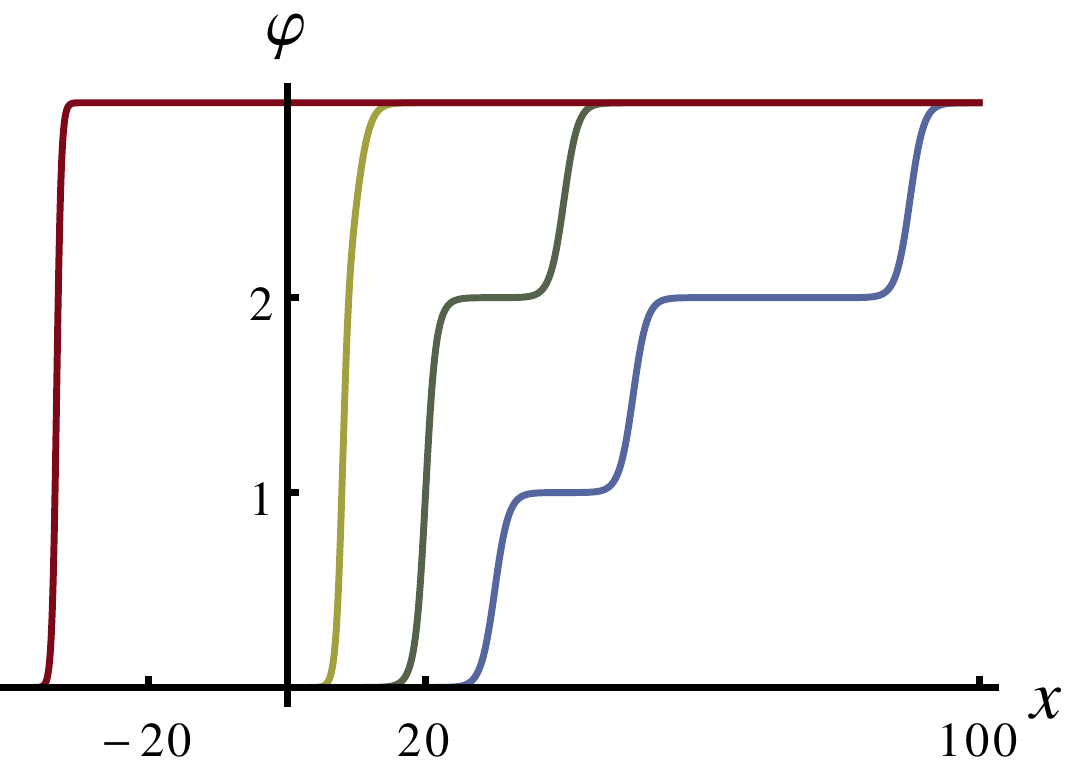} 
\hspace{1.2cm}
\includegraphics[scale=.32]{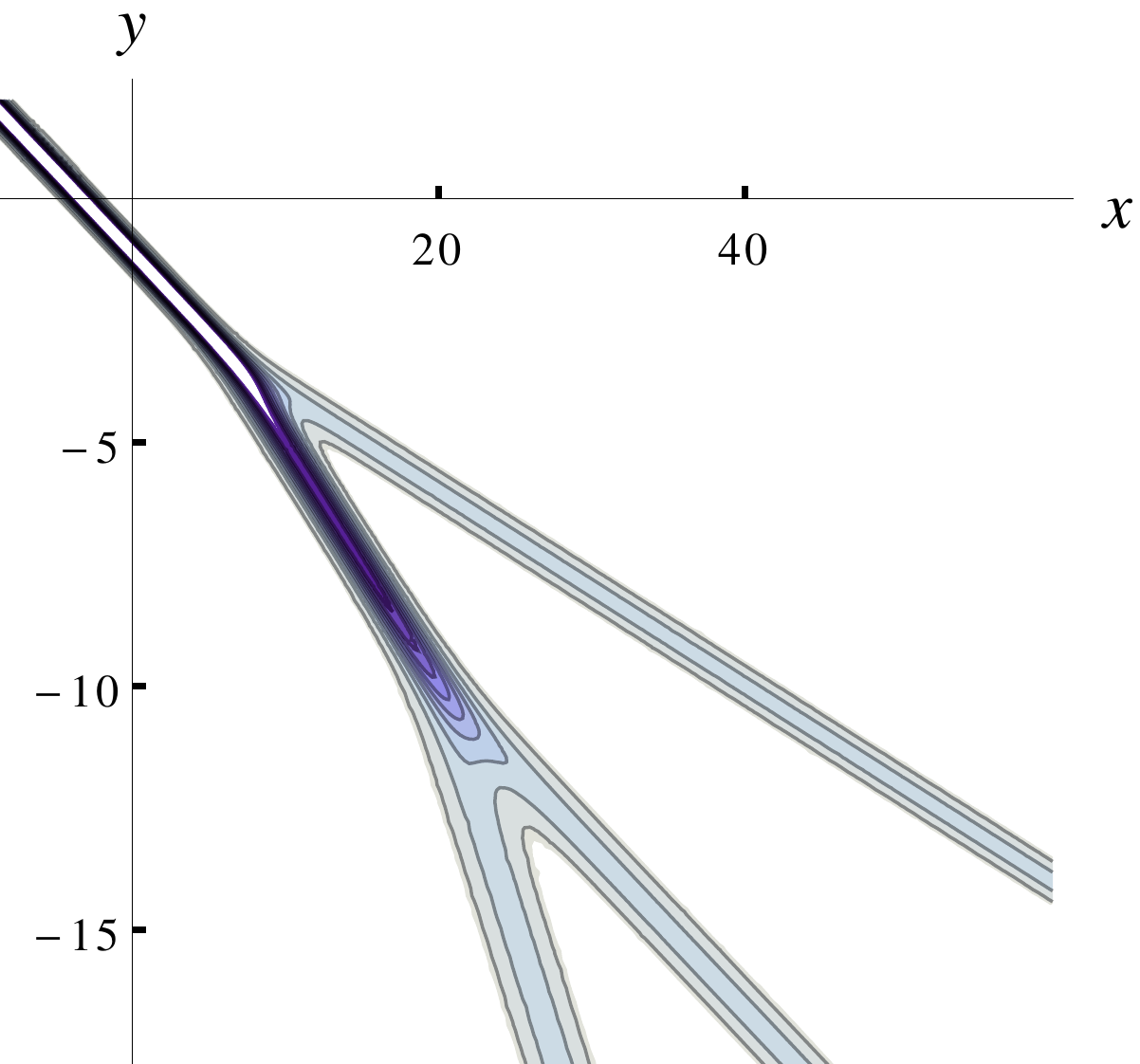} 
\caption{Plots of a solution of the scalar Burgers equation (\ref{Burgers_eq1}), with $\alpha=1$, 
at four different values ($-20,-10,-4,10$) of $y$. It is obtained 
from (\ref{Burgers_heat_sol_n-kink}) with $n=3$, $\lambda_1=1$, $\lambda_2=2$, 
$\lambda_3=3$, and $\gamma_1=-10$, $\gamma_2=0$, $\gamma_3=10$. 
Evolution in the variable $y$ is from right to left. For large negative $y$, this is a $3$-kink 
solution (right curve). It then turns into a single kink and finally  
evolves into a shock wave (left curve). The second figure shows a contour plot of $\varphi_x$, 
which has the form of a rooted binary tree.
\label{fig:Burgers_3kink}  }
\end{center}
\end{figure}

\subsubsection{Darboux transformations}
\label{sec:Burgers_Darboux}
The following is obtained from Corollary~\ref{cor:main}. 
Setting $\balpha = \bbeta =0$ and $\bP = \bQ = -\pa_x$, we have to consider the following equations,
\bez
   && \alpha \, \bU_y - \bU_{xx} = \varphi_{0,x} \, \bU \, , \qquad
    \bV_x = - \bV \, \varphi_0 \, , \qquad 
    \alpha \, \bV_y = - \bV \, \varphi_0^2 \, , 
     \\
   && \bX_x = \bV \, \bU \, , \qquad 
      \alpha \, \bX_y = \bV \bU_x - \bV_x \, \bU \, . 
\eez
The last two equations are solved by
\be
  \bX = \bC + \int \Big( \bV \, \bU \, dx + \frac{1}{\alpha} (\bV \bU_x - \bV_x \, \bU) \, dy \Big) 
            \, ,    \label{Burgers_Darboux_X}
\ee
where $\bC$ is constant and the integration is along any path from a fixed point $(x_0,y_0)$
to $(x,y)$. Then $\varphi = \varphi_0 + \bU \bX^{-1} \bV$ solves the Burgers equation 
(\ref{Burgers_eq1}).

\begin{remark}
In the scalar case (i.e., $m=1$), we have
\bez
   \varphi - \varphi_0 
   = \bU \bX^{-1} \bV = \mathrm{tr}(\bV \bU \bX^{-1})
   = \mathrm{tr}(\bX_x \bX^{-1}) 
   = (\ln \det \bX)_x \, .
\eez
Expressing the seed solution as $\varphi_0 = (\ln \Phi_0)_x$ and setting 
$\Phi = \Phi_0 \, \det \bX$, we find
\bez
     \alpha \, \Phi_y  - \Phi_{xx} 
   = (\alpha \, \Phi_{0,y}  - \Phi_{0,xx}) \, \Phi
    - 2 \, \Phi_0 \, \Big( \mathrm{tr}(\bU \bX^{-1} \bV) \, \varphi_0 
            + \mathrm{tr}(\bU \bX^{-1} \bV_x) \Big) \, ,
\eez
which vanishes by use of $\bV_x = - \bV \, \varphi_0$ and if $\Phi_0$
satisfies (\ref{heat_1}). 
Hence $\Phi$ solves (\ref{heat_1}) and we expressed the Darboux transformation 
as a Cole-Hopf transformation.
\end{remark}

\begin{example}
If the seed solution $\varphi_0$ vanishes, $\bV$ has to be constant and it only remains to 
solve $\alpha \, \bU_y - \bU_{xx} = 0$. 
In the scalar case, writing $\bU = (u_1,\ldots,u_n)$ and 
$\bV = (v_1,\ldots,v_n)^\intercal$, for vanishing seed we find  
$u_i = a_i \, e^{\lambda_i \, x + \lambda_i^2 \, y/\alpha}$, 
with constants $a_i,\lambda_i$, and $\bX_{ij} = \bC_{ij} + v_i \, u_j/\lambda_j$.
If $\bC$ is invertible, we obtain the same solution $\varphi$ with $\bC = \bI$ and 
a redefined $\bV$. Computing $\det \bX$ via Sylvester's determinant theorem and 
setting $a_i v_i/\lambda_i =: e^{\gamma_i}$, yields (\ref{Burgers_heat_sol_n-kink}).
\end{example}

\subsection{The second equation of the Burgers hierarchy} 
Now we extend $\cA_0$ to the algebra of matrices of smooth functions of variables $x,y$ and $t$, 
and we set
\bez
    \mathrm{d} f = \frac{1}{2} \, [\alpha \, \pa_y + \pa_x^2 , f ] \, , \qquad
    \bd f = \frac{1}{3} \, [ \pa_t - \pa_x^3 , f ] \, .
\eez
Writing again $\phi = -\pa_x + \varphi$, and assuming that $\varphi$ is a matrix of functions, 
(\ref{Riemann_1}) splits into (\ref{Burgers_eq1}) and 
\bez
    \frac{1}{3} \, (\varphi_t - \varphi_{xxx}) - \frac{1}{2} \, \varphi_{xx} \, \varphi - (\varphi_x)^2 
      - \frac{\alpha}{2} \, \varphi_y \, \varphi = 0 \, .
\eez
Using (\ref{Burgers_eq1}) in the last equation, converts the latter 
into the second member of a matrix Burgers hierarchy, 
\be
     \varphi_t - \varphi_{xxx} - 3 \, (\varphi_x \, \varphi)_x - 3 \, \varphi_x \, \varphi^2 = 0 \, .
          \label{Burgers_eq2}
\ee
We should stress that, though (\ref{Riemann_1}) is only a single equation for $\phi$, in terms of 
$\varphi$ it consists of two equations (which arise as the coefficients of $\pa_x^0$ and $\pa_x^1$), and 
these are equivalent to the first two equations of the Burgers hierarchy. 

Setting $\phi_0 = -\pa_x$ and $\gamma=0$, (\ref{phi_transform}) is the 
\emph{Cole-Hopf transformation} and 
(\ref{CH_Phi_eq}) is equivalent to (\ref{heat_1}) and the second member of a matrix 
heat hierarchy, $\Phi_t - \Phi_{xxx} = 0$.

\subsection{KP} 
Let us now choose $\Omega$ according to (\ref{Omega_wedge}) with $K=2$. We combine the 
above two bidifferential calculi, associated with the first two members of a matrix  
Burgers hierarchy, to
\be
    \mathrm{d} f = [\pa_x, f] \, \xi_1 + \frac{1}{2} \, [\alpha \pa_y+\pa_x^2,f] \, \xi_2 \, , \qquad
    \bd f = \frac{1}{2} \, [\alpha \pa_y - \pa_x^2,f] \, \xi_1 
            + \frac{1}{3} \, [\pa_t - \pa_x^3,f] \, \xi_2 \, .   \label{KP_bidiff}
\ee
With $\phi = -\pa_x + \varphi$, (\ref{Riemann_1}) is then equivalent to (\ref{Burgers_eq1}) 
and (\ref{Burgers_eq2}). 
Now (\ref{bidiff_conds}) holds, and the integrability condition (\ref{phi_eq}) 
turns out to be the \emph{matrix potential KP equation}
\be
  \left( - 4 \varphi_t + \varphi_{xxx} + 6 (\varphi_x)^2 \right)_x
  + 3 \alpha^2 \varphi_{yy} - 6 \alpha \, [\varphi_x,\varphi_y] = 0   \label{pKP}
\ee
(also see \cite{DMH08bidiff}). 
We recover the well-known fact that any solution of the first two members of the 
(matrix) Burgers hierarchy solves the (matrix) potential KP equation 
(see \cite{DMH09Sigma} and the references cited there).
This is simply a special case of the implication (\ref{Riemann_1}) $\Rightarrow$ (\ref{phi_eq}).

The resulting class of solutions of the scalar KP-II equation (with $\alpha = 1$) for 
the variable $u=\varphi_x$, 
corresponding to the above class of multi-kink solutions of the Burgers equation, consists of 
KP line soliton solutions that form, at any time $t$, a rooted (and generically binary) 
tree-like structure in the $xy$-plane (see Fig.~\ref{fig:Burgers_3kink} for an example). 
An analysis (in a `tropical limit') of this class of KP-II solutions  
has been carried out in \cite{DMH11KPT,DMH12KPBT}. More complicated line soliton solutions 
of the scalar KP-II equation are obtained from $n \times n$ matrix Burgers equations 
via Corollary~\ref{cor:sol_via_Riem_sys}. But according to Remark~\ref{rem:CH=>Riemann_redundant}, 
these solutions can be obtained in a simpler way from Theorem~\ref{thm:main}, with constant 
$\bP$ and $\bQ$. The resulting class of line-soliton solutions is well-known.

\begin{Proposition}
\label{prop:Burgers->KP}
Let $\varphi_0$ be a constant $m\times m$ matrix, and $\tbP$ and 
$\tbQ$ solutions of the following $n\times n$ matrix equations 
\bez
    && \alpha \tbP_y = \tbP_{xx} + 2 \, \tbP_x \, \tbP \, , \qquad \; \;
    \tbP_t = \tbP_{xxx} + 3 \, \tbP_{xx} \, \tbP + 3 \, \tbP_x \, \tbP^2 + 3 \, \tbP_x^2 \, ,  \\
    &&\alpha \tbQ_y = - \tbQ_{xx} + 2 \, \tbQ \, \tbQ_x \, , \qquad
    \tbQ_t = \tbQ_{xxx} - 3 \, \tbQ \, \tbQ_{xx} + 3 \, \tbQ^2 \tbQ_x - 3 \, \tbQ_x^2 \, ,
\eez
which are the first two members of two versions of a matrix Burgers hierarchy. 
Let $\bX$ be an invertible solution of the system of linear ordinary differential equations
\bez
    \bX_x &=& \bV_0 \, \bU_0 + \tbQ \, \bX - \bX \, \tbP \, , \\
    \alpha \bX_y &=& \tbQ \, \bV_0 \, \bU_0 + \bV_0 \, \bU_0 \, \tbP - (\tbQ_x - \tbQ^2) \, \bX 
                   - \bX \, (\tbP_x + \tbP^2) \, ,  \\
    \bX_t &=& \tbQ \, \bV_0 \, \bU_0 \, \tbP - (\tbQ_x - \tbQ^2) \, \bV_0 \, \bU_0
              + \bV_0 \, \bU_0 \, (\tbP_x + \tbP^2)  \\
          &&    + ( \tbQ_{xx} - \tbQ_x \, \tbQ - 2 \, \tbQ \, \tbQ_x + \tbQ^3) \, \bX
              - \bX \, ( \tbP_{xx} + \tbP \, \tbP_x + 2 \, \tbP_x \, \tbP + \tbP^3 ) \, ,
\eez
with any constant matrices $\bU_0$, $\bV_0$, of size $m \times n$, respectively $n\times m$.
Then $\varphi = \varphi_0 + \bU_0 \, \bX^{-1} \bV_0$ solves the $m \times m$ matrix potential KP equation
(\ref{pKP}).
\end{Proposition}
\begin{proof}
This is obtained from Corollary~\ref{cor:sol_via_Riem_sys}, using (\ref{KP_bidiff}) and 
setting $\bP = \tbP-\pa_x$, $\bQ = \tbQ - \pa_x$, $\balpha = \bbeta=0$.
\end{proof}

\begin{remark}
Proposition~\ref{prop:Burgers->KP} is more of structural interest than of practical use. 
Reducing the matrix Burgers hierarchy equations to heat hierarchy equations via Cole-Hopf transformations,
\bez
   && \tbP = \bPhi_x \, \bPhi^{-1} \, , \qquad \; \; \,  \alpha \bPhi_y = \bPhi_{xx} \, , \qquad 
    \bPhi_t = \bPhi_{xxx} \, , \\
   && \tbQ = - \bPsi^{-1} \bPsi_x \, , \qquad \alpha \bPsi_y = - \bPsi_{xx} \, , \quad
    \bPsi_t = \bPsi_{xxx} \, ,
\eez
in terms of ${\bU} := \bU_0 \bPhi$ and ${\bV} := \bPsi \bV_0$, the solution of the equations for $\bX$ 
can be expressed as
\bez
    \bX = \bPsi^{-1} \Big( \bC + \int_{(x_0,y_0,t_0)}^{(x,y_0,t_0)} \bV\bU \, dx
         + \alpha^{-1} \int_{(x,y_0,t_0)}^{(x,y,t_0)} (\bV \bU_x - \bV_x \, \bU) \, dy \\
         + \int_{(x,y,t_0)}^{(x,y,t)} (\bV_{xx} \, \bU - \bV_x \, \bU_x + \bV \bU_{xx}) \, dt \Big) \, \bPhi^{-1} \, ,
\eez
with a constant matrix $\bC$. This expression generalizes (\ref{Burgers_Darboux_X}). Now one only has to 
solve linear (matrix) heat hierarchy equations in order to construct solutions of the KP equation. There are  
quite a number of previous results about matrix solutions of the KP equation and its hierarchy, see \cite{Gils+Nimm07},
for example. 
\end{remark}

Of course, all this can be extended to the whole matrix KP hierarchy, see \cite{DMH08bidiff} 
for a corresponding bidifferential calculus.

\section{Matrix Davey-Stewartson system} 
\label{sec:DS} 

\subsection{Another `Riemann equation'}
\label{subsec:DS_another_Riem}
Let $\cB_0$ be the space of smooth complex functions on $\mathbb{R}^2$. We extend it to 
$\cB =\cB_0[\pa_x]$, where $\pa_x = \pa/\pa x$. On $\mathrm{Mat}(m,m,\cB)$ we define
\be
    \mathrm{d} f = [J,f] \, , \qquad
    \bd f = [\pa_y - J \pa_x,f] \, ,
     \label{DS_Riem1_bidiff}
\ee
where $J \neq I$ is a constant $m \times m$ matrix. Then (\ref{Riemann_1}) becomes
\be
    \phi_y = [J \pa_x,\phi] + [J,\phi] \, \phi
           = J \, \phi_x + [J,\phi] \, (\phi + \pa_x)   \, .  \label{DS_preRiemann_1}
\ee
This suggests to introduce a new dependent variable $\varphi$ via
\be
     \phi = J \, \varphi - \pa_x \, ,    \label{DS_phi->varphi}
\ee
in order to eliminate explicit operator terms. Assuming $J^2 = I$, (\ref{DS_preRiemann_1}) reads
\be
    \varphi_y = J \, \varphi_x + [J,\varphi] \, J \varphi \, , \label{DS_Riemann_1}
\ee
where $\varphi$ can now be restricted to be a matrix over $\cB_0$. 
We choose $J$ and decompose $\varphi$ as in (\ref{J_varphi_decomp}).
The last equation then splits into
\be
   u_y-u_x+2qr=0 \, ,
       \quad 
   v_y+v_x+2rq=0 \, ,
       \quad 
   q_y-q_x+2qv=0 \, ,
       \quad 
   r_y+r_x+2ru=0 \, .        \label{DS_Riem1}
\ee
A Cole-Hopf transformation for this system is obtained according to Section~\ref{sec:CH}. 
The explicit operator term that arises in (\ref{phi_transform}) via the substitution (\ref{DS_phi->varphi})
is eliminated by setting $\phi_0 = -\pa_x$. Then (\ref{phi_transform}) becomes
\be
  \varphi = J \, ( \phi + \pa_x) = - J \Phi \pa_x \Phi^{-1} + J \pa_x = J \Phi_x \Phi^{-1} \, ,
       \label{DSRiem1-CH}
\ee
Using (\ref{DS_Riem1_bidiff}), (\ref{CH_Phi_eq}) with $\gamma=0$ reads 
\be   
    \Phi_y - J \Phi_{x}=0 \, .  \label{DS_heat_eq}
\ee
All solutions of (\ref{DS_Riemann_1}) can be reached in this way. Writing
\bez
     \Phi = \left( \begin{array}{cc} f_1 & f_2 \\ g_1 & g_2 \end{array} \right) \, ,
\eez
(\ref{DS_heat_eq}) means that, for $i=1,2$, $f_i$ ($g_i$) only depends on $\tilde{x}$ ($\tilde{y}$), 
where   
\be
    \tilde{x} = x+y \, , \qquad \tilde{y} = x-y \, .  \label{DS_Riem1_tildexy}
\ee 
Using $\pa_{\tilde{x}} := \frac{1}{2} (\pa_x+\pa_y)$, $\pa_{\tilde{y}} := \frac{1}{2} (\pa_x-\pa_y)$, 
the system (\ref{DS_Riem1}) can actually be expressed in a more compact form. 
But in Section~\ref{subsec:RS_for_DS} we will supplement (\ref{DS_Riem1}) by further equations that are 
not conveniently expressed in terms of these new variables. Therefore we will not pass over to them, except 
in Section~\ref{subsec:DS_Riem1_CH_scalar}, where it achieves a substantial simplification.

Next we implement the Hermitian conjugation reductions (\ref{NLS_hc_red_conditions}) with 
$C = \mbox{block-diag}(C_1,C_2)$,  where $C_1$ and $C_2$ are constant anti-Hermitian matrices. 
This generalization of the naive reduction conditions $\varphi^{\dagger}=\varphi$, respectively 
$\varphi^{\dagger}=J \varphi J$, by introduction of the matrix $C$, is important in order to recover 
relevant solutions of the DS equation from its `Riemann system' (see Section~\ref{subsec:RS_for_DS}). 
The reduction conditions are equivalent to 
\be
     r = \mu \, q^\dagger \, , \qquad
     u = u^\dagger+C_1 \, , \qquad 
     v = v^\dagger+C_2 \, , \qquad \label{DS_hc_red_mu}
\ee
where $\mu=-1$ ($\mu=1$) corresponds to the focusing (defocusing) case in (\ref{NLS_hc_red_conditions}).
This reduces (\ref{DS_Riem1}) to 
\be    
     u_y-u_x+2\mu \, q q^{\dagger}=0 \, ,  \quad
     v_y+v_x+2\mu \, q^{\dagger}q=0  \, , \quad 
     q_y-q_x+2qv = 0 \, , \quad
     q_y+q_x+2u^\dagger q=0 \, .   \label{DS_Riem1_hc_red}
\ee
Using (\ref{DSRiem1-CH}), we obtain the following translation of the reductions conditions
(\ref{NLS_hc_red_conditions}), 
\be
 \varphi^\dagger = \varphi + C & \Longleftrightarrow & 
     \Phi^{\dagger}_xJ\Phi = \Phi^{\dagger}J\Phi_x+\Phi^{\dagger} C \Phi \nonumber \\
 \varphi^\dagger = J \, \varphi \, J + C & \Longleftrightarrow & 
     \Phi^{\dagger}_x\Phi = \Phi^{\dagger}\Phi_x + \Phi^{\dagger}CJ\Phi \, . \label{DS_Phi_red}
\ee
These are \emph{non}linear ordinary differential equations for $\Phi$, so that  
the reduction conditions are difficult to implement on the level of the Cole-Hopf transformation. 
However, at least in the scalar case ($m_1=m_2=1$ in (\ref{J_varphi_decomp})), this problem can 
be solved completely, as shown below.

\subsubsection{Cole-Hopf transformation in the scalar case}
\label{subsec:DS_Riem1_CH_scalar}
We use the coordinates (\ref{DS_Riem1_tildexy}). In the scalar case ($m_1=m_2=1$), writing 
\bez
   f_i(\tilde{x}) = F_i(\tilde{x}) \, e^{\imag c \, \tilde{x}} \, , \qquad 
   g_i(\tilde{y}) = G_i(\tilde{y}) \, e^{-\imag \tilde{c} \, \tilde{y}}  \, ,
\eez
with $c:=\frac{1}{2} \imag C_1$ and $\tilde{c}:=\frac{1}{2} \imag C_2$, 
reduces (\ref{DS_Phi_red}) to the Wronskian conditions
\be
   \left|\begin{array}{cc} 
       F_1 & F_2 \\ F_1' & F_2' \end{array} \right|
     = \alpha \, , \quad
  \left|\begin{array}{cc} G_1 & G_2 \\ 
       G_1' & G_2'  \end{array}\right|
     = \mu \, \alpha \, ,     \label{DS_F,G_det_cond}
\ee
where $\alpha \in \bbR \setminus \{0\}$, a prime denotes a derivative with respect 
to the argument, and the functions $F_i, G_i$ have to be real. 
(\ref{DSRiem1-CH}) leads to
\bez
       q = \alpha \, F^{-1} \,  e^{\imag c \, \tilde{x} + \imag \tilde{c} \, \tilde{y}}  \, , \quad
       u = (\ln F)_{\tilde{x}} + \imag \, c  \, , \quad
       v= - (\ln F)_{\tilde{y}} + \imag \, \tilde{c} \, ,
\eez       
where
\bez   
    F(\tilde{x},\tilde{y}) := F_1(\tilde{x}) \, G_2(\tilde{y}) - F_2(\tilde{x}) \, G_1(\tilde{y})  \, .
\eez

\begin{example}
\label{ex:DS_Riem_scalar}
Choosing
\bez
    f_i = a_i \, e^{\theta}+b_i \, e^{-\theta^\ast} \, , \quad 
    g_i = c_i \, e^{\tilde{\theta}} + d_i \, e^{-\tilde{\theta}^\ast} \, , \quad i=1,2  \, , \qquad
 \theta := \lambda \, \tilde{x} \, , \qquad
    \tilde{\theta} := \tilde{\lambda} \, \tilde{y}  \, , 
\eez
with real constants  $a_i,b_i,c_i,d_i$, and complex constants $\lambda$, $\tilde{\lambda}$, with  
$\mathrm{Im}(\lambda) = c$ and $\mathrm{Im}(\tilde{\lambda}) = -\tilde{c}$, the above Wronskian conditions
are solved with $\alpha = -2 \Re(\lambda) \, (a_1 b_2 - a_2 b_1)$. 
An asterisk denotes complex conjugation. From (\ref{DSRiem1-CH}) we obtain
\bez
    q &=& \frac{\alpha}{\Delta} \, e^{\imag \, \Im(\theta - \tilde{\theta}) } \, , \quad
    r = \frac{\beta}{\Delta} \, e^{\imag \, \Im(\tilde{\theta} - \theta) } \, , 
                     \nonumber \\
    u &=& \imag \, \Im(\lambda) + \Re(\lambda) \,  
           \frac{\Delta_1+\Delta_2}{\Delta} \, , \quad
    v = - \imag \, \Im(\tilde{\lambda}) 
        + \Re(\lambda) \,  \frac{\Delta_1-\Delta_2}{\Delta}   \, , 
\eez
where we set $\beta := -2 \Re(\tilde{\lambda}) \, (c_1 d_2 - c_2 d_1)$ and 
\bez
    \Delta_1 &:=& (a_1d_2-a_2d_1) \, e^{\Re(\theta-\tilde{\theta})}
               - (b_1c_2 - b_2c_1) \, e^{\Re(\tilde{\theta}-\theta)}
                 \, ,   \nonumber \\
    \Delta_2 &:=& (a_1c_2-a_2c_1) \, e^{\Re(\theta+\tilde{\theta})}
          - (b_1d_2 - b_2d_1) \, e^{- \Re(\theta + \tilde{\theta})} \, , \nonumber \\
    \Delta &:=& \Delta_1 + \Delta_2
               + 2 \, (b_1d_2 - b_2d_1) \, e^{- \Re(\theta + \tilde{\theta})}
               + 2 \, (b_1c_2 - b_2c_1) \, e^{\Re(\tilde{\theta}-\theta)} \, .  
\eez
This provides us with solutions of the reduction conditions (\ref{DS_Riem1_hc_red}) if 
$\beta = \mu \, \alpha$, i.e.,
\be
     \Re(\tilde{\lambda}) = \mu \, \frac{a_1b_2 - a_2b_1}{c_1d_2 - c_2d_1} \Re(\lambda) \, , 
         \label{DS_red_cond_abcd}
\ee
assuming $a_1 b_2 \neq a_2 b_1$ and $c_1 d_2 \neq c_2 d_1$. 
A solution from the above family is regular iff $\Delta$ is everywhere different from zero. 
This is so if one of the coefficients of the sum of exponentials is positive  
and all others greater or equal to zero.
\end{example}

\begin{remark}
In the scalar case, the system (\ref{DS_Riem1_hc_red}) has the following consequences,
\bez
   (\ln q)_{\tilde{x}\tilde{y}} = - \mu \, |q|^2 \, , \qquad
   \left( u_{\tilde{x}} + \mathrm{Re}(u^2) \right)_{\tilde{y}} = 0 \, , \qquad
   \left( v_{\tilde{x}} - \mathrm{Re}(u^2) \right)_{\tilde{y}} = 0 \, .
\eez
In terms of $w := \ln q$, the first becomes a Liouville equation. The last 
two equations can be integrated to first order equations. The solutions of (\ref{DS_Riem1_hc_red}) 
obtained above also solve these equations. 
\end{remark}

\subsection{`Riemann system' associated with the matrix Davey-Stewartson system}
\label{subsec:RS_for_DS}
Let now $\cB$ be the space of smooth complex functions on $\mathbb{R}^3$. 
We extend
(\ref{DS_Riem1_bidiff}) as follows,
\be
    \mathrm{d} f = [J,f] \, \xi_1 + [\pa_y + J \pa_x,f] \, \xi_2  \, , \qquad
    \bd f = [\pa_y-J \pa_x,f] \, \xi_1 -[\imag \, \pa_t + J \pa_x^2, \, f] \, \xi_2  \, .
    \label{DS_bidiff}
\ee
In terms of $\varphi$ given by (\ref{DS_phi->varphi}),
(\ref{Riemann_1}) is equivalent to (\ref{DS_Riemann_1}), where $\varphi$ is now allowed to also 
depend on $t$, together with
\be
  \imag \, \varphi_t = - J \varphi_{xx} - \varphi_y J \varphi
                       - J \, \varphi_x J \varphi - [J,\varphi] J \varphi_x \, .  \label{DSRiem2}
\ee
These two equations constitute the `Riemann system' for
\be     
  - {\imag} \, [ J, \varphi_t ] = \varphi_{yy} - J\varphi_{xx} J
    + J[ \varphi_x + J \varphi_y , J[J,\varphi]] - J[J,\varphi] \, [J,\varphi_x] \, ,
    \label{preDS}
\ee
which results from (\ref{phi_eq}). (\ref{DSRiem2}) decomposes into
which results from (\ref{phi_eq}). (\ref{DSRiem2}) decomposes into
\bez
&& \imag q_t=-q_{xx}+q_yv-u_yq+q_xv-u_xq+2qv_x \, , \\
&& \imag r_t=r_{xx}+v_yr-r_yu+r_xu-v_xr+2ru_x \, ,  \\
&& \imag u_t=-u_{xx}+q_yr-u_yu+q_xr-u_xu+2qr_x \, , \\
&& \imag v_t=v_{xx}+v_yv-r_yq+r_xq-v_xv+2rq_x \, .
\eez
Implementing the reduction conditions (\ref{DS_hc_red_mu}), this becomes
\be
&& \imag q_t=-q_{xx}+q_yv-u_yq+q_xv-u_xq+2qv_x \, , \nonumber \\ 
&& \imag u_t=-u_{xx}+\mu q_yq^{\dagger}-u_yu+\mu q_xq^{\dagger}-u_xu+2\mu qq^{\dagger}_x \, ,\label{DS_Riem2_hc_red}\\
&& \imag v_t=v_{xx}+v_yv-\mu q^{\dagger}_yq+\mu q^{\dagger}_xq-v_xv+2\mu q^{\dagger}q_x \, ,    \nonumber\\
&& (qv)_y+(qv)_x+u_xq-u_yq-u^{\dagger}q_y+u^{\dagger}q_x=0 \, . \nonumber
\ee
The last equation of (\ref{DS_Riem2_hc_red}) is a consequence of (\ref{DS_Riem1_hc_red}).
Using the change of variables $\tilde{u} := u_y+u_x$, $\tilde{v} := v_y - v_x$, 
and (\ref{J_varphi_decomp}), we decompose (\ref{preDS}) into 
\be
   2 \imag \, q_t\, + q_{xx} + q_{yy} + 2 \, ( q \tilde{v}\,+\,\tilde{u} q ) = 0 \, , &\qquad&
         \tilde{u}_x - \tilde{u}_y = 2 \, \left( (qr)_x + (qr)_y \right) \, , \nonumber \\
   -2 \imag \, r_t + r_{xx} + r_{yy} + 2 \, (r\tilde{u}+\tilde{v}r)=0 \, ,  
   &\qquad& \tilde{v}_x+\tilde{v}_y = 2 \, \left( (rq)_x - (rq)_y \right) \, . \label{matrixDS}
\ee
The reduction conditions (\ref{DS_hc_red_mu}) imply
\bez
    r = \mu \, q^\dagger \, ,
     \qquad
    \tilde{u} = \tilde{u}^\dagger \, , \qquad \tilde{v} = \tilde{v}^\dagger \, .
\eez
This reduces (\ref{matrixDS}) to
\be    
 &&  2 \imag \, q_t + q_{xx} + q_{yy} + 2( q \, \tilde{v} + \tilde{u} \, q ) = 0 \, , \nonumber   \\
 &&  \tilde{u}_x-\tilde{u}_y=2 \mu \left( (qq^\dagger)_x + (qq^\dagger)_y \right) \, , \qquad
     \tilde{v}_x+\tilde{v}_y = 2 \mu \left( (q^\dagger q)_x - (q^\dagger q)_y \right) \, , \label{trueDS}
\ee
which is a matrix version \cite{March88,Sakh94,Lezn+YYuzb97,DMH09Sigma,Gils+Macf09,Macf10} 
of the \emph{Davey-Stewartson} (DS) equation \cite{Dave+Stew74}. 
(\ref{DS_Riem1_hc_red}) and (\ref{DS_Riem2_hc_red}) constitute the associated `Riemann system', which is 
quite involved. 
Any solution of it is also a solution of the DS system (\ref{trueDS}).
According to Section~\ref{sec:CH}, the Cole-Hopf transformation in Section~\ref{subsec:DS_another_Riem} 
extends to the present `Riemann system', if we add the equation
\be
     \imag \, \Phi_t + J \Phi_{xx} = 0 \, .   \label{DS_heat_eq2}
\ee

\begin{example}
We extend Example~\ref{ex:DS_Riem_scalar}. Now the functions $f_i$ and $g_i$ also depend on $t$, and 
(\ref{DS_heat_eq2}) leads to the additional equations
\bez
     \imag f_{j,t}+f_{j,xx} = 0 \, , \qquad \imag g_{j,t} - g_{j,xx} = 0 \, .
\eez
In the solutions presented in Example~\ref{ex:DS_Riem_scalar}, we then simply have to make the substitutions
\bez
    \theta \mapsto  \lambda \, (x + y) + \imag \lambda^2 \, t \, ,
        \qquad
    \tilde{\theta} \mapsto 
        \tilde{\lambda} \, (x - y) - \imag \tilde{\lambda}^2 \, t \, .
\eez
In this way we recover solutions of the scalar DS system (\ref{trueDS}), with $\mu=-1$,  
in a similar form as presented in \cite{Matv+Sall91}. 
We have a single \emph{dromion} solution if $a_1c_2-a_2c_1>0$, $b_1d_2-b_2d_1 > 0$, 
$b_1c_2-b_2c_1 >0$, and $a_1d_2-a_2d_1>0$.  
This degenerates to a single \emph{solitoff} solution if $a_1c_2-a_2c_1 = 0$, which in 
turn degenerates to a single \emph{soliton} solution if 
$b_1d_2-b_2d_1=0$.
In all these cases, also (\ref{DS_red_cond_abcd}), with $\mu=-1$, has to hold.
\end{example}

\section{Concluding and further remarks}
\label{sec:conclusion}
In this work we explored realizations of the `Riemann equation' (\ref{Riemann_1}) (or (\ref{Riemann_2})) 
in bidifferential calculus. 
The most basic examples are the matrix Riemann equation, and a semi- and a full discretization of it. 
These integrable discretizations are easily obtained in bidifferential calculus from 
the continuous case, essentially by replacing in the expressions for $\mathrm{d}$ and $\bd$ a commutator 
with a partial derivative operator by a commutator with a shift operator. This works correspondingly for  
continuous, semi- and fully discrete matrix NLS equations \cite{DMH10NLS}.  

The semi-discrete Riemann equation (\ref{sd_Riemann_eq_1}) appeared in \cite{HLW99} (see (3.23) therein) 
as an `infinitesimal symmetry' (analog of infinitesimal Lie point symmetry) of the discrete Burgers 
equation (\ref{dBurgers}), in the scalar case. The (fully) discrete Riemann equation (\ref{discrete_Riemann_eq_h}) 
is a corresponding `finite symmetry' of the discrete Burgers equation.
Furthermore, it can be easily verified that the semi-discrete Riemann equation (\ref{sd_Riemann_eq_1})
and the discrete Riemann equation (\ref{discrete_Riemann_eq_v1}) are compatible. More generally, 
the semi-discrete Riemann hierarchy and the (fully) discrete Riemann hierarchy 
form a common hierarchy. This can be concluded from the fact that they share the same Cole-Hopf formula 
(\ref{sdRiem1_CH}) (cf. (\ref{dRiem_1_CH})), and the corresponding linear equations are compatible. 

Another point we concentrated on in this work is the implication (\ref{Riem_to_phi_eq}). 
Special cases are the relation between the Burgers and the KP hierarchy (Section~\ref{sec:B&KP}), as well as an  
observation made in case of the sdYM equation in \cite{Zenc08sdYM}, see Section~\ref{subsec:sdYM}. 
 From the bidifferential calculus generalization it is clear that behind this 
is in fact a common and general feature, so there are counterparts 
in case of other integrable equations. We demonstrated this for (matrix versions of) 
the two-dimensional Toda lattice, a variant of Hirota's bilinear difference equation, 
(2+1)-dimensional NLS and DS equations.

If a `Riemann system' involves a Riemann equation, perhaps via a 
reduction, we may expect to obtain a `breaking soliton' case in the sense, e.g., of 
Bogoyavlenskii's work (see \cite{Bogo90RMS,Bogo91III}, in particular). Here the sdYM equation 
is the prime example. In the case of the (2+1)-dimensional NLS equation, the resulting 
solutions are actually singular, rather than just `breaking'. 

In contrast to the Riemann equation, the integrable discrete versions 
possess infinite families of regular solutions that describe 
multi-kinks. 
The appearance of very much the same families of solutions of seemingly quite different, 
integrable equations, like discrete Riemann equations, Toda, Hirota-Miwa, Burgers and KP equations, 
is traced back to simple relations between the associated `Riemann equations'. 

Our results suggest that \emph{any} integrable equation, which is not already either C-integrable (so that 
there is a kind of Cole-Hopf transformation, cf. \cite{Calo+Eckh87,Calo+Xiao92JMP,Calo92JMP})
or integrable via a hodograph method, or perhaps a combination of both, contains a subset 
of solutions that are the solutions of a system of equations which is integrable in this more special sense. 

In the examples presented in this work, we may think of exchanging $\mathrm{d}$ and $\bd$ in (\ref{Riemann_1})
to get further integrable equations.
However, a simple computation, using the Leibniz rule, shows that
\be
    \bd \phi - (\mathrm{d} \phi) \, \phi = 0  \quad \Longleftrightarrow \quad
    \mathrm{d} \phi^{-1} - (\bd \phi^{-1}) \, \phi^{-1} = 0   \, ,  \label{d,bd_exchange_Riem_sym}
\ee
assuming that $\phi$ has an inverse. 
A corresponding statement also holds for (\ref{g_eq}):
\bez
    \mathrm{d} [ (\bd g) \, g^{-1} ] = 0  \quad \Longleftrightarrow \quad
    \bd [ (\mathrm{d} g^{-1}) \, g ] = 0  \, .
\eez
In contrast, in case of (\ref{phi_eq}) an exchange of $\mathrm{d}$ and $\bd$ can lead to an 
inequivalent equation (cf. Section~\ref{subsec:Hirota}). 
In particular, it may relate a member of a hierarchy with a member of a corresponding `negative' 
or `reciprocal' hierarchy, see \cite{DMH10AKNS}. (\ref{d,bd_exchange_Riem_sym}) shows that 
the corresponding `Riemann systems' are simply related via $\phi \mapsto \phi^{-1}$. 

The `linearization method' of Section~\ref{sec:CH}, if applicable to a `Riemann system', 
does not extend to the respective realizations of 
(\ref{phi_eq}) or (\ref{g_eq}). But the binary Darboux transformation theorem (see 
Section~\ref{sec:bDT}) applies to them and quickly leads to infinite 
(soliton-type) families of explicit solutions. This will be elaborated further in a follow-up work.
In the present work we demonstrated 
that the binary Darboux transformation method also makes sense for `Riemann systems'. 
Needless to say, the list of examples presented in this work can easily be extended. 

We should also emphasize the special case of Theorem~\ref{thm:main}, formulated in 
Corollary~\ref{cor:sol_via_Riem_sys}. If the seed solution satisfies $\mathrm{d} \phi_0 =0$, 
and if the bidifferential calculus extends to second order, this expresses a way to 
generate solutions of (\ref{phi_eq}) or (\ref{g_eq}) from solutions of $n \times n$ 
versions of the associated `Riemann system', for arbitrary $n \in \bbN$. This includes 
a construction of breaking 
multi-soliton-type solutions of the sdYM equation (see the related work in \cite{Zenc08sdYM}) 
and similar solutions of the (matrix) (2+1)-dimensional NLS equation. 

It was not our aim in this work to obtain \emph{new} integrable equations, but according to our 
knowledge the matrix version (\ref{dBurgers}) of the fully discrete Burgers equation is new,
and possibly also the integrable full discretization (\ref{discrete_Riemann_eq_v1}) of the 
Riemann equation (although this would be surprising). Moreover, this also concerns the 
generalization of Hirota's bilinear difference equation obtained in Section~\ref{subsec:Hirota}.

A systematic search for equations possessing a bidifferential calculus formulation, with the help of 
computer algebra, has not yet been undertaken. It will lead to further examples.

\vskip.2cm
\noindent
\textbf{Acknowledgments.} O.C. would like to thank the Mathematical Institute of the University 
of G\"ottingen for hospitality in November 2013 - April 2014, when part of this work 
has been carried out. Special thanks go to Dorothea Bahns. Since October 2014, O.C. has been supported 
via an Alexander von Humboldt fellowship for postdoctoral researchers.  
N.S. has been supported by the Marie Curie Actions Intra-European fellowship HYDRON 
(FP7-PEOPLE-2012-IEF, Project number 332136). 
The authors are also grateful to Aristophanes Dimakis, for sharing some of his insights,
and to Eugene Ferapontov and Maxim Pavlov for some very motivating discussions.

\makeatletter
\newcommand\appendix@section[1]{
  \refstepcounter{section}
  \orig@section*{Appendix \@Alph\c@section: #1}
  \addcontentsline{toc}{section}{Appendix \@Alph\c@section: #1}
}
\let\orig@section\section
\g@addto@macro\appendix{\let\section\appendix@section}
\makeatother

\begin{appendix}
\section{Volterra lattice equation}
\label{app:LV}
The alternative integrable semi-discretization
\be
   u_t = u^+ \, u - u \, u^-         \label{LV}
\ee
of the Riemann equation is known as the (Lotka-) Volterra lattice equation (see, e.g., \cite{Suri03}).
In the scalar case, for positive solutions, writing 
$u = a^2$, it becomes the (integrable) Kac-VanMoerbeke lattice equation 
$a_t = a \, [ (a^+)^2 - (a^-)^2 ]$.
(\ref{LV}) is a matrix version of the Volterra lattice equation. Though not as a realization of 
(\ref{Riemann_1}), a kind of potential version of it can be obtained as a realization of (\ref{phi_eq}). 
Let us consider the bidifferential calculus given by
\bez
    \mathrm{d} f = [\bbS^r , f ] \, \xi_1 + [\bbS,f] \, \xi_2 \, , \qquad
    \bd f = f_t \, \xi_1 + [\bbS^{1-r} , f ] \, \xi_2 \, ,
\eez
with some integer $r>1$. Setting $\phi = \varphi \, \bbS^{-r}$ in (\ref{phi_eq}) yields
\bez
   \varphi^{(1)}_t - \varphi_t =  
    (\varphi^{(r)}-\varphi)(\varphi^{(1)}-\varphi -I)
   -(\varphi^{(1)}-\varphi -I)(\varphi^{(r)}-\varphi)^{(1-r)} \, ,
\eez
where $\varphi^{(r)} = \bbS^r \varphi \bbS^{-r}$. In terms of $u := \varphi^{(1)} - \varphi - I$, 
it takes the form
\bez
    u_t =  \sum_{i=1}^{r-1} u^{(i)} \, u - u \sum_{i=1-r}^{-1} u^{(i)} \, ,
\eez
which for $r=2$ is (\ref{LV}) (cf. \cite{DMH08bidiff}). 
The associated `Riemann system', obtained from (\ref{Riemann_1}), is
\bez
    \varphi_t = (\varphi^{(r)} - \varphi) \, \varphi \, , \qquad
    \varphi^{(r)} = I + \varphi^{(r-1)} \, (I - \varphi^{-1}) 
    \, .
\eez
The first is the semi-discrete Riemann equation (\ref{subsec:sdRiem})
(where $\bbS$ is replaced by $\bbS^r$), the second a recurrence relation. 

\end{appendix}

\end{document}